\documentclass[10pt]{article}
\date{}
\author{
Ali Tajer\thanks{ECSE Department, Rensselaer Polytechnic Institute, Troy, NY 12180.} \and Javad Heydari\thanks{Advanced AI Lab, LG Electronics USA, Santa Clara, CA 95054.} \and H. Vincent Poor\thanks{EE Department, Princeton University, Princeton, NJ  08540.}}

\usepackage{times}
\usepackage{mathrsfs}  
\usepackage{amsmath,amssymb}
\usepackage{graphicx,subfigure}
\usepackage[noadjust]{cite}
\usepackage{hyperref}
\usepackage{flushend}
\usepackage{bm}
\usepackage{dsfont}
\usepackage {tikz}
\usepackage{epstopdf}
\usepackage{amsthm}
\usepackage{appendix}
\usepackage{comment}

\ifpdf
    \graphicspath{{figures/PNG/}{figures/PDF/}{figures/}}
\else
    \graphicspath{{figures/EPS/}{figures/}}
\fi

\newtheorem{theorem}{Theorem}
\newtheorem{lemma}{Lemma}
\newtheorem{definition}{Definition}
\newtheorem{corollary}{Corollary}

\newtheorem{example}{Example}

\def \F {{\cal F}}
\def \cT {{\cal T}}
\def \cS {{\cal S}}

\def \X {{\cal X}}
\def \M {{\cal M}}
\def \A{{\cal A}}

\def \G{{\cal G}}
\def \L{{\cal L}}
\def \N{{\cal N}}
\def \cP{{\cal P}}
\def \V{{V}}
\def \E{{E}}
\def\cH{{\cal H}}
\def\e{{\rm e}}

\def \sH{{\sf H}}
\def \sT{{\sf T}}

\def \H{{\sf H}}
\def \bSigma{{\mathbf \Sigma}}

\newcommand\independent{\protect\mathpalette{\protect\independenT}{\perp}}
\def\independenT#1#2{\mathrel{\rlap{$#1#2$}\mkern2mu{#1#2}}}
\newcommand{\dff}{\stackrel{\scriptscriptstyle\triangle}{=}}
\newcommand{\bbe}{\mathbb{E}}
\def \P{\mathbb{P}}
\def \med{\;|\;}

\def \d{{\delta}}

\def \der{{\rm d}}

\DeclareMathOperator*{\argmax}{arg\,max}

\newcommand\Prob{{{\cal P}({\alpha},{\beta})}}

\begin{document}

\title{\Large \bf Active Sampling for the Quickest Detection of Markov Networks}

\maketitle

\allowdisplaybreaks

\begin{abstract}
\textcolor{black}{Consider $n$ random variables forming a Markov random field (MRF). The true model of the MRF is unknown, and it is assumed to belong to a binary set. The objective is to sequentially sample the random variables (one-at-a-time) such that the true MRF model can be detected with the fewest number of samples, while in parallel, the decision reliability is controlled. The core element of an optimal decision process is a rule for selecting and sampling the random variables over time. Such a process, at every time instant and adaptively to the collected data, selects the random variable that is expected to be most informative about the model, rendering an overall minimized number of samples required for reaching a reliable decision. The existing studies on detecting MRF structures generally sample the entire network at the same time and focus on designing optimal detection rules without regard to the data-acquisition process. This paper characterizes the sampling process for general MRFs, which, in conjunction with the sequential probability ratio test, is shown to be optimal in the asymptote of large $n$. The critical insight in designing the sampling process is devising an information measure that captures the decisions' inherent statistical dependence over time. Furthermore, when the MRFs can be modeled by acyclic probabilistic graphical models, the sampling rule is shown to take a computationally simple form. Performance analysis for the general case is provided, and the results are interpreted in several special cases: Gaussian MRFs, non-asymptotic regimes, connection to Chernoff's rule to controlled (active) sensing,  and the problem of cluster detection.}


\end{abstract}



\section{Introduction}

{\renewcommand{\thefootnote}{}\footnotetext{\noindent This paper was presented in part at the 2015 Annual Allerton Conference on Communication, Control, and Computing.}}


\subsection{Overview}
\textcolor{black}{Driven by advances in information sensing and acquisition, many application domains have evolved towards interconnected networks of information sources in which large-scale and complex data is constantly generated and processed for various inferential and decision-making purposes. Induced by their physical couplings, such information sources generate data streams that often bear strong statistical dependence structures. Probabilistic graphical models, in general, and Markov random fields (MRFs), in particular, provide effective analytical frameworks for encoding the statistical relationship among the datasets generated by different agents in a network~\cite{Ising,PGMbook,Texture,SAR}. 
}

\textcolor{black}{
Forming inferential decisions in an MRF  strongly hinges on determining the dependence structure embedded in the MRF. There are two distinct aspects to determining an MRF structure: selecting (estimating) versus differentiating (detecting) the models. In {\bf model selection} (structure learning), the objective is to sample the random variables that form an MRF, and select (estimate) the edge set of the graphical model associated with the MRF~(a representative list includes~\cite{friedman:08,WR1, WR2, WR3,cai:11,cai:11b,Anandkumar2012,danaher:14,cohn1996active,tong:01,tong:01b,he2008active,vats:14,dasarathy:16}). While the problem of graph structure learning is NP-hard in its general form~\cite{David},  it becomes feasible under proper restrictions on the structure of the graph, e.g., limiting the graph to the classes of sparsely-connected graphs, edge-bounded graphs, and degree-bounded graphs. There is a rich literature investigating the algorithmic and information-theoretic aspects of structure learning, especially for Gaussian and Ising graphical models. The existing studies can be distinguished based on the sampling mechanisms that they adopt. Broadly, there exists two distinct approaches to sampling: (i) pre-specific sampling, in which sampling is agnostic to the data and follows pre-specified rules~\cite{friedman:08,WR1, WR2, WR3,cai:11,cai:11b,Anandkumar2012,danaher:14}, and (ii) active sampling, in which the sampling decisions are data-driven and they are updated dynamically as the data is collected~\cite{cohn1996active,tong:01,tong:01b,he2008active,vats:14,dasarathy:16}. In active sampling methods, sampling and model selection processes are inherently coupled, and the emphasis is on co-designing these two processes. In contrast, when the sampling mechanism is pre-specified, the sampling and model selection processes are decoupled, and the emphasis is placed on forming reliable decisions given a set of samples. 
}

\textcolor{black}{
In contrast to model selection, in {\bf model detection}, the unknown model of an MRF is assumed to belong to a finite set of known models, and the objective is to sample the random variables in order to identify the true model. MRF model detection, in its simplest form, is used for deciding whether a given set of random variables are independent, which is referred to as testing against independence. More generally, dependence model detection is the process of deciding in favor of one dependence model against a group of alternative ones (a representative list of relevant literature includes~\cite{Ku:SP06,Anandkumar:IT09,Solo:ICASSP10,AriasCastro2012,Liang:nonparametric,Arias:2015,Berthet:2013,Hero:2012, xia:17}). The existing studies on MRF model detection adopt pre-specified sampling mechanisms and focus on forming detection rules. Furthermore, existing studies primarily investigate Gaussian MRFs.}

\textcolor{black}{
In this paper, we investigate  {\sl active} sampling for model detection in general MRFs. The objectives are (i) establishing the fundamentally minimum number of samples required for forming decisions with target reliability, and  (ii) characterizing the attendant sampling and detection rules. Characterizing an optimal active sampling algorithm that can detect the model of an MRF with the minimal number of samples is especially imperative as MRF's size or dimension grow, in which case sampling incurs substantial communication, sensing, and decision delay costs.  An active sampling process in an MRF is specified by the aggregate number of samples to be collected as well as the order in which they are collected.} When the order is pre-specified, determining the optimal sampling strategy reduces to minimizing the (average) number of samples. This can be effectively facilitated via sequential hypothesis testing, which is well-investigated. In sequential hypothesis testing, the samples are collected sequentially according to a {\sl pre-specified} order, and the sampling strategy dynamically decides whether to take more samples or to terminate the process and form a decision~\cite{Wald1945,PoorQuick,MhypI,Tartakovski98}. However, incorporating dynamic decisions about the order of sampling introduces a new dimensiont to decision-making, which is less-investigated. Forming such dynamic decisions that pertain to data acquisition naturally arises in a broad range of applications such as sensor management~\cite{sensor}, inspection, and classification~\cite{Classification}, medical diagnosis~\cite{Medical}, cognitive science~\cite{Shenoy}, generalized binary search~\cite{Nowak:IT11}, and channel coding with feedback~\cite{Burnashev}, to name a few.


\newpage



\subsection{Related Literature}

{\bf Controlled (active) sensing for detection.} One directly applicable approach to treat such coupled sampling and decision-making process is {\sl controlled sensing}, originally developed by Chernoff for binary composite hypothesis testing through incorporating a controlled information gathering process that dynamically decides about taking one of a finite number of possible actions at each time~\cite{Chernoff1959}. Under the assumption of uniformly distinguishable hypotheses and having independent control actions, Chernoff's rule decides in favor of the action with the best {\sl immediate} return according to proper information measures and achieves optimal performance in the asymptote of a diminishing rate of erroneous decisions. Chernoff's rule, specifically, at each time, identifies the most likely true hypothesis based on the collected data and takes the action that reinforces the decision. 

Extensions of the Chernoff's rule to various settings are studied in~\cite{Bessler,Albert1961,Box,Meeter}. Specifically, studies in~\cite{Bessler} and~\cite{Albert1961} investigate the extension of Chernoff's rule to accommodate an infinite number of available actions and an infinite number of hypotheses, and~\cite{Box} and~\cite{Meeter} provide alternative rules that are empirically shown to outperform the Chernoff rule in the non-asymptotic regimes. Recent advances in controlled sensing that are relevant to the scope of this paper include~\cite{Atia, Nitinawarat:SA15,Zhao:IT15,switchcost}. In~\cite{Atia}, Chernoff's rule is modified to relax the assumption that the hypotheses should be uniformly distinguishable in the multi-hypothesis setting. In this modified rule, a randomization policy is introduced into the selection rule such that at certain time instants it ignores the Chernoff rule and randomly selects one action according to a uniform distribution. This rule is shown to admit the same asymptotic performance as the Chernoff rule. The results are extended to the setting in which the available data belongs to a discrete alphabet and follows a stationary Markov model~\cite{Nitinawarat:SA15}. An application of the Chernoff rule to anomaly detection in a dataset is investigated in~\cite{Zhao:IT15}, where it is shown that when facing a finite number of sequences consisting of an anomalous one, the Chernoff rule is asymptotically optimal even without assuming that the hypotheses are distinguishable, or exerting randomized actions. The study in~\cite{switchcost} imposes a cost on switching among different actions and offers a modification of the Chernoff rule, which randomly decides between repeating the previous action, and a new action based on Chernoff's rule. It achieves the same asymptotic optimality property as the Chernoff rule. Similarly, the Chernoff rule is also applied to sparse signal recovery~\cite{sparse}, sequential estimation~\cite{Atia:GlobSIP13}, and classification problems~\cite{Kalman:class,Ligo} often resulting in considerable performance gains.

Besides  Chernoff's rule and its variations, there exist alternative strategies admitting certain optimality guarantees. In pioneering studies,~\cite{Schwarz} and~\cite{Kiefer1963} offer a strategy that initially takes a number of samples according to a pre-designated rule in order to identify the true hypothesis, after which it selects the action that maximizes the information under the identified hypothesis. The study in~\cite{Lalley1986} proposes a heuristic strategy and characterizes the deviation of its average delay from the optimal rule. More studies have investigated the Bayesian setting~\cite{Naghshvar2013,Tara,Wang:Arxiv15,Zhao:SP14,Zhao:SP15}. The study in~\cite{Naghshvar2013} considers a sequential multi-hypothesis testing problem with multiple control actions for which the optimal strategy is the solution to dynamic programming that is computationally intractable. Hence, it designs two heuristic policies and investigates their non-asymptotic and asymptotic performances. For the same problem, performance bounds and the gains of sequential sampling and optimal data-adaptive selection rules are analyzed in the asymptote of the large cost of erroneous decisions~\cite{Tara}. The study in~\cite{Wang:Arxiv15} restricts the smaples to be generated by the exponential family distributions and shows that the dimension of the sufficient statistic space is less than both the number of parameters governing the exponential family and the number of hypotheses. Hence, the exactly optimal policy can be characterized by only moderate computational complexity. Other heuristic approaches for anomaly detection are also investigated in~\cite{Zhao:SP14} and~\cite{Zhao:SP15}, which select the action with the minimum immediate effect on the total Bayesian cost and are shown to achieve the same optimality guarantees suggested by Chernoff~\cite{Chernoff1959}.

Despite their discrepancies in settings and approaches, all the studies above on controlled sensing assume that the available actions are {\sl independent} or follow a first-order stationary Markov process. This is in contrast to the setting of this paper, in which the correlation structure in the generated data under one hypothesis or both induces co-dependence among the control actions. In this paper, we devise a sequential sampling strategy for detecting Markov networks, in which the correlation model plays a significant role in forming the sampling decisions. Specifically, the devised selection rule, unlike the Chernoff rule, incorporates the correlation structure into the decision-making via accounting for the impact of each action on the future ones and selecting the one with the largest {\sl expected} information under the most likely true hypothesis. The associated optimality guarantees are established, and the specific results for the special case of Gaussian distributions are characterized. The gains of the proposed selection rule are also delineated analytically and by numerical evaluations.\vspace{.05 in}


\noindent \textcolor{black}{{\bf Active learning for model selection.} Unlike for model detection, active sampling for model selection (structure learning) is investigated in more depth~\cite{cohn1996active,vats:14,dasarathy:16,tong:01,tong:01b,he2008active}. In~\cite{cohn1996active}, a model selection problem in a supervised setting is considered, to which active learning is applied in order to identify the set of training examples that should be used to minimize the integrated variance of the model. The studies in~\cite{vats:14} and~\cite{dasarathy:16} propose active learning algorithms for selecting the structures of MRFs. Their main distinction from our work is that they are concerned with a structure learning problem, while in this paper, the true model is selected from a finite set of candidate models. In~\cite{tong:01,tong:01b,he2008active}, active learning over Bayesian networks. In~\cite{tong:01}, it is assumed that the graphical model underlying the Bayesian network is known, and the objective is estimating network parameters. The studies in~\cite{tong:01b,he2008active,murphy:01} are concerned with learning the connectivity structure, the parameters, and the direction of the causal relationship among the nodes.}

\textcolor{black}{In another related direction, model selection is performed via estimating the covariance matrix (or its inverse) of the data~\cite{xia:17,cai:11,cai:11b,danaher:14,friedman:08}. In~\cite{xia:17}, the temporal correlation of the data is treated as a nuisance parameter, and by making a Gaussian assumption, a sufficient test statistic, as well as a test procedure, is proposed to identify all the non-zero elements of the precision matrix with guaranteed performance. Estimation of sparse covariance matrices via adaptive thresholding is considered in~\cite{cai:11}. By adapting the threshold to the variability of individual entries in a data-driven setting, it is shown that, compared to the commonly used universal thresholding estimators, these estimators achieve the optimal rate of convergence over a large class of sparse covariance matrices under the spectral norm and enjoy excellent performance both theoretically and numerically. In~\cite{cai:11b,danaher:14,friedman:08} estimation of sparse precision matrices is considered, and $\ell_1$ minimization is used to solve the problem by using the estimated covariance matrix of the given data. All the studies above consider the problem in the fixed sample-size setting, and their application to model selection is limited to Gaussian distributions.
}


\section{Data Model and Problem Formulation}
\label{sec:problem_formulation}

\subsection{Notation}

 \textcolor{black}{
Throughout the paper $(\Omega,\F, \mathbb{P})$ is a probability space on which all the probability measures are defined. In this space, consider $n$ random variables $\X\dff\{X_1,\dots,X_n\}$ forming a Markov random field (MRF) with respect to an undirected graph $\G(\V,\E)$ with nodes $\V\dff\{1,\dots,n\}$ and the edge set $\E\subseteq\V\times\V$. Throughout the paper, for any given set $A\subseteq \{1,\dots,n\}$ we define $X_A\dff \{X_i\;:\; i\in A\}$. Random variables $\X$ satisfy the global Markov property, that is  any two disjoint subsets of random variables are conditionally independent given a separating set, i.e., 
\begin{align}
X_A \independent X_B \mid X_C \ ,
\end{align}
where $C$ separates disjoint $A$ and $B$ such that every path between a node in $A$ and a node in $B$ passes through at least one node in $C$. One immediate result of the global Markov property is the pair-wise Markov property, i.e., 
\begin{align}
\forall \;\; (i,j)\notin\E \quad \Leftrightarrow \quad X_i \independent X_j \mid X_{\V\setminus\{i,j\}}\ .
\end{align}
The model of the underlying $\X$ is unknown, and it is assumed to obey one of the two possible known models. Detecting the MRF model can be formalized as the solution to the binary hypothesis test:
\begin{equation}\label{eq:hyp:model}
\H_0\ :\; (X_1,\dots, X_n)\sim \P_0 \qquad \mbox{versus} \qquad \H_1\ :\;  (X_1,\dots, X_n)\sim \P_1\ ,
\end{equation}
where $P_0$ and $P_1$ denote the two known and completely distinct probability measures governing the two models. We denote the  undirected dependency graphs associated with joint measures $P_0$ and $P_1$ by $\G_0(\V,\E_0)$ and $\G_1(\V,\E_1)$, respectively. Figure~\ref{fig:model} depicts the graphs associated with the precision matrices of a a dichotomous Gaussian MRF model, in which the edges encode the conditional dependency structures.} For convenience in notations, we assume that the distributions of the random variables under each hypothesis $\ell\in\{0,1\}$ are absolutely continuous with respect to a common distribution and have well-defined probability density functions (pdfs). For every non-empty set $A\subseteq\V$, we denote the joint pdf of ${X}_A$ under $\H_\ell$ by $f_\ell(\cdot;A)$. We also define $\sT\in\{\H_0,\H_1\}$ as the true hypothesis and denote the prior probability that hypothesis $\H_\ell$ is true by $\epsilon_\ell$, where $\epsilon_0+\epsilon_1=1$.

\begin{figure}[t]
\centering
\includegraphics[width=5in]{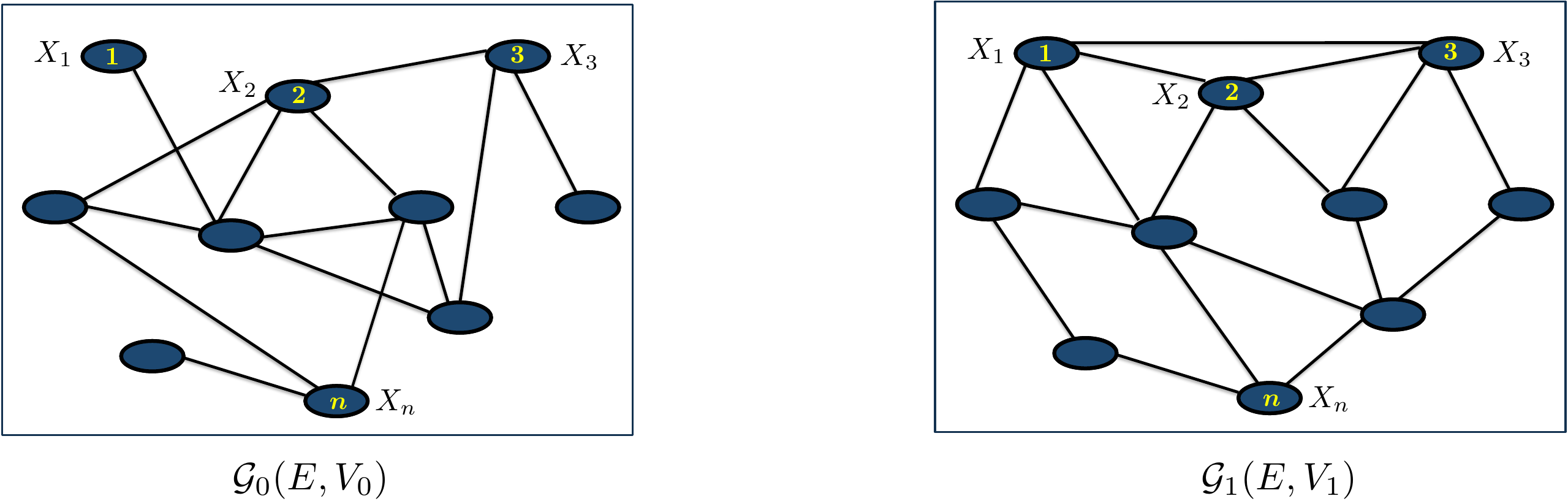}
\renewcommand{\figurename}{Fig.}
\caption{Data model with two different correlation structures.}
\label{fig:model}
\end{figure}

\subsection{Sampling Model}

We consider a fully sequential data acquisition mechanism, in which we select and sample one node at-a-time. The objective is to identify an optimal sequence of nodes, such that with the minimum number of samples, the true model $\sT\in\{\H_0,\H_1\}$ can be discerned. Samples are collected sequentially, such that at any time $t$ and based on the information accumulated up to that time, the sampling procedure takes one of the following actions.
\begin{enumerate}
 \item[A$_1$)] {\sl Exploration:} Due to lack of sufficient confidence, making any decision is deferred, and one more sample is taken from another node in the graph. Under this action, the node to be selected should be specified.
 \item[A$_2$)]{\sl Detection:} Data collection process is terminated, and a reliable decision about the true model of the graph is formed. Under this action, the stopping time and the final decision rule upon stopping will be specified.
\end{enumerate}
The sampling process can be expressed uniquely by the data-adaptive rule for selecting the nodes over time, the stopping rule, the final detection decision rule. \textcolor{black}{To formalize the information-gathering process (exploration), we define  $\psi_n: \{1,\dots,n\}\rightarrow \V$, where $\psi_n(t)$ returns the index of the node observed at time $t$.} Accordingly, we define $\psi_n^t$ as the {\sl ordered} set of nodes selected and sampled up to time $t$, i.e., $\psi_n^t\dff \{\psi_n(1),\dots,\psi_n(t)\}$. 
We also define $\varphi_n^t$ as the set of nodes that are remained unobserved prior to time $t$ and can be observed at $t$, i.e., $\varphi_n^t\dff \V\backslash \psi_n^{t-1}$. We denote sample collected at time $t$ by $Y_t\dff X_{\psi_n(t)}$, and denote the sequence of samples accumulated up to time~$t$ by ${Y}^t\dff\big(Y_1,\dots,Y_t\big)$. The information accumulated sequentially generate a $\sigma$-algebra of $\F$ denoted by  $\{\F_t:\;t=1,2,\dots\}$, where
\begin{equation}\label{eq:F}
\F_t\dff\sigma(Y^t;\psi_n^t)\ .
\end{equation}
We define $\tau_n\in\mathbb{N}$ as the Markov stopping time with respect to the family $\{\F_t\}$, at which the sampling process is terminated and a decision is formed. We also define $\delta_n\in\{0,1\}$ as an $\F_t$-measurable function as the terminal decision rule, where $\delta_n=\ell$ indicates accepting hypothesis $\H_\ell$, for $\ell\in\{0,1\}$. 
Based on the decision rules defined above, we define the tuple $\Phi_n\dff(\tau_n,\delta_n,\psi_n^{\tau_n})$ to uniquely specify the sampling strategy\footnote{We remark that the subscript $n$ is included in all decision rules to signify the effect of network size.} and the decision rules involved. Finally, we define two information measures that are instrumental in formalizing and analyzing various decision rules throughout the paper. Specifically, for any given $\psi^t_n$ and $A\subseteq\{1,\dots, n\}\backslash\psi^t_n$ we define
\begin{align}\label{eq:J0}
\mathscr{J}_0 (A,\psi^t_n) & \dff D_{\rm KL}\big(f_0(X_A; A \med \F_t)\; \|\; f_1(X_A; A \med \F_t)\big)\ ,\\
\label{eq:J1}
\mbox{and}\quad {\mathscr J}_1 (A,\psi^t_n) & \dff D_{\rm KL}\big(f_1(X_A; A \med \F_t)\; \|\; f_0(X_A; A \med \F_t)\big)\ .
\end{align}
$ D_{\rm KL}(f\;\|\;g)$ denotes the Kullback-Leibler (KL) divergence  from a statistical model with pdf $g$ to a model with pdf $f$. 
 
\subsection{Problem Statement}

The coupled information-gathering strategy and decision-making processes are uniquely specified by the triplet $\Phi_n=(\tau_n,\delta_n,\psi_n^{\tau_n})$. Designing the optimal sampling strategy for achieving the quickest reliable decision involves resolving the tension between the {\sl quality} and {\sl agility} of the process as two opposing measures (improving one penalizes the other one). The agility of the process is captured by the average delay in reaching a decision, i.e., $\mathbb{E}\{\tau_n\}$, and the decision quality is captured by the frequency of erroneous decisions denoted by
\begin{align}
\mathsf{P}_n^{0} & \;\dff\; \mathbb{P}_0(\d_n=1)\ , \qquad \mbox{and} \qquad  \mathsf{P}_n^{1}  \;\dff\; \mathbb{P}_1(\d_n=0)\ .
\end{align}
To formalize the quickest reliable decision, we control the quality of the decision and minimize the average number of samples over all possible combinations of $\Phi_n=(\tau_n,\delta_n,\psi_n^{\tau_n})$. Therefore, an optimal sampling strategy of interest is a solution to
\begin{equation}
\cP({\alpha},{\beta})\dff \left\{
\begin{aligned}
\label{eq:Opt}
&\textstyle\inf_{\Phi_n}  &&\mathbb{E}\{\tau_n\} \; \\
&\;\mbox{s.t. } && \mathsf{P}_n^{0}\leq {\e}^{-n{\alpha}} \ \\ 
& && \mathsf{P}_n^{1}\leq {\e}^{-n{\beta}}\ \
\end{aligned} \ ,\right. 
\end{equation}
where ${\alpha},{\beta} \in\mathbb{R}_+$ control the error probability terms $\mathsf{P}_n^{0}$ and $\mathsf{P}_n^{1}$, respectively, and are selected such that the problem $\Prob$ is feasible.

\section{Network-guided Active Sampling}
\label{sec:opt}

\textcolor{black}{The core element in characterizing the decision tuple $\Phi$ is the data-adaptive and sequential sampling process $\psi_n(t)$. The structure of this process is strongly shaped by the two MRFs specified under $\H_0$ and $\H_1$. In this section, we characterize a data-adaptive and sequential sampling process and show that this process, in conjunction with a thresholding rule for the stopping time and a likelihood ratio detection rule, constitutes an optimal solution to \eqref{eq:Opt}. Optimality properties, performance analysis, and complexity analysis are provided in Section~\ref{sec:per}.
}

\subsection{Terminal Decision Rules}
\textcolor{black}{Before providing the details of the core process (node selection rule), we briefly discuss the terminal decision rules.} For this purpose, define\footnote{For simplicity in notations, throughout the rest of the paper, we omit the subscript $n$ in terms $\psi_n^t$, $\psi_n(t)$, and $\varphi_n^t$.}
\textcolor{black}{
\begin{equation}\label{LLR}
\Lambda_t \;\dff\; \ln\frac{f_1(Y^t; \psi^t)}{f_0(Y^t; \psi^t)}\ ,
\end{equation}
}
as the log-likelihood ratio (LLR) of the samples collected up to time $t$. It can be readily verified that
\begin{equation}\label{LLR2}
\Lambda_{t+1} \;=\; \Lambda_t+\ln\frac{f_1(Y_{t+1};\psi(t+1)\,|\,\F_t)}{f_0(Y_{t+1};\psi(t+1)\,|\,\F_t)}\ .
\end{equation}
\noindent {\bf Stopping rule:} To specify the stopping rule of the sampling process, we define 
\begin{align}\label{eq:th1}
\gamma_n^{\rm L}  & \;\dff\; -n{\beta} \ ,\quad \text{and }\quad\gamma_n^{\rm U}   \;\dff\; n{\alpha} \ ,
\end{align}
and specify the stopping time through the following sequential likelihood ratio test:
\begin{equation} \label{eq:stop}
\tau^*_n\;\dff\;\inf\enspace\big\{t\;:\;\Lambda_t\notin (\gamma_n^{\rm L} ,\gamma_n^{\rm U} )\;\; \text{or}\;\;  t=n\big\}\ .
\end{equation}
\textcolor{black}{This is a {\sl truncated} sequential probability ratio test (SPRT), in which the delay is bounded by the total number of samples possibly available. If we drop the condition $t=n$, the stopping rule simplifies to that of the canonical SPRT. We note that dependin on the context, there exist other variations of truncated SPRT as well~\cite{Siegmund:Book1985}.}

\noindent {\bf Detection rule:} 
At the stopping time, we decide on the model according to 
\begin{equation}\label{eq:SPRT}
\d^*_n\dff \left\{
\begin{array}{ll}
\vspace{1.5mm}
0\ , & \quad \mbox{if}\ \Lambda_{\tau^*_n}< 0 \\
1\ , & \quad \mbox{if}\ \Lambda_{\tau^*_n} \geq 0 
\end{array}\right.\ .
\end{equation}
Based on~\eqref{eq:stop} and~\eqref{eq:SPRT}, the sampling process resumes as long as $\Lambda_t\in(\gamma_n^{\rm L} ,\gamma_n^{\rm U} )$ and terminates as $\Lambda_t$ falls outside this band or we exhaust all the samples, i.e., $t=n$. Furthermore, if $\Lambda_t$ exits this interval from the upper threshold~$\gamma_n^{\rm U} $ the set $\{X_1,\dots,X_n\}$ is deemed to form a Markov network with model $\P_1$, and if it falls below the lower threshold $\gamma_n^{\rm L} $ we make a decision in favor of $\P_0$. \textcolor{black}{
We remark this decision rule is different from that of the SPRT. Specifically, our thresholds are constants, while those of the SPRT are controlled by the target error probabilities, i.e., 
\begin{equation}\label{eq:SPRT_SPRT}
\delta^*_{\rm SPRT}=\left\{\begin{array}{ll}
\vspace{1.5mm}
0\ , & \quad \text{if}\ \Lambda_{\tau_{\rm SPRT}}< \gamma^{\rm L}_n\\
1\ , & \quad \text{if}\ \Lambda_{\tau_{\rm SPRT}}\geq \gamma^{\rm U}_n 
\end{array}\right.\ .
\end{equation}
}

\subsection{Dynamic Sampling}
\label{sec:selection}

 \textcolor{black}{At any time $t\in\{1,\dots ,\tau_n\}$, prior to the stopping time, based on the information accumulated up to time $(t-1)$ the sampling process dynamically identifies and takes a sample from one unobserved node that is expected to provide the most relevant information about the true hypothesis. In this subsection, we provide two approaches to dynamic node selection. First, we provide the design of the selection rule based on Chernoff's principle, as the widely used approach for various controlled (active) sensing decisions. Its widespread use is mainly due to its computational simplicity and the fact that it admits asymptotic optimality in a wide range of settings. Next, we discuss the shortcomings of Chernoff's rule, mainly because it loses its optimality (even in the asymptotic regime) for the problem at hand. Motivated by this, we finally offer an alternative rule to circumvent the Chernoff rule's shortcomings. We remark that discussing Chernoff's rule serves a two-fold purpose: it furnishes some of the elements for designing the optimal approach and serves as the baseline for assessing the performance of the proposed rule.
}

\subsubsection{Chernoff's Principle}

In the context of the problem studied in this paper, at any time $t$ and based on the filtration $\F_t$,  Chernoff's rule first forms the maximum likelihood (ML) decision about the true model of the data $\sT\in\{\H_0,\H_1\}$. By denoting the ML decision about the true hypothesis at time $t$ by $\delta_{\rm ML}(t)$ we have
\begin{equation}\label{eq:ml:decision}
\delta_{\rm ML}(t)\dff\left\{\begin{array}{ll}
\vspace{1.5mm}
\H_0 \ , & \rm{if}\ \Lambda_{t} < 0 \\
\H_1 \ , & \rm{if}\ \Lambda_{t}\geq 0 
\end{array}\right. .
\end{equation}
Next, based on this decision Chernoff's rule at time $t$ selects and samples the node whose sample is expected to maximally reinforce the decision $\d_{\rm ML}(t)$ to be also the decision at time $(t+1)$. We define $\psi_{\rm ch}(t)$ as the node selected by Chernoff's rule at time $t$, and accordingly define the {\sl ordered} set $\psi_{\rm ch}^t=\{\psi_{\rm ch}(1), \dots, \psi_{\rm ch}(t)\}$. To formalize  Chernoff's rule in the context of the hypothesis testing problem considered in this paper, and in order to quantify the information gained from each sample, we define the following two measures: 
\begin{align}\label{eq:metric:Chernoff1}
D^i_0(t) \dff \; & {\mathscr J}_0(\{i\},\psi_{\rm ch}^{t-1})\ , 
\qquad\mbox{and }\qquad  \ D_1^i(t)\dff \;  {\mathscr J}_1(\{i\},\psi_{\rm ch}^{t-1}) \  ,
\end{align}
where ${\mathscr J}_0$ and ${\mathscr J}_1$ are defined in \eqref{eq:J0} and~\eqref{eq:J1}, respectively. Measure $D_\ell^i(t)$ quantifies the information gained by observing node $i$ at time $t$ when the true hypothesis is $\H_\ell$. 
Chernoff's rule selects the node that maximizes the distance between $f_\ell$ and its alternative when the ML decision is in favor of $\H_\ell$. Therefore, we obtain the following node selection function:
\begin{equation}\label{eq:Chernoff}
\psi_{\rm ch}(t)\dff\left\{\begin{array}{ll}
\vspace{1.5mm}
\displaystyle \argmax_{i\in\varphi^t}\ D_0^i(t) \ , & \mbox{if}\ \delta_{\rm ML}(t-1)=\H_0 \\
\displaystyle \argmax_{i\in\varphi^t}\ D_1^i(t) \ , & \mbox{if}\ \delta_{\rm ML}(t-1)=\H_1
\end{array}\right. .
\end{equation}
To avoid any ambiguities, whenever $\argmax_{i\in\varphi_n^t}\ D_\ell^i(t)$ is not unique (for instance, at the beginning of the sampling process), we select one node randomly according to a uniform distribution.  Chernoff's rule minimizes the average delay in the asymptote of a low rate of erroneous decisions if all the selection actions are {\sl independent}~\cite{Chernoff1959,Atia}, which in the context of this paper translates to testing for two distributions without any correlation structures. In this paper, however, the available actions, i.e., selecting unobserved nodes, are co-dependent due to the underlying MRF's correlation structure. Therefore, Chernoff's rule, which ignores such correlation, fails to exploit it. Specifically, by selecting the best immediate action, Chernoff's rule ignores the perspective of the decisions and the impact of the current decision on the future ones.

\textcolor{black}{We provide an example in Section~\ref{sec:counter} through which we show that designing the node selection rule based on Chernoff's principle is not optimal (even asymptotically).} Our analyses show that incorporating the impact of the decisions on future actions improves the agility of the process significantly. This, in turn, brings about computational complexities, which we will show that can be reduced considerably by leveraging the MRF structures. In the context of the problem analyzed in this paper, another disadvantage of Chernoff's rule is that when the MRFs are comprised of multiple disconnected subgraphs. In such cases, the sampling strategy will be trapped in one subgraph until it exhausts all the nodes in that subgraph before moving to another one. This limits the flexibility of the sampling strategy for freely navigating the entire graph. Another shortcoming of Chernoff's rule that penalizes the quickness significantly is when the highly correlated nodes (random variables) are concentrated in a cluster with a size considerably smaller than that of the graph $n$. In such cases, our proposed selection rule approaches the cluster more rapidly.

\subsubsection{Active Sampling Rule}

We start by introducing information measures that \textcolor{black}{link the node selection decisions over time. This enables} dynamically incorporating the impact of the decision at any given time on all possible future ones. We select these measures to facilitate selecting the nodes, the samples of which maximize the combination of immediate information, and future expected information. To this end, at time $t$ and for each node $i\in\varphi^t$ we define the set $\mathcal{R}_t^i$ as the set of all subsets of $\varphi^t$ that contain $i$, i.e., 
\begin{align}
\label{eq:S_valid}
\mathcal{R}_t^i\dff\{\cS\;:\;\cS\subseteq \varphi^t\ \text{and}\ i\in\cS\} \ .
\end{align}
Corresponding to the samples collected from the nodes in the set $\cS\in\mathcal{R}_t^i$, under $\H_0$ and $\H_1$ we define the following information measures:
\begin{align} \label{eq:M0}
M_0^i (t,\cS) \dff {\mathscr J}_0(\cS,\psi^{t-1})=\ \bbe_0\bigg\{\ln\dfrac{f_0(X_{\textcolor{black}{\cS}};\textcolor{black}{\cS}\,|\,\F_{t-1})}{f_1(X_{\textcolor{black}{\cS}};\textcolor{black}{\cS}\,|\,\F_{t-1})}\bigg\}\ ,\\
\label{eq:M1}
\mbox{and}\qquad  M_{\textcolor{black}{1}}^i (t,\cS) \dff {\mathscr J}_1(\cS,\psi^{t-1})=\ \bbe_1\bigg\{\ln\dfrac{f_1(X_{\textcolor{black}{\cS}};\textcolor{black}{\cS}\,|\, \F_{t-1})}{f_0(X_{\textcolor{black}{\cS}};\textcolor{black}{\cS}\,|\,\F_{t-1})}\bigg\}\ .
\end{align}
The terms $M_\ell^i(t,\cS)$ capture the information content of $|\cS|$ samples. Hence, the normalized terms $\frac{1}{|\cS|}M_\ell^i(t,\cS)$ account for the {\sl average} information content  per sample. Based on these two normalized measures, an optimal action is to select the node that maximizes the {\sl average} information over all possible future decisions. Therefore, the node selection function is the solution of the following optimization problem over all combinations of the unobserved nodes:
\begin{equation}\label{eq:MC}
\psi^*(t)=\left\{\begin{array}{ll}
\vspace{2.5mm}
\displaystyle \argmax_{i\in\varphi^t}\ \max_{\cS\in\mathcal{R}_t^i} \enspace \frac{M_0^i(t,\cS)}{|\cS|}\ , & \mbox{if}\ \delta_{\rm ML}(t-1)=\H_0 \\
\displaystyle  \argmax_{i\in\varphi^t}\ \max_{\cS\in\mathcal{R}_t^i} \enspace \frac{M_1^i(t,\cS)}{|\cS|} \ , & \mbox{if}\ \delta_{\rm ML}(t-1)=\H_1
\end{array}\right. .
\end{equation}
In this selection rule, an ML decision about the true hypothesis is formed based on the collected data, and the node that maximizes the average information over all possible future sequences of samples is selected. We note that the sets $\cS$ are selected such that they i) contain node $i$, which is a candidate to be observed at time $t$, and ii) contain possibly additional nodes that will be observed in the future. Mimicking this decomposition of $\cS$, for $\cS\in{\cal R}_t^i$, the information measure $M_\ell^i(t,\cS)$ can be also decomposed according to
\begin{align}\label{eq:M_expand}
M_\ell^i(t,\cS)={\mathscr J}_\ell(\{i\},\psi^{t-1})\; + \; {\mathscr J}_\ell(\cS\backslash\{i\},\psi^{t-1}) \ ,\quad \mbox{for }\ell\in\{0,1\} \ .
\end{align}
In this expansion, the first term in the decomposition, i.e., ${\mathscr J}_\ell(\{i\},\psi^{t-1})$ defined in \eqref{eq:J0} and~\eqref{eq:J1}, accounts for the information gained by observing node $i$ at time $t$. Similarly, the second term ${\mathscr J}_\ell(\cS\backslash\{i\},\psi^{t-1})$ is the expected information gained from future samples from the nodes contained in $\cS\backslash\{i\}$ when $\psi(t)=i$. This second term constitutes the key distinction of the proposed rule compared to Chernoff's rule, which accounts for incorporating every possible future action. Finding the optimal node $i$ and set $\cS$ in~\eqref{eq:MC} involves an exhaustive search over all the remaining nodes, which can become computationally prohibitive. \textcolor{black}{In the next subsection we show that} by leveraging the Markov properties of an MRF, and a certain acyclic dependency assumption, the exhaustive search for an optimal $\cS\in{\cal R}_t^i$ can be  simplified significantly. \textcolor{black}{Based on the stopping rule specified in~\eqref{eq:stop}, the terminal decision rule given in~\eqref{eq:SPRT}, and the sampling rule specified in~\eqref{eq:ModChernoff2}, Algorithm~\ref{table} provides the detailed steps for detecting a Markov network with a certain correlation structure.}

\begin{table}
\renewcommand{\arraystretch}{1}
\renewcommand{\tablename}{Algorithm}
\caption{Network-guided active sampling for quickest detection of Markov networks}
\vspace{0.05in}
\centering
\begin{tabular}{ll}
\hline \hline
1 & {\bf set} $t=0$, $\varphi^1=\{1,\dots,n\}$, $\Lambda_0=0$, $\gamma_n^{\rm L} =-n{\beta}$, and $\gamma_n^{\rm U} =n{\alpha}$\\
2 &  $t\leftarrow t+1$ \\
3 &  {\bf for} $i\in\varphi^t$ \\
5 & \quad {\bf for} any $\cS\subseteq{\cal R}_t^i$\\
6 & \quad\quad compute $M_0^i(t,\cS)$ and $M_1^i(t,\cS)$ according to \eqref{eq:M0}-\eqref{eq:M1}\\
7 &  \quad {\bf end for}\\
8 &  {\bf end for}\\
9 &  find $\psi^*(t)$ based on \eqref{eq:MC}  \\
10 &  $\varphi^{t+1} \leftarrow \varphi^{t+1}\setminus \psi(t)$\\
11 &  compute $\Lambda_t$ according to \eqref{LLR2}\\   
12 & {\bf if} $\gamma_n^{\rm L}  \; < \; \Lambda_t \; < \; \gamma_n^{\rm U} $ and $t<n$\\
13 & \quad\quad go to step 2\\
14 & {\bf else if} $\Lambda_t \; < \; 0$\\
15 & \quad\quad {\bf set} $\d^*_n=0$ and $\tau^*_n=t$\\
16 & {\bf else}\\
17 & \quad\quad {\bf set} $\d^*_n=1$ and $\tau^*_n=t$\\
18 & {\bf end if}\\
\hline\hline
\end{tabular}
\label{table}
\end{table}

\section{Main Results}
\label{sec:per}
\textcolor{black}{In this section, we provide performance guarantees for the proposed network-guided active sampling procedure in Algorithm~1. Specifically, we analyze (ii) the delay (sample complexity) of the algorithm in Section~\ref{sec:asymp}; (ii) the accuracy of the decision in Section~\ref{sec:feasible}; (iii)  the error exponents in Section~\ref{sec:exponent}; and (iv) the complexity of the node-selection rule in Section~\ref{sec:complexity}.}

\subsection{Decision Reliability}
\label{sec:feasible}

\textcolor{black}{Problem $\cP({\alpha},{\beta})$ by design faces a hard constraint on the number of available samples $n$. This, in turn, acts as a {\sl hard} constraint on the stopping time $\tau_n$. Under such a constraint, the error probabilities $\mathsf{P}_n^{0}$ and $\mathsf{P}_n^{1}$ cannot necessarily be made arbitrarily small simultaneously. Hence, a decision algorithm provides a {\sl feasible} solution to $\cP({\alpha},{\beta})$ only if it satisfies the constraints enforced on $\mathsf{P}_n^{0}$ and $\mathsf{P}_n^{1}$ while not requiring more than $n$ samples.
\begin{definition}[$(\alpha,\beta)$-accuracy] We say that a decision tuple $\Phi_n\dff(\tau_n,\delta_n,\psi_n^{\tau_n})$ is $(\alpha,\beta)$-accurate if it ensures  $\mathsf{P}_n^{0}\leq {\e}^{-n{\alpha}}$ and $\mathsf{P}_n^{1}\leq {\e}^{-n{\beta}}$. First, we establish that the decisions generated by Algorithm~1 satisfy the performance guarantees of the problem.
\end{definition}}
\noindent \textcolor{black}{In this subsection, we examine the problem \eqref{eq:Opt} in both the asymptotic and non-asymptotic regime with respect to the size of the network $n$, and characterize conditions on $\alpha$ and $\beta$ under which  Algorithm~1 is guaranteed to generate $(\alpha,\beta)$-accurate decisions. To this end, note that the sampling process terminates if i) the LLR $\Lambda_t$ exits the band $(\gamma_n^{\rm L} ,\gamma_n^{\rm U} )$ at some $t\in\{1,\dots,n\}$, or ii) we exhaust all the samples, i.e., $\tau_n=n$. For establishing the conditions for ensuring $(\alpha,\beta)$-accuracy, in the first step, we show that if the process terminates by exiting the band $(\gamma_n^{\rm L} ,\gamma_n^{\rm U} )$, then the decision is $(\alpha,\beta)$-accurate. In the second step, we evaluate the probability of $\Lambda_t$ exiting the band $(\gamma_n^{\rm L} ,\gamma_n^{\rm U} )$ prior to exhaustive all $n$ samples. These two steps, collectively, establish a sufficient condition for ensuring $(\alpha,\beta)$-accuracy of Algorithm~1.
For this purpose, we denote the Bhattacharyya coefficient, as a measure of similarity of the two distributions, by
\begin{align}\label{eq:bc}
{\sf B}_n(f_0,f_1)\dff\int \sqrt{f_0(x;\V)f_1(x;\V)}\; \der x\ .
\end{align}
Accordingly, we denote the {\sl normalized} Bhattacharyya {distance} by
\begin{align}\label{eq:bd}
\kappa (f_0,f_1)\dff  - \lim_{n\rightarrow\infty} \frac{1}{n}\ln{\sf B}_n(f_0,f_1)\ .
\end{align}
}
The following theorem establishes a sufficient condition under which Algorithm~1 generates $(\alpha,\beta)$-accurate solutions. 
\begin{theorem}[Non-asymptotic $(\alpha,\beta)$-accuracy]\label{thm:nonasymp}
\textcolor{black}{For a given network size $n$,  Algorithm~1 generates an $(\alpha,\beta)$-accurate solution with a probability at least
\begin{align}\label{eq:feasibility}
1-{\sf B}_n(f_0,f_1)\left[\epsilon_0 \exp\left({\frac{n{\beta}}{2}}\right)+\epsilon_1 \exp\left({\frac{n{\alpha}}{2}}\right)\right]\ .
\end{align}}
\end{theorem}
\begin{proof}
See Appendix \ref{App:thm:nonasymp}.
\end{proof}
\textcolor{black}{We will evaluate the probability term in \eqref{eq:feasibility} in Section~\ref{sec:feasible_NA}  through an illustrative example and in Section~\ref{sec:simulations} through numerical evaluations. We will show that for widely used MRF models (e.g., Gaussian MRFs), this probability approaches 1 in all practical ranges of $n$ and error probabilities, rendering the Algorithm~1 $(\alpha,\beta)$-accurate almost surely even in the non-asymptotic regime. By leveraging the result of Theorem~\ref{thm:nonasymp}, we can readily provide a sufficient condition for $(\alpha,\beta)$-accuracy in the asymptote of large network sizes.}
\textcolor{black}{\begin{corollary}[Asymptotic $(\alpha,\beta)$-accuracy]\label{thm:asymp}
Algorithm~1 generates $(\alpha,\beta)$-accurate solutions {\sl almost surely} in the asymptote of large networks if
\begin{align}\label{eq:asymp:as}
\max\{{\alpha},{\beta}\}\; < \; 2\kappa(f_0,f_1)\ .
\end{align}
\end{corollary}
\begin{proof}
The proof follows from finding a sufficient condition that ensures probability in~\eqref{eq:feasibility} approaches 1 as $n\rightarrow\infty$. 
\end{proof}
}

\subsection{Delay Analysis}
\label{sec:asymp}

\textcolor{black}{In this subsection, we analyze the performance of the proposed selection rule in the asymptote of large networks sizes, i.e.,  when
 $n\rightarrow\infty$, i.e., $V=\mathbb{N}$.}  We note that the proposed network-guided node selection rule capitalizes on the discrepancies in the information measures corresponding to selecting different nodes. In general, a wider range of information measures leads to more effectively distinguishing the most informative nodes to sample. This, in turn, reduces the average delay for reaching a sufficiently confident decision. In order to analyze the performance, \textcolor{black}{corresponding to any subset of nodes $U\subseteq \mathbb{N}$ we define normalized LLR measures as follows:
\begin{align}
\label{eq:cc1}
{\sf nLLR}_0(X_U;U)\dff\;&\frac{1}{|U|}\ln\textcolor{black}{\frac{f_0(X_U;U)}{f_1(X_U;U}}\ ,\quad (X_1,\dots, X_n)\sim \P_0\  ,\\
\label{eq:cc2}\text{and }\quad {\sf nLLR}_1(X_U;U)\dff\;&\frac{1}{|U|}\ln\textcolor{black}{\frac{f_1(X_U;U)}{f_0(X_U;U)}}\ ,\quad (X_1,\dots, X_n)\sim \P_1\ .
\end{align}
Such log-likelihood ratios play pivotal roles in characterizing the performance  of sequential methods. When the random variables $\{X_i:i\in\{1,\dots,n\}\}$ are independent and identically distributed (i.i.d.), according to the strong law of large numbers, the  measures ${\sf nLLR}_\ell(Y^t;\psi^t)$ converge almost surely to the KL divergence terms as $|U|\rightarrow\infty$. While in an i.i.d. setting these measures are well-defined and can have tangible interpretations (e.g., being random walks), in a non-i.i.d. setting, they are not as well-defined, and their convergence can be guaranteed only under stronger conditions. A relevant notion of convergence for non-i.i.d. settings that is especially widely used in sequential detection is {\sl complete} convergence (introduced in~\cite{Hsu}, a good overview in~\cite{Karr:1993}, and used in the context of sequential detection in~\cite{MhypI,Tartakovski2017})\footnote{\textcolor{black}{In some literature it is also called 1-quick convergence (see \cite{lai}) with generalizations to stronger $r$-quickness convergence in~\cite{MhypI}}}}. \textcolor{black}{For this purpose,  corresponding to the set of nodes $V$ we define
\begin{align}\label{eq:S}
{\cal S}(V) \dff \{\forall A\subseteq V\;:\; |A| \geq g(n)\}\ ,
\end{align}
where $g(x)$ is an arbitrary function that satisfies $g(x)\xrightarrow{x\rightarrow\infty}{\infty}$. Hence, ${\cal S}(V)$ is the collection of all subsets of $V$ whose cardinality is at least $g(n)$. 
\begin{definition}[Complete convergence]
Corresponding to any possible sampling sequence $\psi^\infty\in{\cal S}(\mathbb{N})$, we say that the normalized log-likelihood ratios ${\sf nLLR}_\ell(Y^t;\psi^t)$ converge {\em completely} to a constant $I_\ell(\psi^\infty)$ when
\begin{align}\label{eq:complete_conv1}
\sum_{t=1}^\infty \P_\ell\Big\{\left|{\sf nLLR}_\ell(Y^t;\psi^t)-I_\ell(\psi^\infty)\right|>h\Big\}<+\infty\ , \quad \forall h>0 \ .
\end{align}
\end{definition}
\noindent It can be readily verified that the condition in~\eqref{eq:complete_conv1} is equivalent to
\begin{align}\label{eq:finite_ave}
\bbe_\ell[T_\ell(h,\psi^\infty)] <+\infty \ ,\quad \forall h>0 \  ,
\end{align}
where we have defined
\begin{align}\label{eq:T_l}
T_\ell(h,\psi^\infty )\dff\sup \  \left\{t\in\mathbb{N} \;:\;\Big|{\sf nLLR}_\ell(X_{U^t};U^t)- I_\ell(\psi^\infty)\Big|\geq h\right\}\ .
\end{align}
The term $T_\ell(h,\psi^\infty)$ denotes the last time that the sequence $\{{\sf nLLR}_\ell(Y^t;\psi^t)\}$ leaves the interval 
\begin{align}
[I_\ell(\psi^\infty)-h\; , \; I_\ell(\psi^\infty)+h]\ .
\end{align}
Next, we define two types of networks, depending on how the LLR sequences converge. 	
\begin{definition}[Homogeneous network] We say that an MRF is {\sl homogeneous} when $I_\ell(\psi^\infty)$ exists and it is the same for all possible sets $\psi^\infty$. When we have a homogeneous structure, we replace $I_\ell(\psi^\infty)$ by the shorthand  $I_\ell$, which emphasizes a lack of dependence on $\psi^\infty$. 
\end{definition}
\noindent The critical property of homogeneous networks is that  observing any subsequence of nodes provides the same average amount of information in the long run. 
\begin{example}
Consider a setting in which the samples are i.i.d. under $\sH_0$ and they form a Gauss-Markov random field (GMRF) under $\sH_1$ with the same marginal distributions as the ones under $\sH_0$. If under $\sH_1$ the nodes form a line graph with correlation coefficients $a\neq \pm 1$, then we have a homogeneous network in which
\begin{align}
 I_0 &= \ln(1-a^2)+\frac{2a^2}{1-a^2}\ ,  \qquad  \mbox{and}\qquad I_1 = \ln\frac{1}{1-a^2} \ .
\end{align}
\end{example}
\begin{definition}[Heterogeneous network] 
We say that an MRF is {\sl heterogeneous}  when the two information measures $I_\ell(\psi^\infty)$ exist and vary for different permutations $\psi^\infty$. In such networks, we define
\begin{align}\label{eq:I0*}
I_\ell^*\dff \sup_{\psi^\infty\in{\cal S}(\mathbb{N})} \ I_\ell(\psi^\infty)\ , \qquad \mbox{for}\;\;\ell\in\{0,1\}\ .
\end{align}
\end{definition}
\begin{example}
Consider a setting in which the samples are i.i.d. under $\sH_0$, and they form a GMRF under $\sH_1$ with the same marginal distributions as the ones under $\sH_0$.  Under $\sH_1$ the dependency graph consists of two line subgraphs defined over two distinct sets of nodes denoted by 
\begin{align}
\psi^\infty_1=\{2k-1: k\in\mathbb{N}\} \quad \mbox{and} \quad 
\psi^\infty_2=\{2k: k\in\mathbb{N}\}\ ,
\end{align}
where the elements in $\psi^\infty_i$ have constant correlation coefficients $a_i\neq\pm 1$. Assuming $|a_1|>|a_2|$ we have
\begin{align}
 I_0(\psi^\infty_i) &= \ln(1-a_i^2)+\frac{2a_i^2}{1-a_i^2}  \qquad  , \qquad && I_0^* =  \ln(1-a_1^2)+\frac{2a_1^2}{1-a_1^2}  \ , \\
 I_1(\psi^\infty_i)  & = \ln\frac{1}{1-a_i^2}   \qquad  , \qquad &&  I_1^* = \ln\frac{1}{1-a_1^2} \ .
\end{align}
\end{example}
\noindent Based on the measures $I_\ell$ and $I^*\ell$ in homogeneous and heterogeneous settings, respectively, next, we analyze the average stopping time. We first focus on the homogeneous setting and establish the optimality of stopping and terminal decision rules characterized in~\eqref{eq:th1}--\eqref{eq:SPRT} and then generalize the results to the heterogeneous setting. The following lemma will be instrumental in evaluating the average stopping time.
\begin{lemma}\label{lemma:kappa}
For the choices of ${\alpha}$ and ${\beta}$ that satisfy \eqref{eq:asymp:as}, in the homogeneous and heterogeneous networks we have
\begin{align}
\max\{{\alpha},{\beta}\} \leq \min\{I_0,I_1\}\ , \qquad \mbox{and} \qquad \max\{{\alpha},{\beta}\} \leq \min\{I^*_0,I^*_1\}\ ,
\end{align}
almost surely.
\end{lemma}
\begin{proof}
See Appendix~\ref{app:lemma:kappa}.
\end{proof}
The following theorem provides a universal (algorithm-independent) lower bound on the average delay for any feasible solution to problem~\eqref{eq:Opt} when the network has a homogeneous dependency structure.}
\textcolor{black}{\begin{theorem}[Homogeneous Structures -- Delay Converse]\label{thm:asymp1}
\textcolor{black}{In a homogeneous network with information constants $I_0$ and $I_1$}, any feasible solution of problem~\eqref{eq:Opt} with the stopping time $\tau_n $ satisfies
\begin{align}\label{eq:lower1}
\lim_{n\rightarrow\infty}  \frac{\bbe_0\{\tau_n\}}{n} &\geq\frac{{\beta}}{ I_0}\ , \qquad \mbox{and} \qquad  \lim_{n\rightarrow\infty}  \frac{\bbe_1\{\tau_n\}}{n} \geq\frac{{\alpha}}{I_1}\ .
\end{align}
\end{theorem}}
\begin{proof}
See Appendix \ref{App:thm:asymp1}.
\end{proof}
We show that any selection rule combined with the likelihood ratio test given in~\eqref{eq:th1}--\eqref{eq:SPRT} achieves these lower bounds.
\begin{theorem}[\textcolor{black}{Homogeneous Structures -- Delay Achievability}]\label{thm:asymp2}
In a homogeneous network, for the stopping and terminal decision rules specified in~\eqref{eq:th1}--\eqref{eq:SPRT} and an arbitrary sampling rule, in the asymptote of large~$n$ we have
\begin{align}
\lim_{n\rightarrow\infty} \frac{\bbe_0\{\tau_n^*\}}{n} &\leq\frac{{\beta}}{I_0}  \ , \qquad \mbox{and} \qquad  \lim_{n\rightarrow\infty} \frac{\bbe_1\{\tau_n^*\}}{n} \leq\frac{{\alpha}}{ I_1}\ .
\end{align}
\end{theorem}
\begin{proof}
See Appendix \ref{App:thm:asymp2}.
\end{proof}
\textcolor{black}{The last two theorems, collectively, establish that when the network has a homogeneous structure, irrespectively of how the nodes are selected and sampled over time, the stopping and terminal decision rules specified in~\eqref{eq:th1}--\eqref{eq:SPRT} render asymptotically optimal decisions. The optimality of the decisions being independent of the node selection rule signifies that in homogeneous structures, all sequences of nodes, asymptotically, contain the same average amount of information, and the overall performance does not critically depend on the sampling path. Next, we show that the observation above is not necessarily valid for the networks with heterogeneous structures, and the optimality of decisions in those networks critically depends on the sampling path. By leveraging Theorem~\ref{thm:asymp1}, in the next corollary, we first provide algorithm-independent lower bounds on the average delay in heterogeneous networks. 
\textcolor{black}{
\begin{corollary}[Heterogeneous Structures -- Delay Converse]\label{thm:opt:L}
\textcolor{black}{In a heterogeneous network with information constants $I^*_0$ and $I^*_1$}, any feasible solution of problem~\eqref{eq:Opt} with the stopping time $\tau_n $ satisfies
\begin{align}
\lim_{n\rightarrow\infty}  \frac{\bbe_0\{\tau_n\}}{n} &\geq\frac{{\beta}}{I^*_0}\ ,\qquad \mbox{and} \qquad  \lim_{n\rightarrow\infty}  \frac{\bbe_1\{\tau_n\}}{n} \geq\frac{{\alpha}}{I^*_1}\ .
\end{align}
\end{corollary}
\begin{proof}
\textcolor{black}{By following the same line of argument as in the proof of Theorem~\ref{thm:asymp1} we can show that for any arbitrary sampling path $\psi^\infty\in\mathbb{N}$ we have
\begin{align}\label{eq:lower21}
\lim_{n\rightarrow\infty}  \frac{\bbe_0\{\tau_n\}}{n} &\geq\frac{{\beta}}{ I_0(\psi^\infty)}\ , \qquad \mbox{and} \qquad 
 \lim_{n\rightarrow\infty}  \frac{\bbe_1\{\tau_n\}}{n} \geq\frac{{\alpha}}{I_1(\psi^\infty)}\ .
\end{align}
Since \eqref{eq:lower21} is true for any set $\psi^\infty$, subsequently, we have
\begin{align}\label{eq:low:pr:0}
\lim_{n\rightarrow\infty}  \frac{\bbe_0\{\tau_n\}}{n} &\geq \inf_{\psi^\infty\in\mathbb{N}} \; \frac{{\beta}}{ I_0(\psi^\infty)} \overset{\eqref{eq:I0*}}{=} \frac{{\beta}}{I^*_0}\ ,  \qquad \mbox{and} \qquad  \lim_{n\rightarrow\infty}  \frac{\bbe_1\{\tau_n\}}{n} \geq \inf_{\psi^\infty\in\mathbb{N}} \; \frac{{\alpha}}{ I_1(\psi^\infty)} \overset{\eqref{eq:I0*}}{=} \frac{{\alpha}}{I^*_1}\ .
\end{align}}
\end{proof}}
\noindent Next, we provide the proof for the optimality of the decisions produced by Algorithm~1, and especially the optimality of the proposed dynamic node selection rule when facing heterogeneous networks. This result will also be instrumental in characterizing the performance gap between the proposed sampling strategy and the Chernoff rule. By characterizing such a gap, through an example in Section~\ref{sec:counter}, we will show that the Chernoff rule loses its optimality for the correlation detection problem in networks. To prove the upper bounds on the average delay, we define the random variable $\hat\tau_n$ as the first time instant after which the ML decision about the true hypothesis is always correct, i.e.,
\begin{equation}\label{eq:tau_hat}
\hat\tau_n\dff\inf\{u\;:\;\delta_{\rm ML}(t)=\sT\ ,\ \forall t\geq u\}\ ,
\end{equation}
where we adopt the convention that the infimum of an empty set is $+\infty$. We emphasize that $\hat\tau_n$ is not a stopping time, but rather a term that, as we will show, is dominated by the stopping time. In order to establish the desired upper bounds, we show the following two properties for $\hat\tau_n$:
\begin{enumerate}
\item $\bbe_i\{\hat\tau_n\}$ is upper bounded by a constant.
\item $\frac{1}{n}\bbe_i\{\tau^*_n-\hat\tau_n\}$ is upper bounded according to
\begin{align}\label{eq:itme1}
\lim_{n\rightarrow\infty}  \frac{\bbe_0\{\tau^*_n-\hat\tau_n\}}{n} \leq \frac{{\beta}}{I^*_0}\ , \qquad \mbox{and} \qquad \lim_{n\rightarrow\infty}  \frac{\bbe_1\{\tau^*_n-\hat\tau_n\}}{n} \leq \frac{{\alpha}}{I^*_1}\ .
\end{align}
\end{enumerate}
In order to prove that $\bbe_i\{\hat\tau_n\}$ is finite, we first provide the following lemma, which at the core establishes that the probability  $\mathbb{P}_i(\hat\tau_n \geq t)$ decays exponentially with respect to time $t$.
\begin{lemma}\label{lemma:finite}
$\bbe_i\{\hat\tau_n\}$ is upper bounded by a constant. 
\end{lemma}
\begin{proof}
See Appendix \ref{app:lemma:finite}.
\end{proof}
Next, in order to prove~\eqref{eq:itme1}, we define
\begin{align}\label{eq:U_inf}
U^\infty \dff \arg\max_{\psi^\infty\in{\cal S}(\mathbb{N})} I_1(\psi^\infty)\ ,
\end{align}
corresponding to which we have $I_1(U^\infty)=I_1^*$.  When there are more than one choice for $U^\infty$, we select it to be the largest such set. Based on this definition, we provide the following lemma showing that the number of times that we sample from a set other than $U^\infty$ is, on average, finite. This property follows from the assumption of complete convergence in heterogeneous networks.
\begin{lemma}\label{lemma:psi}
For the set $\mathcal{H}\dff\{s\in\{1,\dots,\tau_n\}\;:\; \psi^*(s)\notin U^\infty\}$ we have $\bbe_i\{H\}<+\infty $.
\end{lemma}
\begin{proof}
See Appendix \ref{app:lemma:psi}.
\end{proof}
This establishes that by the stopping time, the samples collected are taken dominantly from the set $U^\infty$. By leveraging Lemma~\ref{lemma:psi}, we next provide the final ingredient for characterizing the achievable average delay. 
\begin{lemma}\label{lemma2}
$\frac{1}{n}\bbe_i\{\tau^*_n-\hat\tau_n\}$ is upper bounded according to \eqref{eq:itme1}.
\end{lemma}
\begin{proof}
See Appendix \ref{app:lemma2}.
\end{proof}}

\begin{theorem}[\textcolor{black}{Heterogeneous Structures -- Delay Achievability}]\label{thm:opt:U}
\textcolor{black}{Algorithm~1 generates decisions that are asymptotically optimal solutions to problem~\eqref{eq:Opt}.} Specifically
\begin{align}
\label{eq:het_ub1} \lim_{n\rightarrow\infty} \frac{\bbe_0\{\tau^*_n\}}{n} &\leq \frac{{\beta}}{I^*_0}\ ,\qquad \mbox{and} \qquad \lim_{n\rightarrow\infty}  \frac{\bbe_1\{\tau^*_n\}}{n} \leq \frac{{\alpha}}{I^*_1}\ .
\end{align}
\end{theorem}
\begin{proof}
By combining the results of Lemma~\ref{lemma:finite} and Lemma~\ref{lemma2} we obtain
\begin{align}
\frac{{\alpha}}{I^*_1} & \overset{\eqref{eq:itme1}}{\geq} \lim_{n\rightarrow\infty}\frac{\bbe_1\{\tau^*_n-\hat\tau_n\}}{n} \\
& \; \geq \lim_{n\rightarrow\infty}\frac{\bbe_1\{\tau^*_n\}}{n}-  \lim_{n\rightarrow\infty}\frac{B}{n(1- {\e}^{-c})}\\
& \; = \lim_{n\rightarrow\infty}\frac{\bbe_1\{\tau^*_n\}}{n}\ .
\end{align}
which concludes the proof for the upper bound on $\frac{1}{n}{\bbe_1\{\tau^*_n\}}$. The proof of the upper bound on $\frac{1}{n}{\bbe_0\{\tau^*_n\}}$ follows the same line of argument.
\end{proof}

\subsection{\textcolor{black}{Error Exponents}}
\label{sec:exponent}
In this subsection, we characterize the gain obtained from the data-adaptive stopping time. To this end, we compare the performance of sequential sampling procedures with that of the fixed-sample-size setting in terms of their associated error exponents.
\textcolor{black}{In the fixed-sample-size counterpart of the binary testing problem considered in this paper, the optimal decision rule is the Neyman-Pearson (NP) rule, where its associated error exponents are characterized in~\cite{Anandkumar:IT09}. By denoting the NP decision rule by $\d_{\rm NP}$, we define 
\begin{align}
{\sf  P}_{\rm NP}^{0}& \;\dff\; \mathbb{P}_0(\d_{\rm NP}=1)\ ,\qquad \mbox{and} \qquad \quad{\sf  P}_{\rm NP}^{1} \;\dff\; \mathbb{P}_1(\d_{\rm NP}=0)\ ,
\end{align}
as the frequencies of erroneous decisions by the NP test based on $n$ samples. Accordingly, we define
\begin{align}
E_{\rm NP}^{0} & \;\dff\; -\lim_{n\rightarrow\infty} \frac{1}{n} \ln {\sf  P}_{\rm NP}^{0}\ , \qquad \mbox{and} \qquad  E_{\rm NP}^{1} \;\dff\; -\lim_{n\rightarrow\infty} \frac{1}{n} \ln {\sf  P}_{\rm NP}^{1}\ ,
\end{align}
as the associated error exponents. Similarly, we define
\begin{align}
E_n^{0} & \;\dff\; -\lim_{n\rightarrow\infty}\frac{1}{r_1}\ln \mathsf{P}_n^{0}(r_1)\ , \qquad \mbox{and} \qquad  E_n^{1}  \;\dff\; -\lim_{n\rightarrow\infty}\frac{1}{r_0}\ln \mathsf{P}_n^{1}(r_0)\ ,
\end{align}
as the error exponents of the sequential detection approach, where $\mathsf{P}_n^{0}(r_1)$ and $\mathsf{P}_n^{1}(r_0)$ are the error probabilities of sequential sampling when the average number of samples (i.e., the stopping time) is $r_\ell\triangleq \bbe_\ell\{\tau^*_n\}$.} The connections between the error exponents of the NP test and sequential sampling strategies are established in the following theorem.
\begin{theorem}[Gain of Adaptivity]\label{thm:per}
\textcolor{black}{The error exponents of the decision rules in Algorithm~1 are related to those of the NP rule through
\begin{align}\label{eq:seq:exp}
& E_n^{1}=I_0 \qquad\mbox{\rm and}\qquad E_n^{0}=I_1 \ ,\\
\label{eq:fix:exp}
& E_{\rm NP}^{1}= I_0 \qquad\mbox{\rm and}\qquad E_{\rm NP}^{0}=0   \ .
\end{align}}
\end{theorem}
\begin{proof}
See Appendix \ref{App:thm:per}.
\end{proof}

\subsection{\textcolor{black}{Search Complexity Analysis}}
\label{sec:complexity}

\textcolor{black}{In this subsection, we show that under certain connectivity structures for the given MRFs, by judiciously leveraging the structures, the complexity of the search for the optimal node selection path over time can be reduced significantly. For this purpose, based on the given graphs $\G_0(\V,\E_1)$ and $\G_1(\V,\E_2)$ we construct the graph $\G(\V,\E)$ such that
\begin{align}\label{eq:E}
\E \; \dff \; \E_0 \cup \E_1\ .
\end{align}
Based on this, we define the neighborhood of node $i\in\V$ according to
\begin{align}
\N_i\dff \{j\in\V\;:\;j\neq i\;,\;(i,j)\in\E \}\ .
\end{align}
We will show that when $\G$ is acyclic,} for each node $i$, the optimal set $\cS$ is restricted to only contain the neighbors of $i$ that are not observed prior to time $t$, i.e., $\cS\subseteq\L_t^i$ where
\begin{equation}
\L_t^i\dff \{i\}\cup\{\N_i \cap \varphi^t\} \ .
\end{equation}
This indicates that for determining the node to select at each time, it is sufficient to consider a significantly shorter future sampling path for each node. The cardinality of the set of subsets of $\L_t^i$ is significantly smaller than that of $\varphi^t$, which translates to a substantial reduction in the complexity of characterizing the  optimal selection functions. This observation is formalized in the following theorem.
\begin{theorem}\label{thm:markov}
For an acyclic dependency graph $\G$, at each time $t$ and for $\ell\in\{0,1\}$ we have
\begin{align}
\argmax_{i\in\varphi^t} \max_{\cS\in\mathcal{R}_t^i} \frac{M_{\ell}^i(t,\cS)}{|\cS|}=\argmax_{i\in\varphi^t} \max_{\cS\subseteq\L_t^i} \frac{M_{\ell}^i(t,\cS)}{|\cS|}\ .
\end{align}
\end{theorem}
\begin{proof}
See Appendix \ref{App:thm:markov}.
\end{proof}
\textcolor{black}{
\noindent Based on this theorem, the selection function given in \eqref{eq:MC} simplifies to
\begin{equation}\label{eq:ModChernoff2}
\psi^*(t)=\left\{\begin{array}{ll}
\vspace{1.5mm}
\displaystyle \argmax_{i\in\varphi^t}\ \max_{\mathcal{S}\subseteq\L_{t}^i}\, \frac{M_0^i(t,\cS)}{|\cS|}\ , & \rm{if}\ \delta_{\rm ML}(t-1)=\H_0 \\
\displaystyle \argmax_{i\in\varphi^t}\ \max_{\mathcal{S}\subseteq\L_{t}^i}\, \frac{M_1^i(t,\cS)}{|\cS|} \ ,& \rm{if}\ \delta_{\rm ML}(t-1)=\H_1 
\end{array}\right. .
\end{equation}
By further leveraging the Markov property, computing  
\begin{align}
\max_{\mathcal{S}\subseteq\L_{t}^i}\, \frac{M_\ell^i(t,\cS)}{|\cS|}
\end{align}
can be further simplified. Specifically, by recalling the definition of $M_\ell^i(t,\cS)$ given in \eqref{eq:M0} and \eqref{eq:M1} we have
\begin{align}
M_\ell^i(t,\cS) & =  D_{\rm KL}\big(f_\ell(X_{\cS}|\F_{t-1})\;\|\;f_{1-\ell}(X_{\cS}|\F_{t-1})\big) \\
\label{eq:complexity1}
& = D_{\rm KL}\big(f_\ell(X_{i}|\F_{t-1})\;\|\;f_{1-\ell}(X_{i}|\F_{t-1})\big) \\
& \hspace{1 in}  + \sum_{j\in\cS\backslash\{i\}} D_{\rm KL}\big(f_\ell(X_{j}|X_i,\F_{t-1})\;\|\;f_{1-\ell}(X_{j}|X_i\F_{t-1})\big) \\
\label{eq:complexity2}& = D_{\rm KL}\big(f_\ell(X_{i}| X_{\psi^*(t-1)})\;\|\;f_{1-\ell}(X_{i}| X_{\psi^*(t-1)})\big) \\
& \hspace{1 in} + \sum_{j\in\cS\backslash\{i\}} D_{\rm KL}\big(f_\ell(X_{j}|X_i)\;\|\;f_{1-\ell}(X_{j}|X_i\big)\ , 
\end{align}
where the transition from \eqref{eq:complexity1} to \eqref{eq:complexity2} is due to the graph being acyclic and Markov. Hence, for computing the information measures $M_\ell^i(t,\cS)$ we need to compute only the marginal distributions of the form $f_\ell(X_i|X_j)$. }

\section{\textcolor{black}{Special Cases and Illustrative Examples}}
\label{sec:special}

In this section, we consider a few special cases, for each of which we present more specialized results. First, \textcolor{black}{for gaining further insight into the tightness of the probabilistic $(\alpha,\beta)$-accuracy guarantee in the non-asymptotic regime (Theorem~\ref{thm:nonasymp}), we provide an illustrative example showing the achievable ranges of error probabilities for a given network size.} Next, we consider the setting in which both distributions are Gaussian and characterize measures defined for designing the sampling strategy in terms of the covariance matrices of the distributions. Built on these results, next, we provide a counterexample establishing that the Chernoff rule is not asymptotically optimal for carrying out the detection decisions in the MRFs considered in this paper. Finally, we consider detecting whether a given MRF contains a cluster of nodes whose data form a given correlation model. In all the special cases, we quantify the performance gaps between our network-guided active sampling strategy and the Chernoff rule.

\subsection{\textcolor{black}{Non-asymptotic Detection Performance}}
\label{sec:feasible_NA}
\textcolor{black}{In this subsection, we provide an illustrative example to assess the sufficient condition for $(\alpha,\beta)$-accuracy of Algorithm~1 in the non-asymptotic regime, which was established in Theorem~\ref{thm:nonasymp}. We consider  testing correlation versus independence  when both distributions are Gaussian, i.e.,
 \begin{equation}\label{eq:hyp:model2}
\H_0\ :\; (X_1,\dots, X_n)\sim {\cal N}(\boldsymbol{\theta},\mathbf I) \qquad \mbox{versus} \qquad \H_1\ :\; (X_1,\dots, X_n)\sim {\cal N}(\boldsymbol{\theta},\mathbf \Sigma)\ ,
\end{equation}
where ${\mathbf I}$ is the identity matrix and $\bSigma$ is an arbitrary correlation matrix such that $\Sigma_{ii}=1$. Hence, the Bhattacharyya distance is given by 
\begin{align}\label{eq:bd2}
\kappa_n(f_0,f_1)& \dff -\ln{\sf B}_n(f_0,f_1)
  \; = \; \frac{1}{2}\ln \frac{1}{\sqrt{{\sf det}\bSigma}}\cdot {\sf det}\left(\frac{{\mathbf I}+\bSigma}{2}\right)
 \; = \; \frac{1}{2}\ln\prod_{i=1}^n \frac{1+\lambda_i}{2\sqrt{\lambda_i}}\ ,
\end{align}
where $\{\lambda_i\}_{i=1}^n$ are the distinct eigenvalues of the symmetric positive definite matrix $\bSigma$. Accordingly, the Bhattacharyya coefficient is given by
\begin{align}\label{eq:bc2}
{\sf B}_n(f_0,f_1)\; = \; \exp\left(-\kappa_n(f_0,f_1)\right) \; = \; \prod_{i=1}^n \sqrt{\frac{2\sqrt{\lambda_i}}{1+\lambda_i}}\ .
\end{align}
By noting that $\Sigma_{ii}=1$ for all $i\in\{1,\dots,n\}$, according to Gershgorin circle theorem all the eigenvalues $\{\lambda_i\}_{i=1}^n$ lie in closed discs centered at 1. Select $\xi>0$ such that at least half of the eigenvalues $\{\lambda_i\}_{i=1}^n$ lie outside the interval
\begin{align}\label{eq:range}
\Big[\big(\sqrt{1+\xi}-\sqrt{\xi}\big)^2\; , \; \big(\sqrt{1+\xi}+\sqrt{\xi}\big)^2\Big ]\ .
\end{align}
It can be readily verified that if
\begin{align}\label{eq:range2}
\lambda_i\notin \Big[\big(\sqrt{1+\xi}-\sqrt{\xi}\big)^2\; , \; \big(\sqrt{1+\xi}+\sqrt{\xi}\big)^2\Big ]\ ,
\end{align}\label{eq:range3}
then
\begin{align}
\frac{2\sqrt{\lambda_i}}{1+\lambda_i} \; < \; \frac{1}{\sqrt{1+\xi}}\ .
\end{align}
Hence, we have the following upper bound on the Bhattacharyya coefficient:
\begin{align}
{\sf B}_n(f_0,f_1)\; \leq \; (1+\xi)^{-\frac{n}{8}}\ .
\end{align}
Therefore, for all the error probability exponents ${\alpha}$ and ${\beta}$ that satisfy
\begin{align}
\frac{1}{4}\ln(1+\xi) \; > \; \max\{{\alpha},{\beta}\}\ ,
\end{align}
according to Theorem~\ref{thm:nonasymp} Algorithm~1 is $(\alpha,\beta)$-accurate almost surely in the non-asymptotic regime. For instance, for $n=200$, $\xi=0.2$, ${\alpha}={\beta}=0.02$, Algorithm~1 is $(\alpha,\beta)$-accurate with probability at least $0.999$. }
\newpage

\subsection{Gauss-Markov Random Fields}
\textcolor{black}{In this subsection we specialize the general results to  GMRF, where we assume that
\begin{equation}\label{eq:hyp:model_G}
\H_0\ :\;  (X_1,\dots, X_n)\sim {\cal N}({\boldsymbol \theta},{\mathbf I}) \qquad \mbox{versus} \qquad 
\H_1\ :\;  (X_1,\dots, X_n)\sim {\cal N}({\boldsymbol \theta},\bSigma)\ ,
\end{equation}
where $\Sigma_{ii}=1$ for all $i\in\{1,\dots,n\}$. This test is generally known as the problem of testing against independence. The graphical model associated with $\H_0$ consists of $n$ nodes without any edges, and we denote the graphical model associated with $\H_1$ by $\G(\V,\E)$.} A GMRF with covariance matrix $\bm\Sigma$ is non-degenerate if $\bm\Sigma$ is positive-definite, in which case, the potential matrix associated with the GMRF is denoted by $ {\bm J}\dff\bm\Sigma^{-1}$. The non-zero elements of the potential matrix have a one-to-one correspondence with the edges of the dependency graph in the sense that
\begin{equation}\label{eq:pot_matrix2}
{ J}_{uv}=0 \quad  \Leftrightarrow\quad  (u,v)\notin\E \ .
\end{equation}
In a GMRF, the properties of the network are strongly influenced by the structure of the underlying dependency graph. GMRFs with acyclic dependency represent an important class of GMRFs in which there exists at most one path between any pair of nodes, and consequently, the cross-covariance value between any two non-neighbor nodes in the graph is related to the cross-covariance values of the nodes connecting them. Specifically, corresponding to any two edges $(i,j)\in \E$ and $(i,k)\in \E$, which share node $i\in\V$, we have
\begin{equation}\label{eq:prop}
{\Sigma}_{jk}={\Sigma}_{ji}\Sigma^{-1}_{ii}{\Sigma}_{ik}\ , \quad \mbox{for all } \{j,k\}\subseteq \N_i \ .
\end{equation}
In a GMRF with an acyclic graph, the elements and the determinant of the potential matrix can be expressed explicitly in terms of the elements of the covariance matrix.

\begin{theorem}[\cite{Anandkumar:IT09}, Theorem 1]\label{thm:potential}
For a GMRF with an acyclic dependency graph $\G=(\V,\E)$ and covariance matrix $\bm{\Sigma}$, the elements of the potential matrix ${J}$ are given by
\begin{align}
{J}_{ii} =\frac{1}{\Sigma_{ii}}\left(1+\sum_{j\in\N_i}\frac{\Sigma_{ij}^2}{\Sigma_{ii}\Sigma_{jj}-\Sigma_{ij}^2}\right)\ ,\quad \forall i\in\{1,\dots,n\}\ ,
\end{align}
and
\begin{align}
{J}_{ij}=\left\{
\begin{array}{cc}
\vspace{1.5mm}
\dfrac{-\Sigma_{ij}}{\Sigma_{ii}\Sigma_{jj}-\Sigma_{ij}^2} & {\rm if}\ (i,j)\in\E\\
0 & {\rm if}\ (i,j)\notin\E\\
\end{array}\right.\ .
\end{align}
Furthermore, the determinant of the potential matrix is also given by
\begin{equation}
{\sf det}( {\bm J})=\prod_{i\in\V} \Sigma_{ii}^{{\sf deg}(i)-1}\prod_{\substack{(i,j)\in\E}}[\Sigma_{ii}\Sigma_{jj}-\Sigma_{ij}^2]^{-\frac{1}{2}}\ ,
\end{equation}
where ${\sf deg}(i)$ is the degree of node $i$.
\end{theorem}
\begin{figure}[t]
\centering
\includegraphics[width=3.5in]{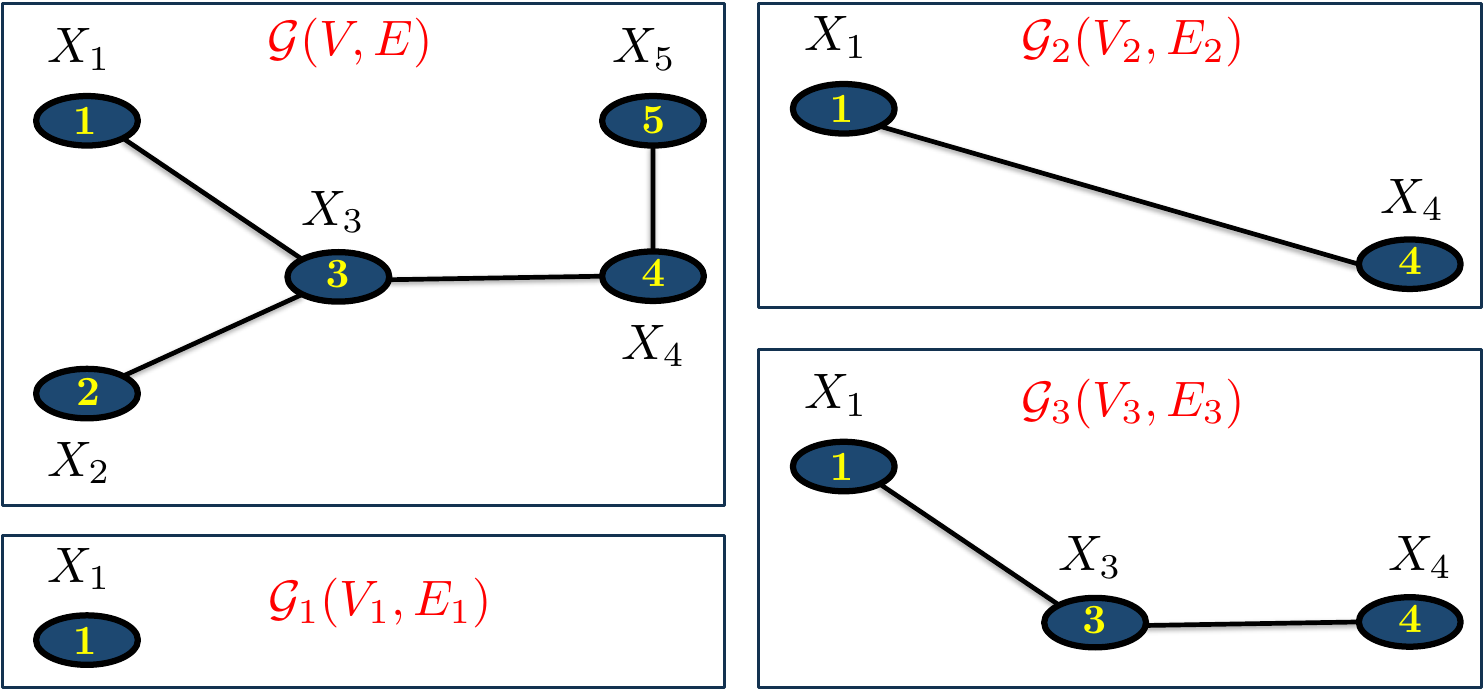}
\renewcommand{\figurename}{Fig.}
\caption{Toy example for the evolution of $\G_t(\V_t,\E_t)$ over time for $\psi^3=\{1,4,3\}$.}
\label{fig:toy}
\end{figure}
\newpage

\noindent We leverages the properties of the GMRFs to obtain closed-form expressions for the information measures defined in \eqref{eq:metric:Chernoff1} and \eqref{eq:M0}--\eqref{eq:M1}, as well as  the node selection rules characterized in Section~\ref{sec:selection}.  In order to describe the effect of the sequential sampling process on different measures that we use,  we sequentially construct the sequence of graphs $\{\G_t(\V_t,\E_t):\;  t\in\{1,\dots,\tau^*_n\}\}$ such that the graph $\G_t(\V_t,\E_t)$ at time $t$ is adapted to the nodes observed up to time $t$. Specifically, we set $\V_t=\psi^t$, and for each pair of nodes $i,j\in\V_t$ we include an edge $(i,j)\in\E_t$ if and only if either $(i,j)\in\E$, or there exists a path between nodes $i$ and $j$ in the original graph $\G(\V,\E)$ such that none of the nodes on this path has been observed up to time $t$ (except for $i$ and $j$). Figure~\ref{fig:toy} depicts a toy example on the evolution of $\G_t(\V_t,\E_t)$ over time for $t\in\{1,2,3\}$ corresponding to an underlying graph $\G(\V,\E)$. Furthermore, for any $(i,j)\in\E_t$ we define
\begin{align}
{\sf LLR}(i,j)\dff \frac{1}{2} \Big[\ln\frac{1}{1-\sigma^2_{ij}}-\frac{\sigma^2_{ij}}{1-\sigma^2_{ij}}\big({X_i^2+X_j^2}\big)+\frac{2\sigma_{ij}}{1-\sigma^2_{ij}}X_iX_j\Big]\ .
\end{align}
Under these definitions and by assuming that $\G_t(\V_t,\E_t)$ remains acyclic at time $t$, for the LLR of the samples up to time $t$ defined in~\eqref{LLR} we have
\begin{align}\label{LLR3}
&\Lambda_{t} = \sum_{i\in\V_t}\sum_{j\in\N_i^t} {\sf LLR}(i,j)\ ,
\end{align}
where $X_i$ is the sample taken from node $i$ and $\N_i^t\dff\{j\in\V_t\;:\;(i,j)\in\E_t\}$. Next, by invoking the GMRF structure, the information measures defined for the Chernoff rule in \eqref{eq:metric:Chernoff1} for any $i\in\varphi^t$ can be further simplified and expressed in terms of the correlation coefficients. Specifically, corresponding to the Chernoff rule and its associated sampling sequence $\psi_{\rm ch}^{\tau_{\rm c}}$ we have
\begin{align}\label{eq:D0_ch}
D_0^i(t)&=\frac{1}{2}\sum_{j\in\N_i^t} \Big[\ln({1-\sigma^2_{ij}})+\frac{\sigma^2_{ij}}{1-\sigma^2_{ij}}\big({X_j^2+1}\big)\Big]\ ,\\
\mbox{and}\qquad D_1^i(t)&=\frac{1}{2}\sum_{j\in\N_i^t} \Big[\ln\frac{1}{1-\sigma^2_{ij}}+\frac{\sigma^2_{ij}}{1-\sigma^2_{ij}}\big({X_j^2-1}\big)\Big]\ .
\end{align}
Furthermore, by defining
\begin{align}
\Delta_t^i \dff\{(j,k)\;:\;\ j,k\in\N_i^t\}\ ,
\end{align}
from \eqref{eq:J0} and \eqref{eq:J1} for the proposed node selection rule we have
\begin{align}\label{eq:metric:GMRF1}
{\mathscr J}_0(\{i\},\psi^{t-1})= \frac{1}{2}\sum_{j\in\N_i^t} \ln(1-\sigma^2_{ij})+\frac{1}{2}\sum_{j\in\N_i^t}\frac{\sigma^2_{ij}}{1-\sigma^2_{ij}}\Big(X_j^2+1\Big)+\sum_{(j,k)\in\Delta_t^i}{\sf LLR}(j,k) \ ,
\end{align}
and
\begin{align}\label{eq:metric:GMRF2}
{\mathscr J}_1(\{i\},\psi^{t-1}) & =\; \frac{1}{2}\sum_{j\in\N_i^t} \ln\frac{1}{1-\sigma^2_{ij}}-\frac{1}{2}\sum_{(j,k)\in\Delta_i^t} \ln\frac{1}{1-\sigma^2_{jk}}\\
& +\frac{1}{2}\Big[\sum_{j\in\N_i^t}\frac{\sigma^2_{ij}}{1-\sigma^2_{ij}}\Big(X_j^2-1\Big)+\sum_{(j,k)\in\Delta_t^i}{\sf LLR}(j,k)\Big]\times \frac{\prod_{j\in\N_i^t} (1-\sigma_{ij}^2)}{\prod_{(j,k)\in\Delta_t^i} (1-\sigma_{jk}^2)} \ .
\end{align}
Similarly, by leveraging the result of Theorem~\ref{thm:markov}, for any $\cS\in\L^i_t$ we find
\begin{align}\label{eq:metric3}
{\mathscr J}_0(\cS\backslash\{i\},\psi^{t-1}) \;=\;  & \frac{1}{2}\sum_{j\in\cS\setminus\{i\}}\Big[\ln(1-\sigma_{ij}^2)+\frac{2\sigma_{ij}^2}{1-\sigma_{ij}^2}\Big]\ ,\\
\mbox{and}\qquad {\mathscr J}_1(\cS\backslash\{i\},\psi^{t-1}) \;=\; & \frac{1}{2}\sum_{j\in\cS\setminus\{i\}}\Big[\ln\frac{1}{1-\sigma_{ij}^2}\Big]\ .
\end{align}
Subsequently, based on \eqref{eq:M_expand}, the closed-form expression for $M_\ell^i(t,\cS)$ is obtained from
\begin{align}
M_\ell^i(t,\cS)={\mathscr J}_\ell(\{i\},\psi^{t-1})\; + \; {\mathscr J}_\ell(\cS\backslash\{i\},\psi^{t-1}) \ ,\quad \mbox{for }\ell\in\{0,1\} \ .
\end{align}
These closed-form expressions of the information measures in terms of the covariance matrix entries and the dependency graph structure substantially reduces the computational complexities involved in calculating these measures from the expected values in~\eqref{eq:metric:Chernoff1} and~\eqref{eq:M0}--\eqref{eq:M1}.

\begin{figure}[b]
\centering
\includegraphics[width=3.5in]{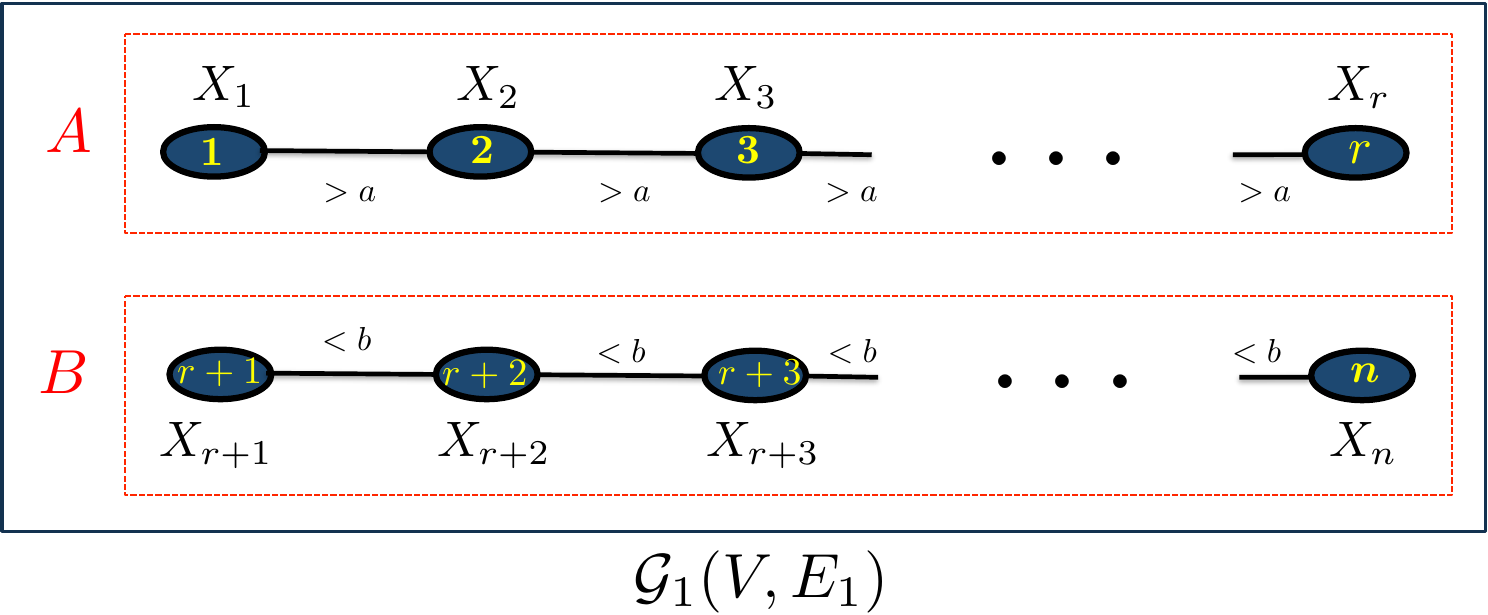}
\renewcommand{\figurename}{Fig.}
\caption{A GMRF consisting of two line graphs.}
\label{fig:Line}
\end{figure}

\subsection{Counter Example for the Optimality of Chernoff rule}
\label{sec:counter}

Building on the results for the GMRF, in this subsection, we provide an example of a heterogeneous network for which the Chernoff rule is not asymptotically optimal, and quantify the gap between its performance and our proposed rule. For this purpose, we consider a setting in which the random variables $X_\V=\{X_i\;:\; i\in\V\}$ are independent under $\H_0$, while under $\H_1$ they form a GMRF with covariance matrix $\bm\Sigma$. As depicted in Fig.~\ref{fig:Line}, the dependency graph of the GMRF consists of two disjoint line graphs corresponding to the nodes in sets $A$ and $B=\V\setminus A$. By denoting the covariance matrix of the random variables generated by sets $A$ and $B$ by $\bm\Sigma^A$ and $\bm\Sigma^B$, respectively, we assume that for any $(i,j)\in\E$ we have
\begin{align}
|\Sigma^A_{ij}|>a \ ,\quad\text{and}\quad |\Sigma^B_{ij}|<b \ ,
\end{align}
where $a>b$. This means that the random variables generated by the nodes in set $ A $ are more strongly correlated than those generated by the nodes in set $ B$. For such a network, the performance gap between the proposed rule and the Chernoff rule is established in terms of $a$ and $b$ in the following theorem.

\begin{theorem}\label{thm:per3}
Consider testing independence in~\eqref{eq:hyp:model_G}, where the GMRF consists of two disjoint line graphs corresponding to the sets of nodes in $A$ and $B$. If the correlation coefficient values between the neighbors in set $A$ are greater than $a$, while in set $B$ they are less than $b$ and $|A|=p=o(n)$, then as $n$ grows for $\ell\in\{0,1\}$
\begin{align}
\lim_{n\rightarrow\infty}  \frac{\mathbb{E}_\ell\{\tau_{\rm c}\}}{\mathbb{E}_\ell\{\tau^*_n\}}=\frac{I_\ell(A)}{I_\ell(B)}  \geq \Big(\frac{a}{b}\Big)^2 > 1\ ,
\end{align}
where $\tau_{\rm c}$ and $\tau^*_n$ are the stopping times of the strategies based on the Chernoff rule and Algorithm~1.
\end{theorem}
\begin{proof}
See Appendix \ref{App:thm:per3}.
\end{proof}
This theorem establishes that the Chernoff rule is not necessarily an asymptotically optimal sampling strategy when selection decisions are statistically dependent.

\begin{figure}[b]
\centering
\includegraphics[width=5in]{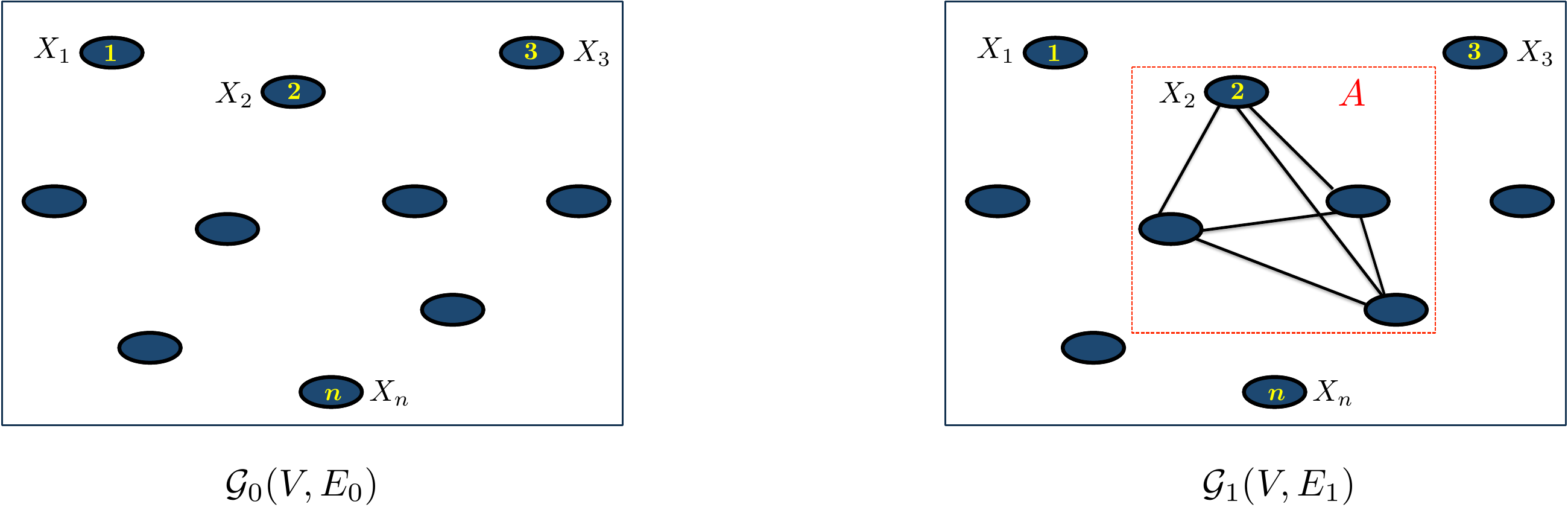}
\renewcommand{\figurename}{Fig.}
\caption{Independence versus a MRF consisting of one cluster and independent random variables.}
\label{fig:Local}
\end{figure}

\subsection{Cluster Detection}

In this subsection, we analyze cases in which the two statistical models under $\H_0$ and $\H_1$ are all similar except for a small cluster of nodes that exhibit two different correlation models. Specifically, we first consider a model in which there is a subset of nodes $B\subseteq\{1,\dots,n\}$ such that random variables $X_B\dff\{X_i:i\in B\}$ are statistically independent under both models $\H_0$ and $\H_1$. This indicates that the correlation models under $\H_0$ and $\H_1$ differ only in their distributions over the random variables from nodes $A\dff\{1,\dots,n\}\setminus B$, as depicted in Fig.~\ref{fig:Local}. Also, we assume that the random variables $X_A\dff\{X_i\;:\;i\in A\}$ form a homogeneous correlation structure, which means that observing any subsequence of the nodes in set $A$, on average provides the same amount of information. Clearly, for any set of nodes $U\subseteq B$, we have \textcolor{black}{
\begin{align}
\forall U\subseteq B\; :\quad I_\ell(U)=0\ .
\end{align}
}
In this setting, we show that  there is a constant gap between the expected stopping times of the proposed rule and the Chernoff rule. This gap stems from the fact that our proposed approach directly starts from sampling the nodes in $A$, and does not waste any sampling time by taking samples from set $B$. However, the Chernoff rule, on average, takes a number of samples from $B$ before sampling from $A$. The gap between the stopping times is formulated in the next theorem.

\begin{theorem}\label{thm:mn}
In a network of size $n$, when there exists a subset of nodes $A$ with size $p$ forming an MRF with a connected graph, while the rest of the network generate independent random variables, we have
\begin{align}
\textcolor{black}{0 \leq} \; \bbe_\ell\{\tau_{\rm c}\}-\bbe_\ell\{\tau^*_n\} =  \textcolor{black}{\Theta}\Big(\frac{n}{p}\Big)\ ,\quad \text{for}\ \ell\in\{0,1\}\ ,
\end{align}
where $\tau_{\rm c}$ and $\tau^*_n$ are the stopping times of the strategies based on the Chernoff rule and the proposed selection rule, respectively.
\end{theorem}
\begin{proof}
See Appendix \ref{App:thm:mn}.
\end{proof}
This theorem establishes the zero order asymptotic gain of the proposed strategy over the Chernoff rule in a special setting. Note that as $p$ (the size of $A$) becomes smaller, which leads to more similar and less distinguishable models under $\H_0$ and $\H_1$, the performance gap increases according to $\frac{n}{p}$. Next, we further generalize the above setting to one in which under $\sH_1$, besides $X_A$, random variables $X_B$ also form a homogeneous correlation structure (not independent anymore) with a connected dependency graph, i.e., for $\ell\in\{0,1\}$ and $\forall U\subseteq B$ we have $I_\ell(U)={I}_\ell(B)$.
This setting is depicted in Fig.~\ref{fig:TwoLocal}. If for set $A$ we have $|A|=o(n)$, then the Chernoff rule starts the sampling process from set $B$ almost surely, and it remains in set $B$ until it exhausts all the nodes of $B$, while the proposed rule always identifies the most informative nodes to take the sample. The following theorem characterizes the performance gap between the Chernoff and the proposed rule in this setting.

\begin{figure}[b]
\centering
\includegraphics[width=5in]{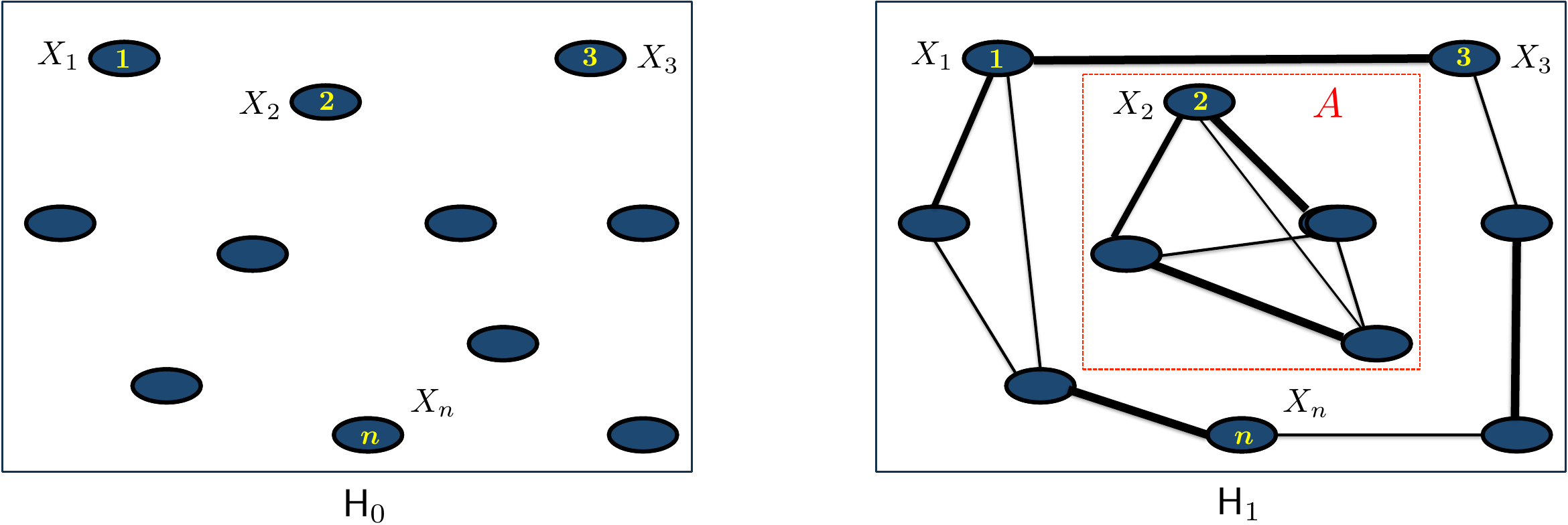}
\renewcommand{\figurename}{Fig.}
\caption{Independence versus a MRF consisting of two clusters.}
\label{fig:TwoLocal}
\end{figure}

\begin{theorem}\label{thm:per2}
Consider a network of size $n$ partitioned into sets $A$ and $B$ specified in Fig.~\ref{fig:TwoLocal}. In the asymptote of large $n$, if the dependency graph of the nodes in both $A$ and $B$ are connected and $|A|=o(n)$, then 
\begin{align}
\textcolor{black}{\lim_{n\rightarrow\infty}}  \frac{\bbe_0\{\tau_{\rm c}\}}{\bbe_0\{\tau_n^*\}}  =\frac{\max\{I_0(A),I_0(B)\}}{I_0(B)} \ ,
\qquad\text{and }\qquad \textcolor{black}{\lim_{n\rightarrow\infty}}  \frac{\bbe_1\{\tau_{\rm c}\}}{\bbe_1\{\tau_n^*\}} =\frac{\max\{I_1(A),I_1(B)\}}{I_1(B)} \ .
\end{align}
\end{theorem}
\begin{proof}
When $p=o(n)$ the Chernoff rule starts the sampling process from set $B$ with probability $1$. By invoking the results of Theorem~\ref{thm:asymp2},~Corollary~\ref{thm:opt:L}, and~Theorem~\ref{thm:opt:U} we  conclude the proof.
\end{proof}

According to the theorem above, when the size of $A$ is sufficiently small such that most of the time, the Chernoff rule starts the sampling process from set $B$, the Chernoff loses its first-order asymptotic optimality property, as shown in the counterexample in Section~\ref{sec:counter}. The settings discussed in this subsection highlight the advantages of the proposed selection rule by quantifying two main gains; the gain of selecting the best node at the beginning of the sampling process, and the gain obtained from freely navigating throughout the entire network by jumping over subgraphs in order to find the most informative nodes. Although these settings are special cases, the gain of the proposed rule for a general network is a combination of these two gains.

\section{Numerical Evaluations}
\label{sec:simulations}

In this section, we evaluate the performance of the proposed sampling strategy by comparing it with the existing approaches through simulations. \textcolor{black}{First, we examine the $(\alpha,\beta)$-accuracy conditions. We consider Gaussian distributions ${\cal N}({\boldsymbol \theta},\bSigma_0)$ and ${\cal N}({\boldsymbol \theta},\bSigma_1)$ under models $\H_0$ and $\H_1$, respectively. The covariance matrices $\Sigma_0$ and $\Sigma_1$ have all their diagonal elements equal to 1, and the off-diagonal elements randomly take values in the range $[-1,1]$, such that the overall combinations constitute valid covariance matrices. Figure~\ref{fig:feasible} shows the variations of the lower bound on the $(\alpha,\beta)$-accuracy probability established in Theorem~\ref{thm:nonasymp} with respect to increasing network size $n$ for four different levels of reliability constraints. It is observed that for reliabilities as small as $10^{-8}$, $(\alpha,\beta)$-accuracy is guaranteed almost surely when the network size is as small as 100 nodes. We remark that for each reliability level, we evaluate two distinct settings where in one the covariance matrices $\Sigma_0$ and $\Sigma_1$ are generated completely randomly (solid curves) and in the other settings half of the $n$ Gaussian random variables, i.e., $\{X_1,\dots,X_{\frac{n}{2}}\}$ have the same joint distribution (dashed curves).}

\begin{figure}[h]
\centering
\includegraphics[width=3.0in]{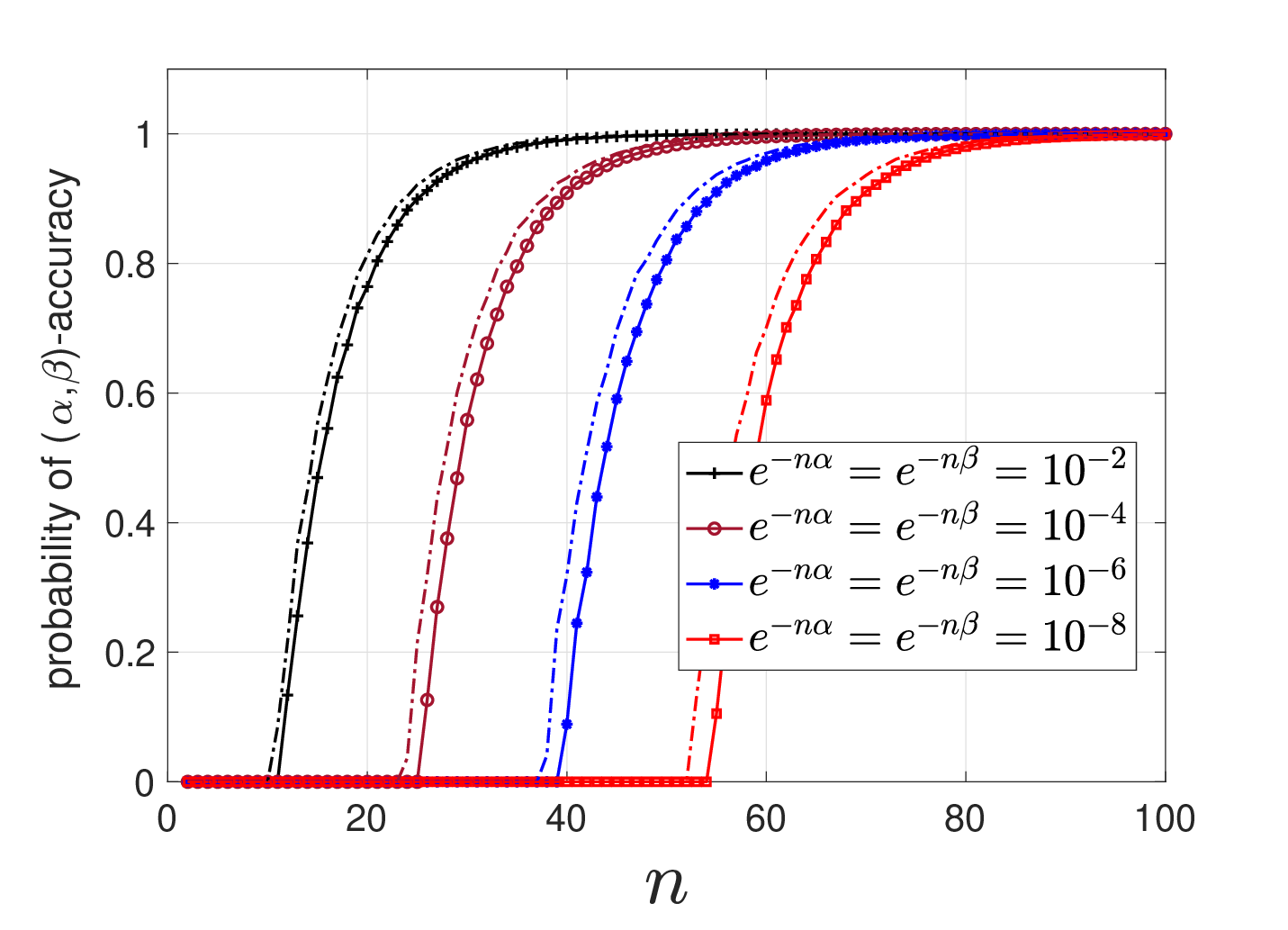}
\renewcommand{\figurename}{Fig.}\vspace{-.2 in}
\caption{Lower bound on the probability of $(\alpha,\beta)$-accuracy.}
\label{fig:feasible}
\end{figure}

\begin{figure}[t]
\begin{minipage}[b]{0.4\linewidth}
\centering
\includegraphics[width=3.0in, height=2.0in]{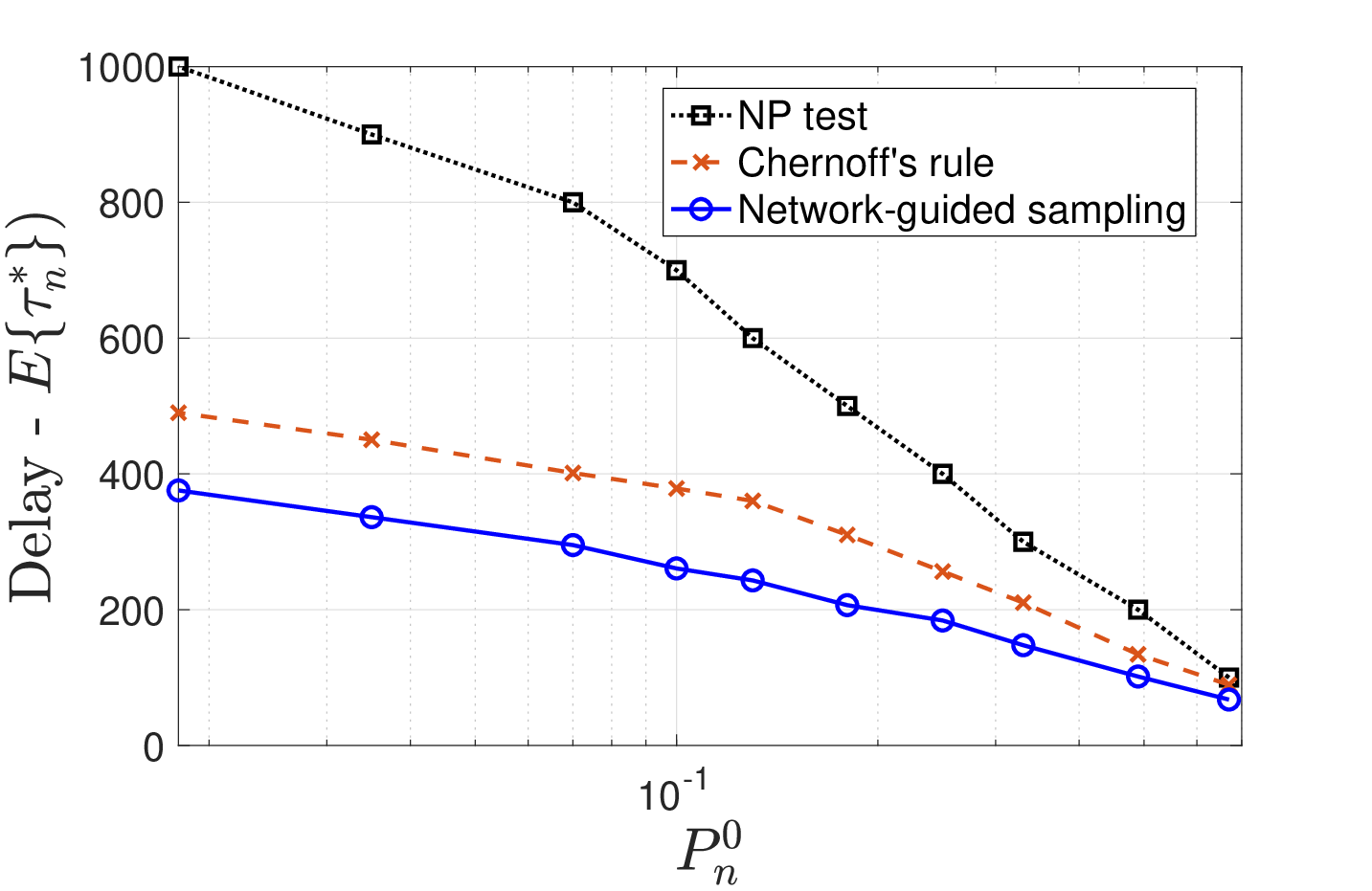}
\renewcommand{\figurename}{Fig.}\vspace{-.2 in}
\caption{Average delay versus error probability in a homogeneous network.}
\label{fig:Comp}
\end{minipage}
\hspace{.1\linewidth}
\begin{minipage}[b]{0.4\linewidth}
\centering
\includegraphics[width=3.0in, height=2.0in]{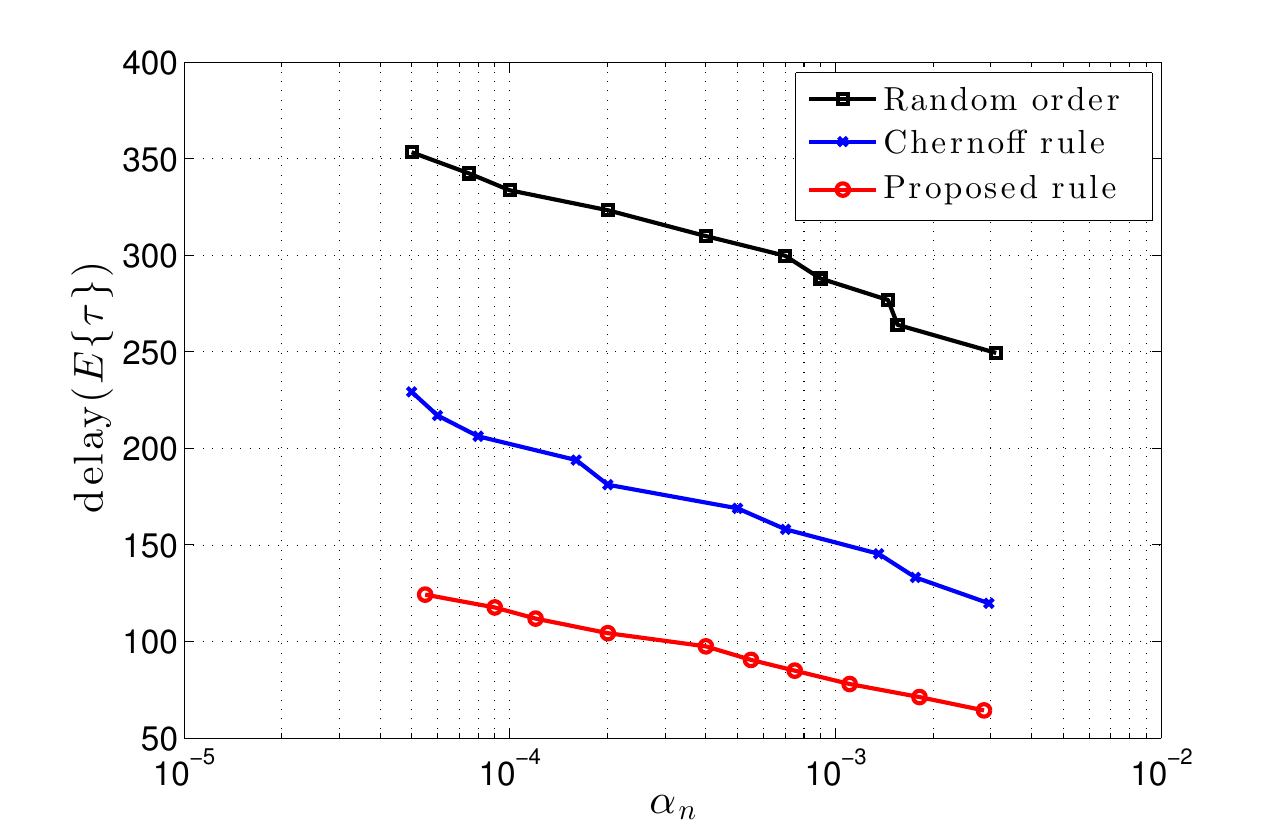}
\renewcommand{\figurename}{Fig.}\vspace{-.2 in}
\caption{Average delay versus error probability in a heterogeneous network.}
\label{fig:delay}
\end{minipage}
\end{figure}

\textcolor{black}{For the rest of the numerical evaluations and simulations,} we use the NP test as the fixed sample-size approach, and for the sequential sampling, we consider random (non-adaptive) sampling order and the Chernoff rule. We consider zero-mean Gaussian distributions for data, and test covariance matrix under $\H_1$ versus $\mathbf{I}_n$ under $\H_0$. We also set $\epsilon_0=\epsilon_1=0.5$. \textcolor{black}{As the first comparison}, we consider the nearest neighbor dependency graph for uniformly distributed nodes in a two-dimensional space, for which the cross-covariance value between two nearest neighbors is a function of their distance. We denote the distance between nodes $i$ and $j$ by $R_{ij}$ and set the correlation coefficient between nodes $i$ and $j$ to $\Sigma_{ij}=M{\e}^{-aR_{ij}}$, where $a,M\in\mathbb{R}_+$. Under $\H_0$ we set $M=0$, which corresponds to independent samples. Under $\H_1$ as $M$ increases the KL divergence between the distributions corresponding to $f_0$ and $f_1$ grows. In Fig.~\ref{fig:Comp}, we set $M=0.1$, $a=0$, and  ${\beta}_n={\e}^{-n{\beta}}=0.1$ and compare the performance of different approaches. To this end, for different values of $n$ we find $\mathsf{P}_n^{0}$ associated with the NP test (i.e., the false alarm probability), based on which we design the sequential sampling strategy for the Chernoff and proposed selection rules and find the average delay. It is observed that the proposed sampling procedure outperforms both the NP test and the Chernoff rule in terms of the reliability-agility trade-off. We also compare the performance of the proposed strategy with that of the Chernoff rule and the random selection rule in a heterogeneous network. For this purpose, we generate a subgraph with three nodes and two edges, in which the cross-covariance values between the neighbors are $0.5$ and $0.1$. We use $500$ copies of this subgraph as the building block for a network consisting of $1500$ nodes. For such a network, the optimal rule is to select the nodes with larger cross-covariance values. Figure~\ref{fig:delay} demonstrates the average delay before reaching a confident decision for different target accuracies and the selection rules when  ${\alpha}={\beta}$. By comparing Fig.~\ref{fig:Comp} and Fig.~\ref{fig:delay}, it is observed that in heterogeneous networks, the proposed strategy improves significantly compared to the Chernoff rule. The reason is the larger discrepancy in the amount of information gained from different nodes.

\begin{figure}[h]
\begin{minipage}[b]{0.4\linewidth}
\centering
\includegraphics[width=3.0in, height=2 in]{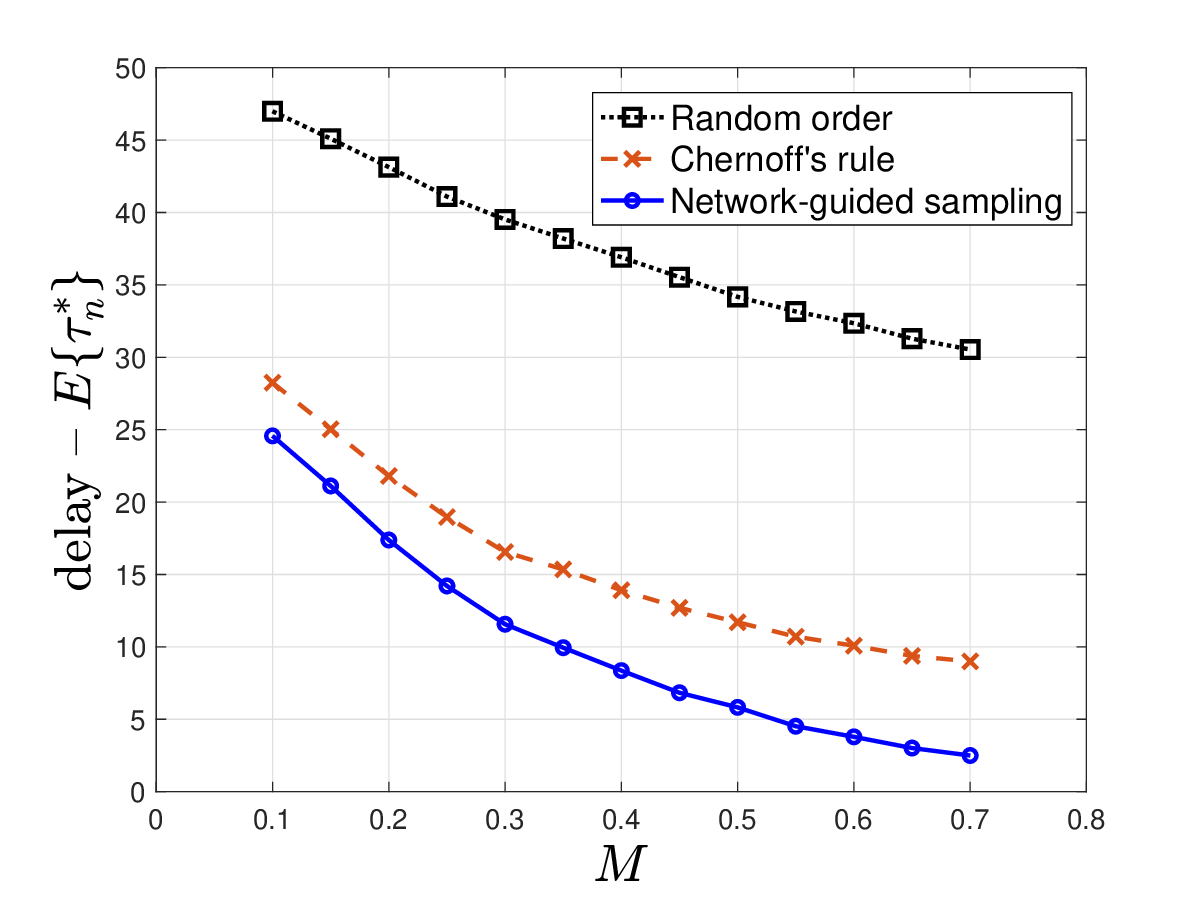}
\renewcommand{\figurename}{Fig.}\vspace{-.2 in}
\caption{Average delay versus $M$.}
\label{fig:DelayVsM}
\end{minipage}
\hspace{.1\linewidth}
\begin{minipage}[b]{0.4\linewidth}
\centering
\includegraphics[width=3.0in, height=2.0in]{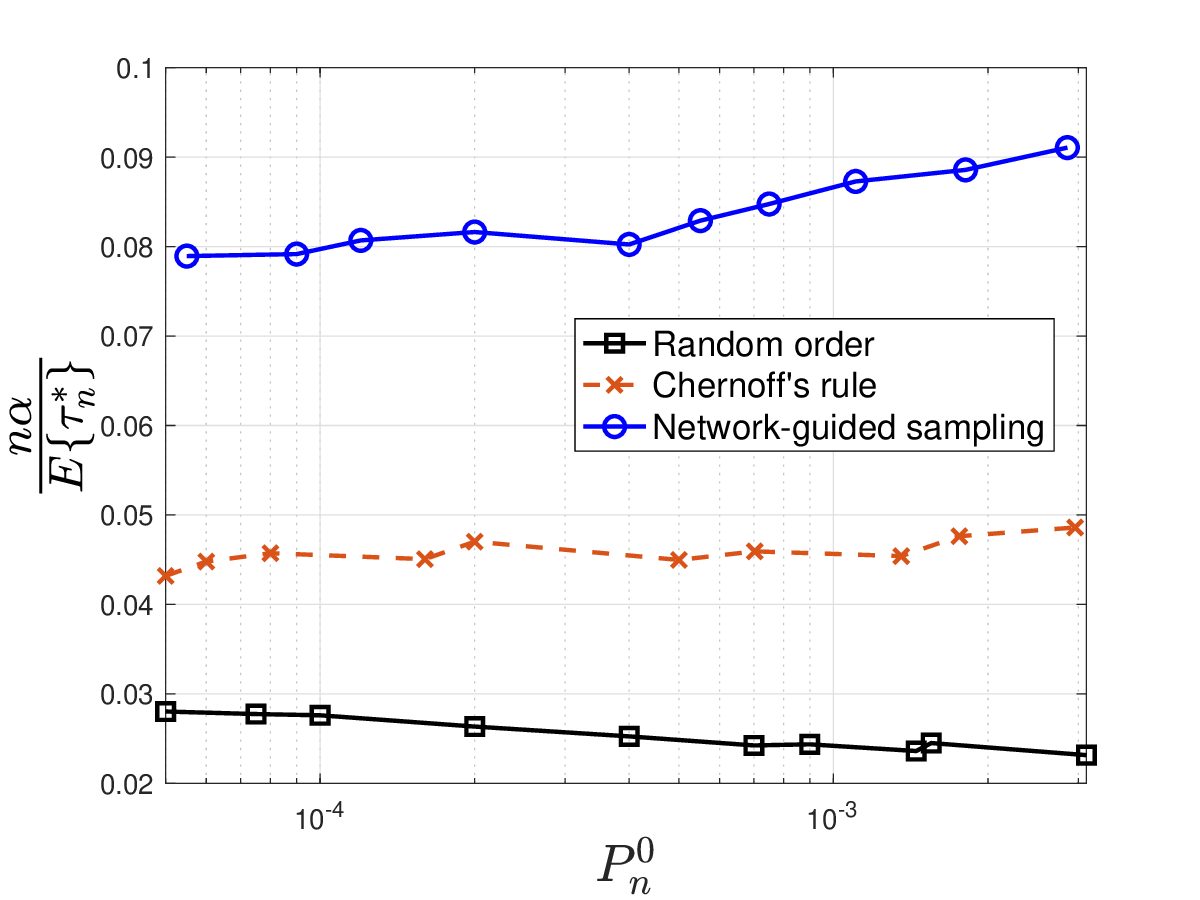}
\renewcommand{\figurename}{Fig.}\vspace{-.2 in}
\caption{Error exponent vs. error probability.}
\label{fig:EE}
\end{minipage}
\end{figure}

In order to compare the performance of different selection rules for different levels of correlation strength, Fig.~\ref{fig:DelayVsM}  compares the average delays incurred by the proposed approach, the Chernoff rule and a random selection rule for different values of $M$ when $n=1000$, ${\alpha}={\beta}=1.6\times 10^{-3}$, and $a=1$. It is observed that both the Chernoff rule and the proposed approach outperform the random selection rule, and as the KL divergence grows by increasing $M$ the improvement is more significant.  Furthermore, in Fig.~\ref{fig:EE} the error exponents are compared where it is observed that the proposed strategy has an error exponent twice as large as that of the Chernoff rule and both of them outperform the strategy based on a random selection of nodes.




\begin{figure}[h]
\begin{minipage}[b]{0.4\linewidth}
\centering
\includegraphics[width=3.0in, height=2.0in]{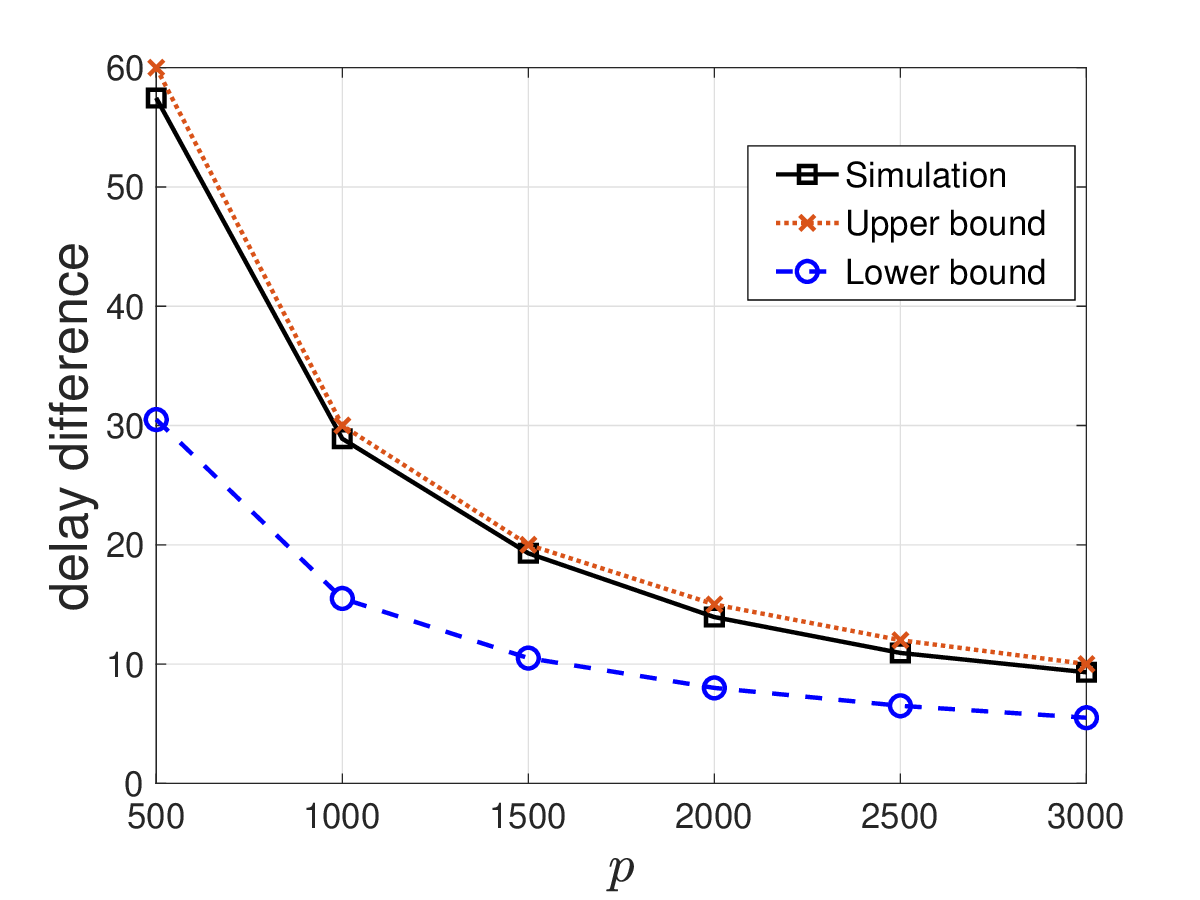}
\renewcommand{\figurename}{Fig.}\vspace{-.2 in}
\caption{The average delay difference between  Chernoff's rule and tnetwork-guided active sampling.}
\label{fig:n_over_m}
\end{minipage}
\hspace{.1\linewidth}
\begin{minipage}[b]{0.4\linewidth}
\centering
\includegraphics[width=3.0in, height=2.0in]{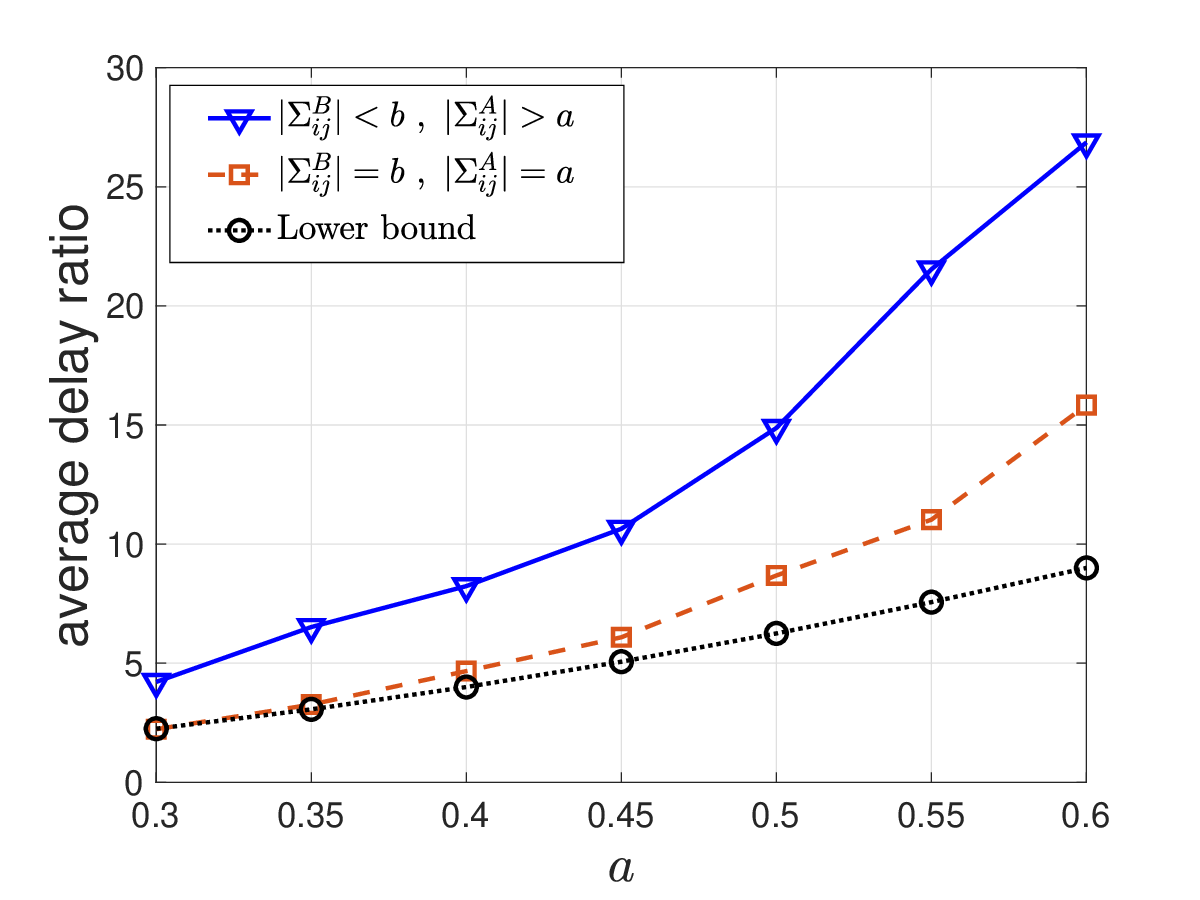}
\renewcommand{\figurename}{Fig.}\vspace{-.2 in}
\caption{The ratio of the expected delay of Chernoff's rule and network-guided active sampling.}
\label{fig:a_over_b}
\end{minipage}
\end{figure}

In order to verify the results of Theorem~\ref{thm:mn}, we consider a network with $n=30000$ nodes, in which only a subset $A$ consisting of $p$ nodes generate correlated random variables under one of the two hypotheses, while the random variables generated by the rest of the nodes are independent under both hypotheses. Figure~\ref{fig:n_over_m} demonstrates the average delay of the Chernoff rule in taking its first sample from set $A$. The upper bound and lower bound obtained in Theorem~\ref{thm:mn} are also shown for comparison. It is observed that the delay difference is always between the obtained bounds, which confirms that it is $\Theta(\frac{n}{p})$.

Finally, we consider a network with $10000$ nodes, from which $50$ nodes, denoted by set $A$, are strongly correlated, i.e., the cross-covariance values between the neighbors in set $A$, denoted by $\Sigma_{ij}^A$, are greater than a constant $a\in(0.3,0.6)$, while the rest of the nodes, denoted by set $B$, also form a connected graph with cross-covariance values $\Sigma_{ij}^B$ less than a constant $b=0.2$. In Fig.~\ref{fig:a_over_b} the ratio between the average delay of the proposed sampling strategy is compared with the lower bound $(\frac{a}{b})^2$ obtained in Theorem~\ref{thm:per3} for different values of $a$. We also include the ratio between the average delays for the setting in which the cross-covariance values in sets $A$ and $B$ are equal to $a$ and $b$, respectively, for which it is observed that the lower bound is tighter.


\section{Conclusion}

We have considered the quickest detection of a correlation structure in a Markov network, with the objective of determining the true model governing the samples generated by different nodes in the network. After discussing the widely used Chernoff rule and its shortcomings, we have designed a sequential and data-adaptive sampling strategy to determine the true correlation structure with the fewest average number of samples while, in parallel, the final decision is controlled to meet target reliability. The proposed sampling strategy, which judiciously incorporates the network's correlation structure into its decision rules, involves dynamically deciding whether to terminate the sampling process, or to continue collecting further evidence, and prior to terminating the process which node to observe at each time. We have established the optimality properties of the proposed sampling strategy and leveraged the Markov properties of the network to reduce the computational complexities involved in the implementation of the proposed approach. We have provided an example for which the Chernoff rule is not optimal. Finally, we have quantified the advantages of the proposed rule over the Chernoff rule for some special cases.

\appendix

\section*{Notations throughout the Proofs}
For the convenience in notations, throughout the proofs we drop the necessary subscript $n$ in $\varphi_n^t$, $\psi_n^t$, $\gamma_n^{\rm L} $, $\gamma_n^{\rm U} $.

\section{Proof of Theorem \ref{thm:nonasymp}} 
\label{App:thm:nonasymp}

\textcolor{black}{We start by showing that if
\begin{align}\label{eq:exit2}
\exists t\in \{1,\dots,n\}\quad \mbox{\rm such that} \quad  \Lambda_t\notin (\gamma_n^{\rm L} ,\gamma_n^{\rm U} )\ ,
\end{align}
then any sequential decision algorithm with the stopping rule $\tau^*_n$ and the detection rule $\delta^*_n$ specified in~\eqref{eq:stop}~and~\eqref{eq:SPRT}, respectively, is $(\alpha,\beta)$-accurate. Given the structure of the stopping time, according to which the sampling process terminates as soon as $\Lambda_t$ exits the band $(\gamma^{\rm L},\gamma^{\rm U})$, the assumption in~\eqref{eq:exit2} is equivalent to having
\begin{align}\label{eq:exit3}
\Lambda_{\tau^*}\notin(\gamma^{\rm L},\gamma^{\rm U})\ .
\end{align}
Therefore, for $\mathsf{P}^{0}_n$ we have
\begin{align}
\mathsf{P}^{0}_n\, &=\,\mathbb{P}_0(\d^*_n=1)  \\
&=\,\sum_{k=1}^n\mathbb{P}_0(\d^*_n=1,\tau^*_n=k)  \\
\label{eq:feas3} & \overset{\eqref{eq:SPRT},\; \eqref{eq:exit3}}{=}\,\sum_{k=1}^n\mathbb{P}_0(\Lambda_{\tau^*_n} \geq \gamma^{\rm U},\tau^*_n=k)  \\
\label{eq:feas} &\overset{\eqref{eq:exit3}}{=} \,\sum_{k=1}^n\mathbb{P}_0(\Lambda_{\tau^*_n}\geq\gamma^{\rm U},{\tau}^*_n=k)  \\
&=\,\sum_{k=1}^n\int_{(\Lambda_k\geq\gamma^{\rm U},{\tau^*_n}=k)} f_0(Y^k;\psi^k)\; \der Y^k  \\
\label{eq:feas4} &\overset{\eqref{LLR}}{=}\,\sum_{k=1}^n\int_{(\Lambda_k\geq\gamma^{\rm U},{\tau}^*_n=k)} \exp(-\Lambda_k)f_1(Y^k;\psi^k)\; \der Y^k  \\
\label{eq:st}
&{\leq}\,\sum_{k=1}^n\int_{(\Lambda_k\geq\gamma^{\rm U},{\tau}^*_n=k)}\exp(-\gamma^{\rm U})f_1(Y^k;\psi^k)\; \der Y^k  \\
&\overset{\eqref{eq:th1}}{=} \,{\e}^{-n{\alpha}}\sum_{k=1}^n\int_{(\Lambda_k\geq\gamma^{\rm U},{\tau^*_n}=k)}f_1(Y^k;\psi^k)\; \der Y^k  \\
\label{eq:st2}&=\,{\e}^{-n{\alpha}}\sum_{k=1}^n\mathbb{P}_1(\delta^*_n=1,{\tau^*_n}=k)  \\
&=\,{\e}^{-n{\alpha}}\cdot\mathbb{P}_1(\delta^*_n=1)  \\
&\leq\,{\e}^{-n{\alpha}} \ ,
\end{align}
where \eqref{eq:feas3} holds according to the definition of the terminal decision rule in~\eqref{eq:SPRT},~\eqref{eq:feas} holds due to the assumption in~\eqref{eq:exit3}, \eqref{eq:feas4} holds due to the definition of LLR in~\eqref{LLR}, and~\eqref{eq:st} holds due to the structure of the region over which the integral is computed. Finally~\eqref{eq:st2} holds by noting that the decision rule $\delta^*_n=1$ specifies that $\Lambda_{\tau^*_n}>0$, which by taking into account~\eqref{eq:exit3} and the fact that $\gamma^{\rm L}<0$, becomes equivalent to $\Lambda_{\tau^*_n}\geq\gamma^{\rm U}$.
By following the same line of argument for $\mathsf{P}^{1}_n$ we obtain
\begin{align}
\mathsf{P}^{1}_n\, & \leq\, {\e}^{\gamma^{\rm L}}\cdot\mathbb{P}_1(\delta^*_n=0)  =\, {\e}^{-n{\beta}}\cdot\mathbb{P}_1(\delta^*_n=0) \leq\, {\e}^{-n{\beta}}  \ .
\end{align}}
\noindent \textcolor{black}{Next, we analyze the likelihood of the condition in~\eqref{eq:exit2} being valid, which establishes a probabilistic guarantee for Algorithm~1 generating $(\alpha,\beta)$-accurate solutions to $\Prob$. 
\begin{align}
1- \P\left(\exists t\in \{1,\dots,n\}\quad {\rm s.t.} \quad  \Lambda_t\notin (\gamma^{\rm L},\gamma^{\rm U})\right) & = \P\left(\Lambda_t\in (\gamma^{\rm L},\gamma^{\rm U})\ ,\quad \forall t\in \{1,\dots,n\}\right) \\
& \leq \P\left( \Lambda_n\in (\gamma^{\rm L},\gamma^{\rm U})\right) \\
\label{Up} & = \sum_{i=0}^ 1 \; \epsilon_i \P_i\left( \Lambda_n\in (\gamma^{\rm L},\gamma^{\rm U})\right)
\end{align}
Next, for the  probability terms in the right hand side we have
\begin{align}
\mathbb{P}_0(\Lambda_n\in(\gamma^{\rm L},\gamma^{\rm U}))\;&{\leq}\; \mathbb{P}_0(\Lambda_n>\gamma^{\rm L})\\
\label{P01}
&{\leq}\; \frac{1}{\sqrt{\exp(\gamma^{\rm L})}}\cdot \mathbb{E}_0\{\sqrt{\exp(\Lambda_n)}\}\\
& \overset{\eqref{eq:bc}}{=} \frac{1}{\sqrt{\exp(\gamma^{\rm L})}}\cdot {\sf B}_n(f_0,f_1) \\
\label{P02}
& \overset{\eqref{eq:th1}}{=} \exp\left({\frac{n{\beta}}{2}}\right)\cdot {\sf B}_n(f_0,f_1) \ ,
\end{align}
where~ \eqref{P01} follows the Markov inequality. By following a similar line of argument we obtain
\begin{align}\label{P03}
\mathbb{P}_1(\Lambda_n\in(\gamma^{\rm L},\gamma^{\rm U}))&\leq \exp\left({\frac{n{\alpha}}{2}}\right)\cdot {\sf B}_n(f_0,f_1) \ .
\end{align}
Hence, from~\eqref{Up}, \eqref{P02}, and \eqref{P03} we obtain
\begin{align}
\P\left(\exists t\in \{1,\dots,n\}\quad {\rm s.t.} \quad  \Lambda_t\notin (\gamma^{\rm L},\gamma^{\rm U})\right) \; \geq \; 1-{\sf B}_n(f_0,f_1)\left[\epsilon_0 \exp\left({\frac{n{\beta}}{2}}\right)+\epsilon_1 \exp\left({\frac{n{\alpha}}{2}}\right)\right]\ .
\end{align}}

\section{Proof of Lemma~\ref{lemma:kappa}}
\label{app:lemma:kappa}
\textcolor{black}{
From the definition of $\kappa(f_0,f_1)$ in \eqref{eq:bd} we have
\begin{align}
2\kappa (f_0,f_1) & =   - \lim_{n\rightarrow\infty} \frac{2}{n}\ln{\sf B}_n(f_0,f_1) \\
& \overset{\eqref{eq:bc}}{=} - \lim_{n\rightarrow\infty} \frac{2}{n}\ln\int \sqrt{f_0(x;\V)f_1(x;\V)}\; \der x \\
& = - \lim_{n\rightarrow\infty} \frac{2}{n}\ln\int \sqrt{\frac{f_0(x;\V)}{f_1(x;\V)}} \; f_1(x;\V)\; \der x\\
\label{eq:b3} & \leq  -\lim_{n\rightarrow\infty} \frac{2}{n}\int \ln\left(\sqrt{\frac{f_0(x;\V)}{f_1(x;\V)}}\right) f_1(x;\V)\; \der x\\
& =  \lim_{n\rightarrow\infty} \int \frac{1}{n}\ln\left({\frac{f_1(x;\V)}{f_0(x;\V)}}\right) f_1(x;\V)\; \der x\\
\label{eq:b4}  & \overset{\eqref{eq:cc2}}{=}  
\lim_{n\rightarrow\infty}  \bbe_1\left[{\sf nLLR}_1(X_V;V)\right]\ ,
\end{align}
where \eqref{eq:b3} holds due to Jensen's inequality. By definition, in a homogeneous network, when the limit exists, the term ${\sf nLLR}_1(X_V;V)$ converges completely to $I_1$. This, in turn, implies that $ \bbe_1\left[{\sf nLLR}_1(X_V;V)\right]$ also converges completely to $I_1$, which, subsequently, converges almost surely to $I_1$ (complete convergence implies almost sure convergence~\cite{Karr:1993}). Hence, in homogeneous networks
\begin{align}
2\kappa (f_0,f_1) \leq \lim_{n\rightarrow\infty}  \bbe_1\left[{\sf nLLR}_1(X_V;V)\right] \xrightarrow{\rm a.s.} I_1\ .
\end{align}
For the heterogeneous networks, we will follow the same line of argument to show that 
\begin{align}
2\kappa (f_0,f_1) \leq \lim_{n\rightarrow\infty}  \bbe_1\left[{\sf nLLR}_1(X_V;V)\right] \xrightarrow{\rm a.s.} I_1(V) \overset{\eqref{eq:I0*}}\leq I^*_1\ .
\end{align}
A similar line of argument also shows that almost surely $2\kappa (f_0,f_1) \leq I_0$ and $2\kappa (f_0,f_1) \leq I_0^*$ in  homogeneous and heterogeneous networks.  By noting the assumption $\max\{{\alpha},{\beta}\}\leq 2\kappa(f_0,f_1)$, the desired conclusion is established.
}

\section{Proof of Theorem \ref{thm:asymp1}} 
\label{App:thm:asymp1}

\textcolor{black}{
In order to prove~\eqref{eq:lower1}, we show that for any feasible solution to~\eqref{eq:Opt} and for all $\epsilon>0$ we have
\begin{align}
\label{eq:p1}
\lim_{n\rightarrow \infty} \mathbb{P}_1\left(\frac{\tau_n}{n}>\frac{{{\alpha}}}{{I}_1+\epsilon}\right)=1\ .
\end{align}
This property, in turn, establishes the desired result in~\eqref{eq:lower1}. Specifically, by applying the Markov inequality we obtain
\begin{align}\label{eq:p2}
\lim_{n\rightarrow \infty} \mathbb{E}_1\left\{\frac{\tau_n}{n}\cdot \frac{{I}_1}{{\alpha}} \right\} & \geq \lim_{n\rightarrow \infty}  \frac{I_1}{I_1+\epsilon}\cdot\mathbb{P}_1\left(\frac{\tau_n}{n}\cdot \frac{{I}_1}{{\alpha}}> \frac{I_1}{I_1+\epsilon} \right)\overset{\eqref{eq:p1}}{=}\frac{I_1}{I_1+\epsilon}\ , \qquad \forall \epsilon > 0\ . 
\end{align}
Since the inequality in~\eqref{eq:p2} is valid for all $\epsilon>0$ we have
\begin{align}\label{eq:p3}
\lim_{n\rightarrow \infty} \mathbb{E}_1\left\{\frac{\tau_n}{n}\cdot \frac{{I}_1}{{\alpha}} \right\} & \geq \sup_{\epsilon>0} \; \frac{I_1}{I_1+\epsilon} = 1\ ,
\end{align}
which concludes~\eqref{eq:lower1}. To prove~\eqref{eq:p1}, for $i\in\{0,1\}$ and $L\in\{2,\dots, n-1\}$, and corresponding to any $(\alpha,\beta)$-accurate algorithm  with stopping time $\tau_n$ and decision rule $\delta_n$ let us define the event
\begin{align}
\A(i,L)\dff\{{\delta_n}=i\, ,\, {\tau_n}\leq L\} \ .
\end{align}
Then, for any $\zeta>0$, \textcolor{black}{for the error probability term $\mathsf{P}_n^{0}$ when the stopping time is  $\tau_n$ and the decision rule is $\delta_n$,} we have
\begin{align}
\label{eq:fa00}\mathsf{P}^{0}_n & = \mathbb{P}_0(\delta_n=1)  \\
\addtocounter{equation}{1}
& = \mathbb{E}_0\{\mathds{1}_{\{\textcolor{black}{\delta_n}=1\}}\} \\
\label{eq:fa011} & = \mathbb{E}_1\{\mathds{1}_{\{\textcolor{black}{\delta_n}=1\}}\exp(-\Lambda_{{\tau_n}})\} \\
\label{eq:fa0} & \geq \mathbb{E}_1\{\mathds{1}_{\{\A(1,L),\Lambda_{{\tau_n}}<\zeta\}}\exp(-\Lambda_{{\tau_n}})\} \\
& \geq {\e}^{-\zeta}\;\mathbb{P}_1(\A(1,L),\Lambda_{{\tau_n}}<\zeta) \\
& \geq {\e}^{-\zeta}\; \mathbb{P}_1\Big(\A(1,L) \; , \;  \sup_{t<L}\Lambda_t<\zeta\Big) \\
\label{eq:fa1}
& \geq {\e}^{-\zeta}\; \Big[\mathbb{P}_1(\A(1,L))-\mathbb{P}_1\big(\sup_{t<L}\Lambda_t\geq \zeta \big)\Big] \\
\label{eq:fa2}
& \geq {\e}^{-\zeta}\; \Big[\mathbb{P}_1(\textcolor{black}{\delta_n}=1)-\mathbb{P}_1({\tau_n}>L)-\mathbb{P}_1\big(\sup_{t<L}\Lambda_t\geq \zeta\big)\Big]  \ ,
\end{align}
\textcolor{black}{where~\eqref{eq:fa011} holds by changing the probability measure,~\eqref{eq:fa0} holds by noting that the event $\{\A(1,L),\Lambda_{{\tau_n}}<\zeta\}$ is a subset of the event $\{\delta_n=1\}$,} and~\eqref{eq:fa1} and~\eqref{eq:fa2} hold due to basic set operations properties. By rearranging the terms in~\eqref{eq:fa00}  and \eqref{eq:fa2} and invoking $\mathbb{P}_0(\delta_n=1)\leq {\e}^{-n{\alpha}}$ and $\mathbb{P}_1(\delta_n=0)\leq {\e}^{-n{\beta}}$ (the decision rules are $(\alpha,\beta)$-accurate) we obtain 
\begin{align}
\label{eq:p1L_1} \mathbb{P}_1({\tau_n}>L) &\; \geq \;  \mathbb{P}_1(\textcolor{black}{\delta_n}=1)-e^\zeta\; \mathbb{P}_0(\textcolor{black}{\delta_n}=1)-\mathbb{P}_1\big(\sup_{t<L}\Lambda_t\geq \zeta \big)  \\
& =1-  \mathsf{P}_n^{1} -e^\zeta\; \mathsf{P}_n^{0}-\mathbb{P}_1\big(\sup_{t<L}\Lambda_t\geq \zeta \big)  \\
\label{eq:p1L_2} & \; \overset{\eqref{eq:Opt}}{\geq} \;  1-{\e}^{-n{\beta}}-e^\zeta\;{\e}^{-n{\alpha}}-\mathbb{P}_1\big(\sup_{t<L}\Lambda_t\geq \zeta \big)\ .
\end{align}
Note that~\eqref{eq:p1L_2} holds for any $\zeta>0$. Next, we set $\zeta\dff cLI_1$ where
\begin{align}\label{eq:c}
c \triangleq 1+\frac{\epsilon}{2I_1}\ .
\end{align}
Hence, for any $K\in\{2,\dots,L-1\}$ for the last term in~\eqref{eq:p1L_2} we have
\begin{align}
\label{eq:zeta1}
\mathbb{P}_1\Big(\sup_{t<L}\Lambda_t\geq \zeta\Big) 
&=\mathbb{P}_1\Big(\sup_{t<L}\Lambda_t\geq cLI_1\Big) \\
&\leq\mathbb{P}_1\Big(\sup_{t< K}\Lambda_t+\sup_{K\leq t<L}\Lambda_t\geq cLI_1\Big) \\
&\leq\mathbb{P}_1\Big(\sup_{t < K}\Lambda_t+\sup_{K\leq t<L}\Big\{\frac{L}{t}\Lambda_t\Big\}\geq cLI_1\Big) \\
& = \mathbb{P}_1\Big(\frac{1}{L}\sup_{t < K}\Lambda_t+\sup_{K\leq t<L}\Big\{\frac{\Lambda_t}{t}-I_1\Big\}\geq (c-1)I_1\Big) \\
&\leq\mathbb{P}_1\Big(\frac{1}{L}\sup_{t < K}\Lambda_t+\sup_{t\geq K}\Big|\frac{\Lambda_t}{t}-I_1\Big|\geq (c-1)I_1\Big)\\
& \overset{\eqref{eq:c}}{=}  \mathbb{P}_1\Big(\frac{1}{L}\sup_{t < K}\Lambda_t + \sup_{t\geq K}\Big|\frac{\Lambda_t}{t}-I_1\Big| \geq  \frac{\epsilon}{2} \Big)\\
& \label{eq:zeta2} \leq  \mathbb{P}_1\Big(\frac{1}{L}\sup_{t < K}\Lambda_t \geq  \frac{\epsilon}{4} \Big)+\P\Big(\sup_{t>K}\Big|\frac{\Lambda_{t}}{t}-I_1\Big|>\frac{\epsilon}{4}\Big)\ .
\end{align}
We show that both probability terms in~\eqref{eq:zeta2} diminish as $n$ grows. Fro the second term in~\eqref{eq:zeta2} note that from the definition of $T_\ell(h,\psi^\infty)$ in \eqref{eq:T_l} we know that corresponding to any given sampling path $\psi^\infty$ we have
\begin{align}
 \forall t \; \geq T_\ell\left(\frac{\epsilon}{4},\psi^\infty\right): \qquad \Big|\frac{\Lambda_{t}}{t}-I_1\Big|\leq \frac{\epsilon}{4}  \qquad \ .
\end{align}
This indicates that by setting $K=T_\ell(\frac{\epsilon}{4},\psi^\infty)$, it can be readily verified that
\begin{align}\label{eq:part1}
\lim_{n\rightarrow\infty} \P\Big(\sup_{t>K}\Big|\frac{\Lambda_{t}}{t}-I_1\Big|> \frac{\epsilon}{4}\Big)=0\ .
\end{align}
As a result, for $K=T_\ell(\epsilon/4,\psi^\infty)$ from \eqref{eq:zeta1}-\eqref{eq:zeta2} we obtain
\begin{align}\label{ineq}
\lim_{n\rightarrow\infty} \mathbb{P}_1&\Big(\sup_{t<L}\Lambda_t\geq  cLI_1\Big) \leq \lim_{n\rightarrow\infty} \mathbb{P}_1\Big(\frac{1}{L}\sup_{t \leq K}\Lambda_t\geq  \frac{\epsilon}{4} \Big)\ .
\end{align}
For the right hand side of \eqref{ineq} we find that for any $\epsilon>0$
\begin{align}
\mathbb{P}_1\Big(\frac{1}{L}\sup_{t<K}\Lambda_t \geq  \frac{\epsilon}{4} \Big)  \; & \leq  \; \mathbb{P}_1\Big(\frac{1}{L}\sum_{t=1}^K \Lambda_t\geq  \frac{\epsilon}{4} \Big)  \\
\label{ineq1_1}& \leq\frac{4}{\epsilon}\cdot \frac{1}{L}\bbe_1\left[\sum_{t=1}^K \Lambda_t\right] \\
\label{ineq1_2}\; & = \; \frac{4}{\epsilon}\cdot \frac{1}{L}\bbe_1\left[\sum_{t=1}^K \bbe_1[\Lambda_t]\right]\\
\label{ineq1_3}\; & \leq  \; \frac{4}{\epsilon}\cdot \frac{1}{L}\bbe_1[K]\max_{1\leq t\leq K}\bbe_1[\Lambda_t]\ ,
\end{align}
where \eqref{ineq1_1} holds due to Markov's inequality and \eqref{ineq1_2} follows from Wald's identity (general form). Next, we set
\begin{align}\label{eq:L}
L= \left\lceil\frac{n{\alpha}}{I_1+\epsilon}\right\rceil\ .
\end{align}
By recalling Lemma~\ref{lemma:kappa} we know that for sufficiently large $n$,  we have $L\leq n$. Hence, based on \eqref{ineq} and \eqref{ineq1_3} we get
\begin{align}
\label{ineq2_1}\lim_{n\rightarrow\infty} \mathbb{P}_1  \Big(\sup_{t<L}\Lambda_t\geq  cLI_1\Big) & \leq   \lim_{n\rightarrow\infty} \frac{4}{\epsilon}\cdot \frac{1}{L}\bbe_1[K]\max_{1\leq t\leq K}\bbe_1[\Lambda_t]\ \\
\label{ineq2_2} & \overset{\eqref{eq:L}} {\leq } \frac{4}{\epsilon}\cdot\frac{I_1+\epsilon}{{\alpha}} \lim_{n\rightarrow\infty}\frac{1}{n} \; \bbe_1[K]\max_{1\leq t\leq K}\bbe_1[\Lambda_t]\\
\label{ineq2_3} & = 0\ ,
\end{align}
where the last step holds by noting that $\bbe_\ell[K]=\bbe_\ell[T_\ell(\epsilon/4,\psi^\infty)]<+\infty$ specified in~\eqref{eq:finite_ave}.
Subsequently, from \eqref{eq:zeta1}, \eqref{eq:zeta2}, \eqref{eq:part1}, \eqref{ineq2_1}, and \eqref{ineq2_3} we have
\begin{align}\label{eq:part2}
\lim_{n\rightarrow\infty} \mathbb{P}_1&\Big(\sup_{t<L}\Lambda_t\geq \zeta\Big) = 0\ .
\end{align}
As a result, from \eqref{eq:p1L_1}-\eqref{eq:p1L_2} we obtain 
\begin{align}
\label{ineq3} \lim_{n\rightarrow \infty} \mathbb{P}_1\left(\frac{\tau_n}{n}>\frac{{{\alpha}}}{{I}_1+\epsilon}\right) & \overset{\eqref{eq:L}}{=}  \lim_{n\rightarrow \infty} \mathbb{P}_1\left(\tau_n > L \right)  \\
\label{ineq222} & \overset{\eqref{eq:p1L_2} \;, \; \eqref{eq:part2}}{=}  \lim_{n\rightarrow\infty} \left[1-\exp(-n{\beta}) -\exp\Big(-n{\alpha}\cdot \frac{\epsilon}{2I_1+\epsilon}\Big)\right]\\
\label{ineq4} & = 1\ ,
\end{align}
which proves \eqref{eq:p1}. Since this is always valid irrespectively of the sampling procedure and the stopping rule, we conclude that \eqref{eq:lower1} is always valid, establishing 
\begin{align}
\lim_{n\rightarrow\infty}  \frac{\bbe_1\{\tau_n\}}{n} \geq\frac{{\alpha}}{I_1}\ .
\end{align}
We can prove
\begin{align}
\lim_{n\rightarrow\infty}  \frac{\bbe_0\{\tau_n\}}{n} &\geq\frac{{\beta}}{ I_0}\ ,
\end{align}
by following the same line of argument. 
}

\section{Proof of Theorem \ref{thm:asymp2}} 
\label{App:thm:asymp2}

\textcolor{black}{Following the definition of $T_1(h,\psi^\infty)$ in~\eqref{eq:T_l}, we provide a truncated counterpart of it for a network with $n$ nodes (non-asymptotic regime) as follows.
\begin{align}
\label{T1}
R_1(h,\psi^{n})&\dff\sup\ \Big\{t\leq n\,:\,\Big|\frac{\Lambda_t}{t}-I_1\Big|>h\Big\} \ , \quad \forall h>0 \ ,
\end{align}
where we adopt the convention that the supremum of an empty set is $+\infty$. Obviously,
\begin{align}\label{eq:TR}
\lim_{n\rightarrow\infty} R_1(h,\psi^{n}) = T_1(h,\psi^\infty)\ .
\end{align}
According to the definition of the stopping time in~\eqref{eq:stop}, at the instance prior to stopping, i.e., at time $\tau^*_n-1$, we always have $\Lambda_{\tau^*_n-1}\in(\gamma^{\rm L},\gamma^{\rm U})$. We start the proof by comparing $\Lambda_{\tau^*_n-1}$ with these two bounds. First, consider the following relationship
\begin{align}\label{ineq11}
\Lambda_{\tau^*_n-1}<\gamma^{\rm U} \ .
\end{align}
Based on the definition of $R_1(h,\psi^n)$ in~\eqref{T1}, if $R_1(h,\psi^{\tau^*_n})<{\tau^*_n}-1$, then for $t=\tau^*_n-1$ we have
\begin{align}
\Big|\frac{\Lambda_{\tau^*_n-1}}{\tau^*_n-1} - I_1\Big|\leq h \ , \quad \forall h>0 \ ,
\end{align}
which indicates that for all $h\in(0,I_1)$ we have
\begin{align}\label{ineq23}
  \tau^*_n\leq   \frac{\Lambda_{\tau^*_n-1}}{I_1-h} +1 \; \overset{\eqref{ineq11}}{\leq} \; \frac{\gamma^{\rm U}}{I_1-h}+1 \ .
\end{align}
Hence, from \eqref{ineq23} for all $h\in(0,I_1)$ we have
\begin{align}
\tau^*_n & = \tau^*_n\cdot  \mathds{1}_{\{{\tau^*_n}>R_1(h,\psi^{\tau^*_n})+1\}} + \underbrace{\tau^*_n\cdot  \mathds{1}_{\{{\tau^*_n}\leq R_1(h,\psi^{\tau^*_n})+1\}}}_{\leq R_1(h,\psi^{\tau^*_n})+1}\\
& \overset{\eqref{ineq23}}{\leq }\left[\frac{\gamma^{\rm U}}{I_1-h} +1\right]\cdot  \mathds{1}_{\{{\tau^*_n}>R_1(h,\psi^{\tau^*_n})+1\}} +R_1(h,\psi^{\tau^*_n})+1\\
& \leq 2+\frac{\gamma^{\rm U}}{I_1-h}+R_1(h,\psi^{\tau^*_n})\ .
\end{align}
Subsequently,
\begin{align}
\label{ineq24}\tau^*_n & \leq 2+\inf_{h\in(0,I_1)}\frac{\gamma^{\rm U}}{I_1-h}+R_1(h,\psi^{\tau^*_n})\\
\label{ineq25} & \leq  2+\frac{\gamma^{\rm U}}{I_1}+ R_1(h,\psi^{\tau^*_n})\ .
\end{align}
Since
\begin{align}\label{eq:T_l:b}
\mathbb{E}_1\{T_1(h,\psi^\infty)\}<+\infty\ , \quad \forall h>0\ ,
\end{align}
by recalling that $\gamma^{\rm U}=n{\alpha}$, from  \eqref{eq:TR} and \eqref{ineq24}-\eqref{ineq25} we obtain
\begin{align}
\lim_{n \rightarrow\infty }\frac{\mathbb{E}_1\{ {\tau^*_n}\}}{n}\leq \frac{{\alpha}}{I_1}\ .
\end{align}
Similarly, by also considering 
\begin{align}\label{ineq1144}
\Lambda_{\tau^*_n-1}>\gamma^{\rm L} 
\end{align}
and following the same line of argument we obtain
\begin{align}
\lim_{n \rightarrow\infty }\frac{\mathbb{E}_0\{ {\tau^*_n}\}}{n}\leq \frac{{\beta}}{I_0}\ ,
\end{align}
which concludes the proof.
}

\section{Proof of Lemma \ref{lemma:finite}} 
\label{app:lemma:finite}

\textcolor{black}{
We start by showing that there exist positive constants $B$ and $c$ such that for all $t\in\{1,\dots,n\}$
\begin{align}\label{eq:conv}
\mathbb{P}_1(\hat\tau_n \geq t)\leq B {\e}^{-ct} \ .
\end{align}
For this purpose, note that
\begin{align}
\mathbb{P}_1(\hat\tau_n \geq t)&= \sum_{u=t}^n\P_1(\hat\tau_n=u)\\
& = \sum_{u=t}^n\mathbb{P}_1\Big(\delta_{\rm ML}(u-1)=\H_0 \; , \; \delta_{\rm ML}(u)=\dots=\delta_{\rm ML}(n)=\H_1\Big)\\
&\leq\sum_{u=t}^n \mathbb{P}_1(\delta_{\rm ML}(u-1)=\H_0) \\
\label{eq:exx} & \overset{\eqref{eq:ml:decision}}{=} \sum_{u=t-1}^{n-1} \mathbb{P}_1(\Lambda_u<0)\ .
\end{align}
Next, we find an upper bound on $\mathbb{P}_1(\Lambda_u<0)$. For this purpose, note that for any $s\in\mathbb{R}$ we have
\begin{align}\label{eq:exp:ineq}
\mathbb{P}_1(\Lambda_t<0)\cdot\mathbb{E}_1\big\{\exp\{s\Lambda_t\}\;|\;\mathds{1}_{\{\Lambda_t<0\}}\big\} 
&= \mathbb{E}_1\big\{\exp\{s\Lambda_t\}\mathds{1}_{\{\Lambda_t<0\}}\big\}\  \leq \mathbb{E}_1\big\{\exp\{s\Lambda_t\}\big\}\ .
\end{align}
Furthermore, for any $s<0$ we have 
\begin{align}\label{eq:gt:1}
\mathbb{E}_1\big\{\exp\{s\Lambda_t\}\;|\;\mathds{1}_{\{\Lambda_t<0\}}\big\}\geq 1 \ .
\end{align}
By combining~\eqref{eq:exp:ineq}--\eqref{eq:gt:1} we find that for any $s<0$
\begin{align}\label{eq:up}
\mathbb{P}_1(\Lambda_t<0) \leq \mathbb{E}_1\big\{\exp\{s\Lambda_t\}\big\}\ .
\end{align}
The right hand side of \eqref{eq:up} can be expanded by using the towering property of expectation as follows:
\begin{align}\label{eq:up1}
\mathbb{E}_1\Big\{\exp\{s\Lambda_t\}\Big\}\overset{\eqref{LLR}}{=}\mathbb{E}_1\Big\{\exp\{s\Lambda_{t-1}\}\cdot\mathbb{E}_1\Big\{\Big[\frac{f_1(Y_t;\psi(t)|\F_{t-1})}{f_0(Y_t;\psi(t)|\F_{t-1})}\Big]^s\;\Big|\;\F_{t-1}\Big\}\Big\}\ .
\end{align}
Now, consider the inner expectation and define
\begin{align}
\xi_t(s)\dff \mathbb{E}_1\Big\{\Big[\frac{f_1(Y_t;\psi(t)|\F_{t-1})}{f_0(Y_t;\psi(t)|\F_{t-1})}\Big]^s\;\Big|\;\F_{t-1}\Big\}\ .
\end{align}
It can be ready verified that $\xi_t(s)$ is convex in $s$ and satisfies
\begin{align}
\xi_t(-1)=\xi_t(0) = 1\ .
\end{align}
 $\xi_t(s)$ can have two possible behaviors in the range $s\in(-1,0)$:\\
{\bf Case 1:} $\xi_t(s)=1,\ \forall s\in (-1,0)$. This occurs only when the likelihood ratio inside the expectation is equal to $1$, i.e., the sample taken at time $t$ has the same likelihood values under both hypotheses. This event has measure zero. As a result, the probability of this case occurring is 0.\\
{\bf Case 2:} $\xi_t(s)<1,\ \forall s\in (-1,0)$. It means that in this case there exists a constant ${c}>0$ such that for some $s^*\in(-1,0)$ and $\forall t\leq\tau^*_n$
\begin{align}\label{eq:up2}
\xi_t(s^*) \leq {\e}^{-{c}}<1\ .
\end{align}
By successively applying the towering property as in~\eqref{eq:up1}, and accounting for Case $1$ we obtain
\begin{align}\label{eq:up3}
\mathbb{P}_1(\Lambda_t<0) \overset{\eqref{eq:up}}{\leq} \mathbb{E}_1\Big\{\exp\{s^*\Lambda_t\}\Big\} \leq {\e}^{-{c} t}\ .
\end{align}
Next, by combining~\eqref{eq:exx} and~\eqref{eq:up3} we obtain
\begin{align}
\mathbb{P}_1(\hat\tau_n \geq t) \leq \sum_{u=t-1}^{n-1} {\e}^{-c u} \leq \sum_{u=t-1}^{\infty } {\e}^{-c u} =\frac{{\e}^{c}}{1-{\e}^{-c}}{\e}^{-ct}=b {\e}^{-ct}\ ,
\end{align}
where we have defined $b\dff \frac{1}{1-{\e}^{-c}}$. By using this result, it can be ready verified that $\bbe_1\{\hat\tau_n\}$ is finite. Specifically,
\begin{align}\label{eq:eqqq}
\bbe_1\{\hat\tau_n\}&=\sum_{t=1}^\infty \mathbb{P}(\hat\tau_n \geq t) \leq\sum_{t=0}^\infty b {\e}^{-ct} =\frac{b}{1-{\e}^{-c}}\ ,
\end{align}
which shows that $\bbe_1\{\hat\tau_n\}$ is asymptotically upper bounded by a constant. The proof for $\bbe_0\{\hat\tau_n\}$ being bounded by a constant follows a similar line of argument. }

\section{Proof of Lemma \ref{lemma:psi}}
\label{app:lemma:psi}

\noindent \textcolor{black}{
For any $h>0$ define
\begin{align}\label{eq:T_max}
T_1(h)\dff \max_{\psi^\infty\in\cS(\mathbb{N})}T_1(h,\psi^\infty)\ .
\end{align}
Based on the \eqref{eq:finite_ave}, for any sampling path $\psi^\infty$ and any $t>T_1(h)$ we have
\begin{align}\label{eq:T_max2}
{\sf nLLR}_1(Y^t;\psi^t)\; \in \; [I_1(\psi^\infty)-h\; , \;  I_1(\psi^\infty)+h]\ . 
\end{align}
By expanding the joint pdfs we find that for any pair of time instances $t$ and $s$ such that $s>t>T_1(h)$ we have
\begin{align}\label{eq:eq11}
\underset{=(t+|\cS|)\times {\sf nLLR}_1(Y^t\cup X_{\cS};\psi^t\cup \cS)}{\underbrace{\ln\dfrac{f_1(Y^t\cup X_{\cS};\psi^t\cup \cS) }{f_{0}(Y^t\cup X_{\cS};\psi^t\cup \cS)}}}\; = \;
\underset{=t\times {\sf nLLR}_1(Y^t\psi^t)}{\underbrace{\ln\dfrac{f_1(Y^{t};\psi^{t})}{f_{0}(Y^{t};\psi^{t})}}}+\ln\dfrac{f_1( X_{\cS};\cS\,|\,\F_{t})}{f_{0}( X_{\cS};\cS\,|\,\F_{t})}\ .
\end{align}
Hence, from \eqref{eq:T_max2} and \eqref{eq:eq11} we find that corresponding to any sampling path $\psi^\infty$, for any set $\cS\subseteq\psi^t$ and any $h>0$ we have
\begin{align}
\frac{1}{|\cS|}\ln\dfrac{f_1( X_{\cS};\cS\,|\,\F_{t})}{f_{0}( X_{\cS};\cS\,|\,\F_{t})} \in \;\left[I_1(\psi^\infty)-\left(\frac{2t}{|\cS|}+1\right)h\; , \;  I_1(\psi^\infty)+
\left(\frac{2t}{|\cS|}+1\right) h\right]\ ,
\end{align}
which indicates that for all $h>0$ we have
\begin{align}
&\forall \cS\subseteq U^\infty\; : && \max_{\cS\in\varphi^t}   \enspace \bbe_1\bigg\{ \frac{1}{|\cS|}  \ln\dfrac{f_\ell(X_{\cS};\cS\,|\,\F_{t})}{f_{0}(X_{\cS};{\cS}\,|\,\F_{t})}\bigg\} \; \geq  I^*_1-\left(\frac{2t}{|\cS|}+1\right)h\ .
\end{align}
By selecting $h$ arbitrarily small, this lower bound can be made arbitrarily close to $I^*_1$. By noting the characteristics of the sampling rule $\psi^*(t)$ specified in~\eqref{eq:MC}, at time $t+1$ we select a node  from a set $\cS^*$ that satisfies 
\begin{align}\label{eq:S_lowerbound}
 \enspace \bbe_1\bigg\{ \frac{1}{|\cS^*|}  \ln\dfrac{f_\ell(X_{\cS^*};\cS^*\,|\,\F_{t})}{f_{0}(X_{\cS^*};{\cS^*}\,|\,\F_{t})}\bigg\} \; \geq  I^*_1-\left(\frac{2t}{|\cS|}+1\right)h\ .
\end{align}
On the other hand, corresponding to set $\psi^\infty\neq U^\infty$, we have $I_1(\psi^\infty)<I^*_1$. Consequently, for all sets $\cS\subseteq \psi^\infty$ and arbitrarily small $h$ we have 
\begin{align}\label{eq:S_upperbound}
\enspace \bbe_1\bigg\{ \frac{1}{|\cS|}  \ln\dfrac{f_\ell(X_{\cS};\cS\,|\,\F_{t})}{f_{0}(X_{\cS};{\cS}\,|\,\F_{t})}\bigg\} \; \leq  I_1(\psi^\infty)+\left(\frac{2t}{|\cS|}+1\right) h\;<\; I^*_1-\left(\frac{2t}{|\cS|}+1\right)h\ ,
\end{align}
Hence, the set $\cS^*$ from which we sample at time $t+1$ must be a subset of $U^\infty$. In other words, for all $t>T_1(h)$ the samples are taken from $U^\infty$. This indicates that for the set $\cH$ we have
\begin{align}
\bbe_1\{\cH\} \leq \bbe_1\{T_1(h)\} \overset{\eqref{eq:T_max}}{< } +\infty\ .
\end{align}
}

\newpage
\section{Proof of Lemma~\ref{lemma2}}
\label{app:lemma2}

\textcolor{black}{
Following the definition of $T_1(h,\psi^\infty)$ in~\eqref{eq:T_l} for heterogeneous networks , we provide a truncated counterpart for it defined as follows. For this purpose, we denote the first $n$ elements of $U^\infty$ by $U^n$.
\begin{align}
\label{T1hat}
R_1(h,U^{n})&\dff \sup\ \Big\{t\leq n\,:\,\Big|\frac{\Lambda_t}{t}-I^*_1\Big|>h\Big\} \ , \quad \forall h>0 \ ,
\end{align}
where clearly
\begin{align}
\lim_{n\rightarrow\infty} R_1(h,U^{n}) = T_1(h,\psi^n)\ .
\end{align}
  Hence, by accounting for the first $(\hat\tau_n+T_1(h))$ samples, some of which may have been observed from nodes not included in the first $\hat(\tau_n+T_1(h))$ elements of  $U^\infty$, the last time that the normalized log-likelihood ratios $\frac{\lambda_t}{t}$ leaves the interval $[I^*_1-h,I^*_1+h]$ will happen no later than $R_1(h,U^{n})+\hat\tau_n+T_1(h)$. In other words,
\begin{align}
\label{T1hat2}
\forall t\geq R_1(h,U^{n})+\hat\tau_n+T_1(h)\; : \qquad -h \leq\frac{\Lambda_t-\Lambda_{\hat\tau_n}}{(t-\hat\tau_n)}-I_1^* \leq h\ , \quad \forall h>0 \ .
\end{align}
If ${\tau^*_n}> R_1(h,U^n)+\hat\tau_n+T_1(h)$, then for all $h\in(0,I^*_1)$ we have
\begin{align}\label{ineq23hat}
 \tau^*_n-\hat\tau_n \leq \frac{\gamma^{\rm U}-\Lambda_{\hat\tau_n}}{I_1^*-h} +1\ .
\end{align}
Hence, for all $h\in(0,I^*_1)$ we have
\begin{align}
\tau^*_n - \hat\tau_n& = (\tau^*_n - \hat\tau_n)\cdot  \mathds{1}_{\{{\tau^*_n}>R_1(h,U^n)+\hat\tau_n+T_1(h)\}} + \underbrace{(\tau^*_n - \hat\tau_n)\cdot  \mathds{1}_{\{{\tau^*_n}\leq R_1(h,U^n)+\hat\tau_n+T_1(h)\}}}_{\leq R_1(h,U^n)+T_1(h)}\\
& \overset{\eqref{ineq23hat}}{\leq }\left[\frac{\gamma^{\rm U}-\Lambda_{\hat\tau_n}}{I_1^*-h}+1\right]\cdot  \mathds{1}_{\{{\tau^*_n}> R_1(h,U^n)+\hat\tau_n+T_1(h)\}} + R_1(h,U^n)+T_1(h) \\
& \;\; \leq \frac{\gamma^{\rm U}-\Lambda_{\hat\tau_n}}{I_1^*-h}+R_1(h,\psi^n)+T_1(h)+1\ .
\end{align}
Hence,
\begin{align}
\label{ineq24hat}
\tau^*_n-\hat\tau_n & \leq 1+\inf_{h\in(0,I^*_1)}\frac{\gamma^{\rm U}-\Lambda_{\hat\tau_n}}{I_1^*-h}+R_1(h,U^n)+T_1(h)\\
\label{ineq25hat}& = 1+\frac{\gamma^{\rm U}-\Lambda_{\hat\tau_n}}{I_1^*}+R_1(h,\psi^n)+T_1(h)\ .
\end{align}
Since the convergence of the $\sf nLLR$ is complete, we can conclude the proof of~\eqref{eq:itme1} by combining~\eqref{eq:th1} and~\eqref{ineq24hat}--\eqref{ineq25hat} to obtain
\begin{align}
\label{eq:kl:pos} \lim_{n \rightarrow\infty }\frac{\mathbb{E}_1\{ {\tau^*_n-\hat\tau_n}\}}{n} &\leq \frac{{\alpha}}{I_1^*}-\lim_{n \rightarrow\infty }\frac{\mathbb{E}_1\{\Lambda_{\hat\tau_n}\}}{nI_1^*} +\lim_{n\rightarrow\infty}\frac{\bbe_1[R_1(h,\psi^n)]}{n}+\lim_{n \rightarrow\infty }\frac{\bbe_1\{T_1(h)\}}{n}\\
\label{eq:kl:pos2}
&\leq \frac{{\alpha}}{I_1^*}+\lim_{n\rightarrow\infty}\frac{\bbe_1[T_1(h,U)]}{n}+\lim_{n \rightarrow\infty }\frac{\bbe_1\{T_1(h)\}}{n}\\
& = \frac{{\alpha}}{I_1^*}\ ,
\end{align}
where~\eqref{eq:kl:pos} holds since $\mathbb{E}_1\{\Lambda_{\hat\tau_n}\}$ is a KL divergence term and it is non-negative, and \eqref{eq:kl:pos2} holds since $\bbe_1\{T_1(h)\}$ and $\bbe_1\{T_1(h)\}$ are finite values.
}




\section{Proof of Theorem \ref{thm:per}} 
\label{App:thm:per}

The error exponents of the NP test are studied in \cite{Chen}, where it is shown that when \textcolor{black}{$\mathsf{P}^{0}_{\rm NP}$} is fixed, which is equivalent to an error exponent of $0$, the error exponent of $\mathsf{P}_n^{1}$ is the convergence limit of ${\sf nLLR}_0(Y^n;\psi^n)$ as $n$ grows under the assumption that $\{Y_1,\dots,Y_n\}$ are drawn from distribution $f_0$. This is equivalent to the definition of $I_0$. Hence, for the NP test we have \textcolor{black}{$E^{1}_{\rm NP}=I_0$ and $E^{0}_{\rm NP}=0$}. For the sequential sampling setting, based on the analysis of the average delay in  \textcolor{black}{
Theorems~\ref{thm:asymp1} and~\ref{thm:asymp2}} we have
\begin{align}
\label{eq:eqeq}
\text{and}\quad \lim_{n\rightarrow\infty} \frac{\bbe_1\{\tau_n^*\}}{n} &=\frac{{\alpha}}{ I_1}\ .
\end{align}
\textcolor{black}{
For the error exponent of $\mathsf{P}_n^{0}$ yielded by Algorithm~\ref{table} we have
\begin{align}\label{eq:Efan}
E_n^{0} & = -\lim_{n\rightarrow\infty} \frac{1}{r_1}\ln\mathsf{P}^{0}_n(r_1)\\
\label{eq:feaz}
& \geq -\lim_{n\rightarrow\infty}-\frac{n{\alpha}}{r_1}\\
& = \lim_{n\rightarrow\infty}\frac{n}{\bbe_1\{\tau^*_n\}}\cdot {\alpha} \\
& \overset{\eqref{eq:eqeq}}{=} I_1 \ , 
\end{align}
where~\eqref{eq:feaz} follows from Algorithm~1 generating $(\alpha,\beta)$-accurate decisions. Next, we define $\Delta \dff E_n^{0} - I_1\geq 0$, based on which we have 
\begin{align}\label{neqeq}
-\ln \mathsf{P}_n^{0}(r_1) \overset{\eqref{eq:Efan}}{=} r_1E_n^{0}  +o(n)= r_1\Delta + r_1I_1+o(n)\ .
\end{align}
On the other hand, we observe that in the proof of Theorem~\ref{thm:asymp1}, $\mathsf{P}_n^{0}$ has been replaced by its upper bound. By keeping $\mathsf{P}_n^{0}$ throughout the proof it can be readily shown that
\begin{align}
\label{eq:ffeeaa}
 \lim_{n\rightarrow\infty} \frac{\mathbb{E}_1\{{\tau^*_n}\}}{n}&\geq \frac{|\ln\mathsf{P}^{0}_n(r_1)|}{nI_1}\ .
\end{align}
By combining~\eqref{neqeq} and~\eqref{eq:ffeeaa} we obtain
\begin{align}
\lim_{n\rightarrow\infty} \frac{\mathbb{E}_1\{{\tau^*_n}\}}{n}
& \geq \lim_{n\rightarrow\infty}\frac{|r_1\Delta + r_1I_1+o(n)|}{nI_1} \\
& = \lim_{n\rightarrow\infty}\frac{r_1}{n}\cdot \frac{\Delta+I_1}{I_1}\\
\label{llaa}& \overset{\eqref{eq:eqeq}}{=} \frac{{\alpha}}{I_1}\cdot \frac{\Delta+I_1}{I_1}\ .
\end{align}
On the other hand, from Theorem~\ref{thm:asymp2} we have
\begin{align}\label{llaa2}
\lim_{n\rightarrow\infty} \frac{\bbe_1\{\tau_n^*\}}{n} &\leq\frac{{\alpha}}{ I_1}\ .
\end{align}
By comparing~\eqref{llaa} and \eqref{llaa2} and noting that $\Delta\geq 0$, in the asymptote of large $n$ we should have $\Delta=0$, and consequently, $E_n^{0}=I_1$. The error exponent of $\mathsf{P}_n^{1}$ can be obtained by following the same line of argument.}

\section{Proof of Theorem \ref{thm:markov}} 
\label{App:thm:markov}

\textcolor{black}{
Without loss of generality assume that at time $t-1$ we have $\delta_{\rm ML}(t-1)={\sf H}_\ell$. By recalling the definition of  $\mathcal{R}_t^i$ given  in~\eqref{eq:S_valid}, corresponding to any unobserved node $i\in\varphi^t$ at time $t$ we  define $\bar\cS_t^i\in\mathcal{R}_t^i$ as the {\sl smallest} set of nodes that maximizes the normalized information measure assigned to node $i\in\varphi^t$ at time $t$, i.e.,
\begin{align}\label{eq:S_bar}
\bar\cS_t^i\dff\arg\max_{\cS\in\mathcal{R}_t^i} \frac{M_\ell^i(t,\cS)}{|\cS|}\ .
\end{align}
Also, we define
\begin{align}\label{eq:u}
u\dff  \argmax_{i} \frac{M_\ell^i(t,\bar\cS_t^i)}{|\bar\cS_t^i|}\ ,
\end{align}
as the index of the node that exhibits the largest normalized information measure\footnote{For convenience in notation, we suppressed the dependence of $u$ on $t$, $\ell$, and the past samples.}, selected by the selection rule specified in~\eqref{eq:MC}. Hence, the optimal sampling path is the set $\bar\cS_t^u$. In order to prove the theorem, we show that the maximum normalized information measure achieved by the set $\bar\cS_t^u$ is equal to the normalized information measure achieved by only the members of $\bar\cS_t^i\in\mathcal{R}_t^i$ that are neighbors of $u$. In other words, by defining
\begin{align}
{\cT}_t^u\dff \bar\cS_t^u\cap\mathcal{L}_t^u\ ,
\end{align}
we show that
\begin{align}
\frac{M_\ell^u(t,\bar\cS_t^u)}{|\bar\cS_t^u|} =\frac{M_\ell^u(t,\cT_t^u)}{|\cT_t^u|}\ . 
\end{align}
We prove this identity by removing the nodes not neighboring $u$ in four steps, and in each step showing that removing those nodes does not penalize $\frac{M_\ell^u(t,\bar\cS_t^u)}{|\bar\cS_t^u|} $.
\begin{enumerate}
\item Removing any node in $\bar\cS_t^u$ that belongs to a subgraph of $\G$ different from the subgraph that containins $u$, does not decrease $\frac{M_\ell^u(t,\bar\cS_t^u)}{|\bar\cS_t^u|} $. 
\item Furthermore, removing any node of $\bar\cS_t^u$ whose path to $u$ contains a node that has been observed earlier, does not decrease $\frac{M_\ell^u(t,\bar\cS_t^u)}{|\bar\cS_t^u|} $. 
\item Moreover, removing any node of $\bar\cS_t^u$ whose path  to $u$ contains an unobserved node that does not belong to $\bar\cS_t^u$, does not decrease $\frac{M_\ell^u(t,\bar\cS_t^u)}{|\bar\cS_t^u|} $.
\item Finally, removing any remaining node that is not a neighbor of $u$ does not decrease $\frac{M_\ell^u(t,\bar\cS_t^u)}{|\bar\cS_t^u|} $.
\end{enumerate}
{\bf Step 1:} First we show that removing the nodes from all subgraphs of $\G$ other than the one that containing node $u$, does not increase the information measure of node $u$. For this purpose,  we partition $\bar{S}_t^u$ according to
\begin{align}
\bar{S}_t^u= A\cup \bar A\ ,\quad\text{and}\quad A \cap \bar A=\phi \ ,
\end{align}
where $A\subseteq\bar{S}_t^u$ is the set of nodes that belong to the same subgraph as $u$, and $\bar A\dff\bar{S}_t^u\setminus A$.
We expand the information measure of $u$ as follows:
\begin{align}
\label{eq:external}
\frac{M_\ell^u(t,\bar\cS_t^u)}{|\bar\cS_t^u|} 
&= \frac{D_{\rm KL}\big(f_\ell(X_{\bar A}|\F_{t-1})\;\|\;f_{1-\ell}(X_{\bar A}|\F_{t-1})\big)}{|\bar\cS_t^u|} \\
\label{eq:internal}
&\qquad + \frac{D_{\rm KL}\big(f_\ell(X_{A}|\F_{t-1})\;\|\;f_{1-\ell}(X_{A}|\F_{t-1})\big)}{|\bar\cS_t^u|} \ .
\end{align}
We note that $A$ is non-empty since $u\in A$. We show that if $\bar A$ is not empty, removing it does not decrease the information measure of node $u$. Suppose otherwise, i.e., $\bar A$ is non-empty and 
\begin{align}\label{eq:cont:assump}
\frac{D_{\rm KL}\big(f_\ell(X_{A}|\F_{t-1})\;\|\;f_{1-\ell}(X_{A}|\F_{t-1})\big)}{|A|} < \frac{M_\ell^u(t,\bar\cS_t^u)}{|\bar\cS_t^u|}\ .
\end{align}
Then, in order for~\eqref{eq:external}--\eqref{eq:internal} to hold, we must have
\begin{align}
\frac{D_{\rm KL}\big(f_\ell(X_{\bar A}|\F_{t-1})\;\|\;f_{1-\ell}(X_{\bar A}|\F_{t-1})\big)}{|\bar A|} > \frac{M_\ell^u(t,\bar\cS_t^u)}{|\bar\cS_t^u|}\ .
\end{align}
Denote one of the members of $\bar A$ by $v$. Then, by noting that $\bar A\subseteq {\cal}^v_t$ and invoking the definition of $u$, we have
\begin{align}\label{eq:cont:assump2}
\frac{D_{\rm KL}\big(f_\ell(X_{\bar A}|\F_{t-1})\;\|\;f_{1-\ell}(X_{\bar A}|\F_{t-1})\big)}{|\bar A|} 
\overset{\eqref{eq:S_bar}}{\leq} \max_{\cS\in\mathcal{R}_t^v}\frac{M_\ell^v(t,\cS)}{|\cS|} \overset{\eqref{eq:u}}{\leq} \frac{M_\ell^u(t,\bar\cS_t^u)}{|\bar\cS_t^u|}\ ,
\end{align}
which contradicts \eqref{eq:cont:assump}. 
Hence, we remove  all the nodes that do not belong to the subgraph of $\G$ that contains $u$, and assume that the optimal set $\bar\cS^u_t$ is free of such nodes. In the next steps, we focus only on the nodes that belong to the same subgraph that $u$ lies in. \vspace{.1 in}}\\
\textcolor{black}{
\noindent {\bf Step 2:} Next, we show that further removing the nodes whose path to $u$ contains a node that has been observed earlier, does not increase the information measure of $u$. For this purpose, we partition $\bar\cS_t^u$ according to
\begin{align}
\bar{S}_t^u= B \cup \bar B\ ,\quad\text{and}\quad B \cap \bar B=\phi \ ,
\end{align}
where $B\subseteq \bar{S}_t^u$ is the set of nodes whose paths to $u$ includes an observed node, i.e. an element of $\psi_n^{t-1}$. According to the global Markov property we have 
\begin{align}
B \independent \bar B\ \big|\ \F_{t-1}\ ,
\end{align}
Hence, we have the decomposition
\begin{align}
\frac{M_\ell^u(t,\bar\cS_t^u)}{|\bar\cS_t^u|} 
&= \frac{D_{\rm KL}\big(f_\ell(X_{B}|\F_{t-1})\;\|\;f_{1-\ell}(X_{B}|\F_{t-1})\big)}{|\bar\cS_t^u|} \\
&\qquad + \frac{D_{\rm KL}\big(f_\ell(X_{\bar B}|\F_{t-1})\;\|\;f_{1-\ell}(X_{\bar B}|\F_{t-1})\big)}{|\bar\cS_t^u|} \ .
\end{align}
We can follow the exact same line of argument as in Step 1,  to prove that removing the the nodes in $B$ does not decrease the information measure of node $u$, and consequently, the selected node. \vspace{.1 in}}\\
\textcolor{black}{
\noindent{\bf Step 3:} In the next step, we show that further removing any node of $\cS^u_t$ whose path to $u$ contains an unobserved node that does not belong to $\cS^u_t$ can be also removed without penalizing the desired information measure. For this purpose, we partition the set $\bar\cS_t^u$ according to
\begin{align}
\bar{S}_t^u= C\cup \bar C\ ,\quad\text{and}\quad C\cap \bar C=\phi \ ,
\end{align}
where $\bar C$ is the set of nodes whose paths to $u$ contains at least one node that does not belong to $\bar\cS_t^u$. Let us also define the set $C_j$ as a subset of $\bar C$ whose paths to $u$ contains the unobserved node $j\notin\cS^u_t$. Since the graph is acyclic, the sets $\{C_j\}$ are disjoint and partition $\bar C$, i.e. 
\begin{align}
\bar C=\bigcup_{j\in J} C_j\ , \quad\text{and}\quad C_j \cap C^{j'}=\phi\ ,\ \forall j,j'\in J \ ,
\end{align}
where we have defined $J$ as the smallest set that separates $C$ and $\bar C$. Then, we expand the information measure of $u$ as follows:
\begin{align}
\label{eq:connect}
\frac{M_\ell^u(t,\bar\cS_t^u)}{|\bar\cS_t^u|} 
&= \frac{D_{\rm KL}\big(f_\ell(X_{C}|\F_{t-1})\;\|\;f_{1-\ell}(X_{C}|\F_{t-1})\big)}{|\bar\cS_t^u|} \\
\label{eq:diconnect}
&\qquad + \sum_{j\in J} \frac{ D_{\rm KL}\big(f_\ell(X_{C_j}|X_{\M_j},\F_{t-1})\;\|\;f_{1-\ell}(X_{C_j}|X_{\M_j},\F_{t-1})\big)}{|\bar\cS_t^u|}\\
&\leq \max\Bigg\{\frac{D_{\rm KL}\big(f_\ell(X_{C}|\F_{t-1})\;\|\;f_{1-\ell}(X_{C}|\F_{t-1})\big)}{|C|}, \\
\label{eq:diconnect22}
&\qquad\qquad\qquad \max_{j\in J}\frac{ D_{\rm KL}\big(f_\ell(X_{C_j}|X_{\M_j},\F_{t-1})\;\|\;f_{1-\ell}(X_{C_j}|X_{\M_j},\F_{t-1})\big)}{|C_j|}\Bigg\} \ ,
\end{align}
where we have defined
\begin{align}\label{eq:n_j}
 \M_j \dff \N_j \cap C\ .
\end{align}
We prove this step by contradiction. Suppose that
\begin{align}
\frac{M_\ell^u(t,\bar\cS_t^u)}{|\bar\cS_t^u|} 
& > \frac{D_{\rm KL}\big(f_\ell(X_{C}|\F_{t-1})\;\|\;f_{1-\ell}(X_{C}|\F_{t-1})\big)}{|C|} \ .
\end{align}
Hence, for~\eqref{eq:connect}--\eqref{eq:diconnect22} to hold, we should have
\begin{align}
\frac{M_\ell^u(t,\bar\cS_t^u)}{|\bar\cS_t^u|} 
& < \max_{j\in J}\frac{ D_{\rm KL}\big(f_\ell(X_{C_j}|X_{\M_j},\F_{t-1})\;\|\;f_{1-\ell}(X_{C_j}|X_{\M_j},\F_{t-1})\big)}{|C_j|} \ ,
\end{align}
indicating that there exists at least one $j\in J$ such that
\begin{align}\label{contr}
\frac{M_\ell^u(t,\bar\cS_t^u)}{|\bar\cS_t^u|} 
& < \frac{ D_{\rm KL}\big(f_\ell(X_{C_j}|X_{\M_j},\F_{t-1})\;\|\;f_{1-\ell}(X_{C_j}|X_{\M_j},\F_{t-1})\big)}{|C_j|} \ .
\end{align}
Next, by defining $\bar C_j\dff C_j\cup\{j\}$, we consider the following two different expansions for 
\begin{align}
D_{\rm KL} \big(f_\ell(X_{\bar C_j}|X_{\M_j},\F_{t-1})\;\|\;f_{1-\ell}(X_{\bar C_j}|X_{\M_j},\F_{t-1})\big)\ .
\end{align}
Specifically, on one hand we have
\begin{align}\label{eq:LHS:11}
D_{\rm KL} \big(f_\ell(X_{\bar C_j}|X_{\M_j},\F_{t-1})&\;\|\;f_{1-\ell}(X_{\bar C_j}|X_{\M_j},\F_{t-1})\big) \\
& = D_{\rm KL} \big(f_\ell(X_j|X_{\M_j},\F_{t-1})\;\|\;f_{1-\ell}(X_j|X_{\M_j},\F_{t-1})\big) \\
& + D_{\rm KL}\big(f_\ell(X_{C_j}|X_{j},\F_{t-1})\;\|\;f_{1-\ell}(X_{C_j}|X_{j},\F_{t-1})\big)\ ,
\end{align}
and on the other hand we have
\begin{align}
D_{\rm KL} \big(f_\ell(X_{\bar C_j}|X_{\M_j},\F_{t-1})&\;\|\;f_{1-\ell}(X_{\bar C_j}|X_{\M_j},\F_{t-1})\big) \\
& = D_{\rm KL} \big(f_\ell(X_{C_j}|X_{\M_j},\F_{t-1})\;\|\;f_{1-\ell}(X_{C_j}|X_{\M_j},\F_{t-1})\big) \\
& + D_{\rm KL}\big(f_\ell(X_j|X_{\M_j},X_{C_j},\F_{t-1})\;\|\;f_{1-\ell}(X_{j}|X_{\M_j},X_{C_j},\F_{t-1})\big)\ .
\end{align}
Since the KL divergence is a convex function in both of its arguments and $f_\ell(X_j|X_{\M_j},\F_{t-1})$  is the average of $f_\ell(X_j|X_{\M_j},X_{C_j},\F_{t-1})$, by applying Jensen's inequality we obtain
\begin{align}
D_{\rm KL}\big(f_\ell(X_j|X_{\M_j},X_{C_j},\F_{t-1})\;\|\;&f_{1-\ell}(X_{j}|X_{\M_j},X_{C_j},\F_{t-1})\big) \geq \\
\label{eq:jensen:1}
&D_{\rm KL} \big(f_\ell(X_{j}|X_{\M_j},\F_{t-1})\;\|\;f_{1-\ell}(X_{j}|X_{\M_j},\F_{t-1})\big)\ .
\end{align}
By combining~\eqref{eq:LHS:11}--\eqref{eq:jensen:1} we get
\begin{align}
D_{\rm KL}\big(f_\ell(X_{C_j}|X_{j},\F_{t-1})\;\|\;&f_{1-\ell}(X_{C_j}|X_{j},\F_{t-1})\big) \geq \\
& D_{\rm KL} \big(f_\ell(X_{C_j}|X_{\M_j},\F_{t-1})\;\|\;f_{1-\ell}(X_{C_j}|X_{\M_j},\F_{t-1})\big)\ ,
\end{align}
which in conjunction with~\eqref{contr} yields
\begin{align}
\frac{M_\ell^u(t,\bar\cS_t^u)}{|\bar\cS_t^u|} 
< \frac{D_{\rm KL}\big(f_\ell(X_{C_j}|X_{j},\F_{t-1})\;\|\;f_{1-\ell}(X_{C_j}|X_{j},\F_{t-1})\big)}{|C_j|} \leq \frac{M_\ell^u(t,C_j)}{|C_j|} \ .
\end{align}
This identity, however, contradicts the optimality of $u$, that is $u$ is the node with the largest information measure.  \vspace{.1 in}}\\
\textcolor{black}{
\noindent {\bf Step 4:} The first three steps, collectively, establish that based on the definition of $\bar\cS_t^u$ (being the smallest set that maximizes the information measure), the graph formed by the set of nodes in $\bar\cS_t^u$ is connected and is not separated by any subset of nodes in $\V\setminus\bar\cS_t^u$. This indicates that so far we have shown that $\bar\cS_t^u$ should contain only neighbors of $u$ or other nodes that are connected to $u$ via a neighbor of $u$. In the final stage we show cannot contain any node other than the neighbors of $u$. By contradiction, suppose that $\bar\cS_t^u$ contains at least one node that is not a neighbor of $u$. We denote this node by $k$. By defining
\begin{align}
\cS_t^u \dff\bar\cS_t^u\setminus\{k\}\ ,
\end{align}
we have
\begin{align}\label{eq:ffff}
\frac{M_\ell^u(t,\bar\cS_t^u)}{|\bar\cS_t^u|} =  
\frac{M_\ell^u(t,\cS_t^u) + D_{\rm KL}\big(f_\ell(X_{k}|X_{\M_k},\F_{t-1})\;\|\;f_{1-\ell}(X_{k}|X_{\M_k},\F_{t-1})\big)}{|\bar\cS_t^u|}\ ,
\end{align}
where we have defined
\begin{align}\label{eq:n_k}
\M_k \dff \N_k \cap \cS_t^u \ .
\end{align}
Since $\bar\cS_t^u$ maximizes the normalized information content of $u$, we have
\begin{align}
\frac{M_\ell^u(t,\bar\cS_t^u)}{|\bar\cS_t^u|} >  
\frac{M_\ell^u(t,\cS_t^u)}{|\cS_t^u|}\ ,
\end{align}
and, consequently, in order for~\eqref{eq:ffff} to hold we should have
\begin{align}\label{qwer}
D_{\rm KL}\big(f_\ell(X_{k}|X_{\M_k},\F_{t-1})\;\|\;f_{1-\ell}(X_{k}|X_{\M_k},\F_{t-1})\big) > \frac{M_\ell^u(t,\bar\cS_t^u)}{|\bar\cS_t^u|}\ .
\end{align}
On the other hand, we have
\begin{align}
D_{\rm KL}\big(f_\ell(X_{k}|X_{\M_k},\F_{t-1})\;\|\;f_{1-\ell}(X_{k}|X_{\M_k},\F_{t-1})\big) < \frac{M_\ell^k(t,\bar\cS_t^k)}{|\bar\cS_t^k|}\ ,
\end{align}
which combined with~\eqref{qwer} indicates
\begin{align}
\frac{M_\ell^u(t,\bar\cS_t^u)}{|\bar\cS_t^u|} < \frac{M_\ell^k(t,\bar\cS_t^k)}{|\bar\cS_t^k|}\ .
\end{align}
This contradicts the optimality of of $u$, and as a result  $\bar\cS_t^u$ cannot contain any node that is not a neighbor of $u$. This completes the proof.}

\section{Proof of Theorem \ref{thm:per3}} 
\label{App:thm:per3}

For a GMRF with an underlying line dependency graph, when $\sigma_{ij}=\sigma$ \textcolor{black}{among the neighboring nodes, we have a homogeneous networks} in which
\begin{align}
 I_0 &= \ln(1-\sigma^2)+\frac{2\sigma^2}{1-\sigma^2}\ ,  \qquad  \mbox{and}\qquad I_1 = \ln\frac{1}{1-\sigma^2} \ .
\end{align}
By applying these identities \textcolor{black}{to sets $A$ and $B$ in which $\sigma>a$ and $\sigma<b$, respectively,} and noting that $I_0$ and $I_1$ are monotonically increasing functions of $|\sigma|$ we have
\begin{align}
\frac{ I_0(A)}{I_0(B)} & \geq \frac{\ln(1-a^2)+\frac{2a^2}{1-a^2}}{\ln(1-b^2)+\frac{2b^2}{1-b^2}}   \\
& = \frac{-a^2-\frac{a^4}{2}-\frac{a^6}{3}-o(a^6)+2a^2\big(1+a^2+a^4+o(a^4)\big)}{-b^2-\frac{b^4}{2}-\frac{b^6}{3}-o(b^6)+2b^2\big(1+b^2+b^4+o(b^4)\big)}   \\
& = \frac{a^2+\frac{3}{2}a^4+\frac{5}{6}a^6+o(a^6)}{b^2+\frac{3}{2}b^4+\frac{5}{6}b^6+o(b^6)}   \\
& \geq \frac{a^2}{b^2} \ ,
\end{align}  
where the last inequality holds since $a>b$. Similarly, for the expected delays under $\H_1$ we have
\begin{align}
\frac{ I_1(A)}{I_1(B)}  & \geq \frac{-\ln(1-a^2)}{-\ln(1-b^2)}   \\
& = \frac{a^2+\frac{a^4}{2}+\frac{a^6}{3}o(a^6)}{b^2+\frac{b^4}{2}+\frac{b^6}{3}+o(b^6)}   \\
& = \frac{a^2(1+\frac{1}{2}a^2+\frac{1}{3}a^4+o(a^4)}{b^2(1+\frac{1}{2}b^2+\frac{1}{3}b^4+o(b^4)}   \\
& \geq \frac{a^2}{b^2} \ .
\end{align}  
When $|A|=o(n)$, the Chernoff rule starts the sampling process from set $B$ with probability $1$ and since the graph is connected stays in set $B$ until it exhaust all its nodes. By invoking the results of Theorem~\ref{thm:asymp2}, we can conclude that the expected delay of the Chernoff rule under $\H_\ell$ is inversely proportional to $ I_\ell(B)$. Furthermore, from Corollary~\ref{thm:opt:L} and Theorem~\ref{thm:opt:U} the expected delay of our strategy under $\H_\ell$ is inversely proportional to $I_\ell(A)$, which concludes the proof

\section{Proof of Theorem \ref{thm:mn}} 
\label{App:thm:mn}

\textcolor{black}{
We define $\tau_d\triangleq \tau_{\rm c}-\tau^*_n$. The optimal sampling strategy starts by directly sampling from set $A$. For the Chernoff rule, however, there is a chance that it starts sampling from $B$ before entering $A$. We define $\tau_{\rm c}^A$ and $\tau_{\rm c}^B$ as the number of samples that the Chernoff rule spends on sets $A$ and $B$, respectively. We show that
\begin{align}\label{eq:diff}
\bbe_\ell\{\tau_{\rm c}^A\} \geq \bbe_\ell\{\tau^*_n\}\ ,\qquad \mbox{and} \qquad \bbe_\ell\{\tau_{\rm c}^B\} = \Theta\left(\frac{n}{p}\right)\ ,
\end{align}
which indicates the desires result, i.e., 
\begin{align}
0 \leq \bbe_\ell\{\tau_d\} = \bbe_\ell\{\tau_{\rm c}^A\} + \bbe_\ell\{\tau_{\rm c}^B\} -  \bbe_\ell\{\tau^*_n\}\ \geq  \Theta\left(\frac{n}{p}\right)\ .
\end{align}
The first identity in  \eqref{eq:diff}  follows the optimality of $\tau^*_n$. Specifically, the optimal rule starts by sampling from $A$ and stays inside $A$ until the stopping time $\tau^*_n$. On the other hand, the Chernoff rule might start from sampling $B$, but once it enters $A$ it remains there until it takes $\tau_{\rm c}^A$ samples. By noting the optimality of $\tau^*_n$, we immediately have the first identity in \eqref{eq:diff}. }\textcolor{black}{
In order to establish the second identity in \eqref{eq:diff}, we provide lower and upper bounds on the asymptotic value of $\bbe_\ell\{\tau_{\rm c}^B\}$. By definition, any sampling rule can take at most $(n-p)$ samples from set $B$. Hence, we obtain an upper bound as follows:
\begin{align}
\label{eq:diff1} \bbe_\ell\{\tau_{\rm c}^B\} & = \sum_{k=0}^{n-p} k\cdot \P_\ell(\tau_{\rm c}^B=k) \\
\label{eq:diff2} & = \sum_{k=0}^{n-p} k\cdot \frac{{n-p \choose k }}{{n \choose k}} \cdot \frac{p}{n-k} \\
\label{eq:diff3} & = \sum_{k=1}^{n-p} k\cdot\frac{p}{n} \cdot {\frac{(n-p)!}{(n-p-k)!}}\cdot {\frac{(n-k-1)!}{(n-1)!}}\\
\label{eq:diff4} & = \sum_{k=1}^{n-p} k\cdot\frac{p}{n} \prod_{i=0}^{k-1}\underbrace{\frac{n-p-i}{n-1-i}}_{\leq \frac{n-p}{n-1}}\\
\label{eq:diff5} & \leq  \frac{p}{n} \sum_{k=1}^{n-p} k\cdot \left(1-\frac{p-1}{n-1}\right)^k \ .
\end{align}
Hence, by noting that $p=o(n)$ we obtain
\begin{align}\label{eq:diff6}
\lim_{n\rightarrow \infty} \frac{\bbe_\ell\{\tau_{\rm c}^B\}}{\frac{n}{p}}  & \leq \lim_{n\rightarrow \infty} \left(\frac{p}{n}\right)^2 \sum_{k=1}^{n-p} k\cdot \left(1-\frac{p-1}{n-1}\right)^k  = \lim_{n\rightarrow \infty} \left(\frac{p}{n}\right)^2 \left(\frac{n-1}{p-1}\right)^2 = 1\ .
\end{align}
For the lower bound, from \eqref{eq:diff4}  we have
\begin{align}
\bbe\{\tau_d\} & = \sum_{k=1}^{n-p} k\cdot\frac{p}{n} \prod_{i=0}^{k-1}\frac{n-p-i}{n-1-i}\\
& \geq  \sum_{k=1}^{\lfloor\frac{n-p}{2}\rfloor} k\cdot\frac{p}{n} \prod_{i=0}^{k-1} \frac{n-p-i}{n-1-i}\\
& \geq  \sum_{k=1}^{\lfloor\frac{n-p}{2}\rfloor} k\cdot\frac{p}{n} \prod_{i=0}^{k-1} \frac{n-p-{\lfloor\frac{n-p}{2}\rfloor}}{n-1-{\lfloor\frac{n-p}{2}\rfloor}}\\
& \geq  \sum_{k=1}^{\lfloor\frac{n-p}{2}\rfloor} k\cdot\frac{p}{n} \prod_{i=0}^{k-1} \frac{\frac{n-p}{2}}{\frac{n+p}{2}}\\
& = \frac{p}{n}  \sum_{k=1}^{\lfloor\frac{n-p}{2}\rfloor} k\left(\frac{n-p}{n+p}\right)^k\\
& = \frac{p}{n}  \sum_{k=1}^{\lfloor\frac{n-p}{2}\rfloor} k\left(1-\frac{2p}{n+p}\right)^k\ .
\end{align}
Hence, by noting that $p=o(n)$ we obtain
\begin{align}
\lim_{n\rightarrow \infty} \frac{\bbe_\ell\{\tau_{\rm c}^B\}}{\frac{n}{p}}  & \geq \lim_{n\rightarrow \infty} \left(\frac{p}{n}\right)^2\sum_{k=1}^{\lfloor\frac{n-p}{2}\rfloor} k\left(1-\frac{2p}{n+p}\right)^k\ \\
& = \lim_{n\rightarrow \infty} \left(\frac{p}{n}\right)^2 \left(\frac{n+p}{2p}\right)^2\\
\label{eq:diff7}  & = \frac{1}{4}\ .
\end{align}
Hence, from \eqref{eq:diff6} and \eqref{eq:diff7} we have
\begin{align}
\bbe_\ell\{\tau_{\rm c}^B\} =\Theta\left(\frac{n}{p}\right)\ ,
\end{align}
which completes the proof.
}

\bibliographystyle{unsrt}

\bibliography{CS_20200720}

\end{document}